\documentclass[reqno,11pt,nosumlimits]{article}
\usepackage{amssymb,amsfonts,amscd}
\usepackage[mathscr]{eucal}
\usepackage{amsmath}
\usepackage{graphics}
\usepackage{color}

\usepackage{ulem}
\usepackage{graphicx}

\allowdisplaybreaks[1]
\numberwithin{equation}{section}

\newtheorem{theorem}{Theorem}[section]

\newtheorem{proposition}[theorem]{Proposition}

\newtheorem{lemma}[theorem]{Lemma}
\newtheorem{conjecture}{Conjecture}

\newtheorem{definition}[theorem]{Definition}

\newtheorem{remark}[theorem]{Remark}

\newtheorem{observation}[theorem]{Observation}
\newtheorem{convention}[theorem]{Convention}

\newtheorem{comment}{Comment}

\newenvironment{proof}{{\it Proof. }}{{ \vskip 0.1cm
 \hfill{$\square$}}  \vspace{0.0cm} \medskip}

 \DeclareMathOperator{\Map}{Map}
 \DeclareMathOperator{\writhe}{writhe}

\DeclareMathOperator{\start}{\bullet}

\DeclareMathOperator{\Image}{Image}
\DeclareMathOperator{\arc}{arc}
 \DeclareMathOperator{\Tr}{Tr}
\DeclareMathOperator{\GL}{GL}
\DeclareMathOperator{\gl}{gl}

 \DeclareMathOperator{\Det}{Det}

 \DeclareMathOperator{\End}{End}
\DeclareMathOperator{\Hom}{Hom}
\DeclareMathOperator{\Ad}{Ad}
\DeclareMathOperator{\ad}{ad}

 \DeclareMathOperator{\WLO}{WLO}

\DeclareMathOperator{\wind}{wind}

\DeclareMathOperator{\sgn}{sgn}

\DeclareMathOperator{\aff}{aff}

\DeclareMathOperator{\Hol}{Hol}

\DeclareMathOperator{\mean}{mean}

\newcommand{\smooth}{(s)}

\newcommand{\be}{\begin{equation}}
\newcommand{\ee}{\end{equation}}

\newcommand{\bN}{{\mathbb N}}
\newcommand{\bQ}{{\mathbb Q}}
\newcommand{\bR}{{\mathbb R}}
\newcommand{\bC}{{\mathbb C}}
\newcommand{\bZ}{{\mathbb Z}}

\newcommand{\ct}{{\mathfrak t}}
\newcommand{\cG}{{\mathfrak g}}

\newcommand{\ck}{{\mathfrak k}}
\newcommand{\cg}{c_{\cG}}

\newcommand{\cA}{{\mathcal A}}
\newcommand{\cB}{{\mathcal B}}
\newcommand{\cC}{{\mathcal C}}

\newcommand{\cK}{{\mathcal K}}

\newcommand{\cP}{{\mathcal P}}

\newcommand{\cR}{{\mathcal R}}

\newcommand{\cT}{{\mathcal T}}

\newcommand{\cW}{{\mathcal W}}

\newcommand{\CW}{{\mathcal C}}

\newcommand{\face}{\mathfrak F}

\newcommand{\reg}{reg}
\newcommand{\orth}{\perp}

\newcommand{\eps}{\epsilon}

\newcommand{\id}{{ \rm id}}

\begin{document}

\title{Torus Knots and the Chern-Simons path integral: a rigorous treatment}

% \date{\today}

\maketitle

\begin{center} \large
Atle Hahn
\end{center}

\begin{center}
\it   \large  Group of Mathematical Physics, University of Lisbon \\
% Grupo de  F{\'i}sica Matem{\'a}tica da Universidade de Lisboa \\
Av. Prof. Gama Pinto, 2\\
PT-1649-003 Lisboa, Portugal\\
Email: atle.hahn@gmx.de
  \end{center}

\begin{abstract}
In 1993  Rosso and Jones computed for every simple, complex Lie algebra $\cG_{\bC}$
and every colored torus knot in $S^3$
 the value of the corresponding $U_q(\cG_{\bC})$-quantum invariant
 by  using the machinery of quantum groups.
In the present paper we derive a $S^2 \times S^1$-analogue of the Rosso-Jones formula
 (for colored torus ribbon knots)    directly  from a  rigorous realization of the corresponding
  (gauge fixed) Chern-Simons path integral.
 In order to compare the explicit expressions  obtained for torus knots in $S^2 \times S^1$
 with those for torus knots in $S^3$  one can perform  a suitable surgery operation.
By doing so we verify that  the original  Rosso-Jones formula is indeed recovered
for every $\cG_{\bC}$.
\end{abstract}

%\medskip

%{\em Key words:} Chern-Simons models,   Quantum invariants }

%\medskip

%{AMS subject classifications:}  57M27,   81T08,  81T45

\medskip

\section{Introduction}
\label{sec1}

Let $\cG_{\bC}$ be an arbitrary  simple complex Lie algebra, let
 $q \in \bC \backslash \{0\}$ be either generic or a root of unity of sufficiently high order,
   and let $U_q(\cG_{\bC})$ be the corresponding quantum group.
In \cite{RoJo} an explicit formula for the values
of the  $U_q(\cG_{\bC})$-quantum invariant of an arbitrary colored torus knot  in $S^3$
was found and proven using the representation theory of $U_q(\cG_{\bC})$.\par

In the special case where $q$ is a root of unity
the quantum invariants studied in \cite{RoJo} are  normalized
 versions of the Reshetikhin-Turaev invariants associated to $M = S^3$ and $U_q(\cG_{\bC})$
(cf. Eq. \eqref{eq_QI_RT} in the Appendix).
It is widely believed that the Reshetikhin-Turaev invariants associated to
a closed oriented 3-manifold $M$ and to $U_q(\cG_{\bC})$
are equivalent  to  Witten's heuristic path integral expressions based on the Chern-Simons action function
    associated to $(M,G,k)$  where
 $G$ is the simply connected, compact Lie group corresponding to the compact real form
  $\cG$ of $\cG_{\bC}$   and $k \in \bN$ is chosen suitably (cf. Remark \ref{rm_shift_in_level}  below).
 Accordingly, it is natural to  ask whether it is possible to derive the Rosso-Jones formula
 (or analogues/generalizations for base manifolds $M$ other than $S^3$)
 directly from Witten's path integral expressions. \par

 In the present paper I will show how one can do this
  for a large class of colored torus (ribbon) knots $L$ in the manifold
  $M=S^2 \times S^1$ in a rigorous  way.
  The approach of the present paper is based on the so-called torus gauge fixing procedure
which was introduced in \cite{BlTh1,BlTh3} for the study of the CS path integral on manifolds of the form
 $M=\Sigma \times S^1$.  In \cite{Ha3c,Ha4}  the basic heuristic formula of \cite{BlTh1}
 was generalized to general colored links $L$ in $M$.
 The generalized formula in \cite{Ha3c,Ha4} was recently simplified in \cite{Ha7a}, cf. the heuristic equation \eqref{eq2.48_Ha7a} below,  which will be the starting point for the rigorous treatment of the present paper.\par

In order to make rigorous sense of the RHS of the aforementioned Eq. \eqref{eq2.48_Ha7a}
we will  work within the simplicial setting developed in \cite{Ha7a}.
The simplicial setting not only allows a completely rigorous treatment but also one that is essentially elementary:
apart from some\footnote{In fact, even most of the Lie theoretic results appear only
after the path integral expressions have already been evaluated explicitly (cf. Steps 1--4 in the proof
of Theorem \ref{theorem1}) and we compare the explicit expressions with
those in the Rosso-Jones formuala (cf. Step 5 in the proof
of Theorem \ref{theorem1} and Sec. \ref{subsec6.2})}
 basic results from general Lie theory only a few quite
simple results on  oscillatory Gauss-type integrals on Euclidean vector spaces
will be needed, cf. Sec. \ref{sec4} below.\par

The paper is organized as follows: \par

In Sec. \ref{sec2} we first recall the aforementioned heuristic formula Eq. \eqref{eq2.48_Ha7a}
for the CS path integral in the torus gauge and later
give a ribbon version of  Eq. \eqref{eq2.48_Ha7a}, cf. Eq. \eqref{eq2.48} below. \par

In Sec. \ref{sec3} we introduce a (rigorous) simplicial realization $\WLO^{disc}_{rig}(L)$
of the RHS of the heuristic formula Eq. \eqref{eq2.48} in Sec. \ref{sec2}  for generic colored ribbon links $L$.
 The definition of $\WLO^{disc}_{rig}(L)$ is similar to the one in \cite{Ha7a}
  but incorporates some improvements  and simplifications.   \par

In Sec. \ref{sec4}  we recall the relevant results from \cite{Ha7b}
on  oscillatory Gauss-type integrals on Euclidean vector spaces
which we will use in Sec. \ref{sec5}.\par

In Sec. \ref{sec5} we compute $\WLO^{disc}_{rig}(L)$ (and the normalized version $\WLO^{disc}_{norm}(L)$)
 explicitly for a large class of colored torus ribbon knots $L$ in $S^2 \times S^1$,
see  Theorem \ref{theorem1} and its proof.
Apart from   Theorem \ref{theorem1}  a straightforward
generalization  is proven, cf. Theorem \ref{theorem2}.\par

In Sec. \ref{sec6} we   combine Theorem \ref{theorem2} with a suitable surgery argument.
The explicit expressions obtained in this way are then compared with those in the   Rosso-Jones
formula for colored torus  knots in $S^3$.
We find agreement for all $\cG_{\bC}$.   \par

The paper  concludes with Sec. \ref{sec7} and  a short appendix.

\begin{comment} \label{comm1} \rm The present paper pursues two closely related but still quite different goals:
\begin{description}
\item[\it Goal 1 (= Main Goal):] Make progress with the simplicial program for Chern-Simons theory, cf. Sec. 3 in \cite{Ha7a} and Remark \ref{rm_sec3.9b} below.   \par

     Goal 1 is achieved by  Theorem \ref{theorem1},  Theorem \ref{theorem2},  and the partial verification of Conjecture~\ref{conj0}  given in Sec. \ref{subsec6.2}.

\item[\it Goal 2:]  Give a (new) heuristic derivation of the original Rosso-Jones formula
for general $\cG_{\bC}$.  \par

 Goal 2 is achieved by  combining Theorem \ref{theorem2} (and  Remark \ref{rm_sec3.9c}) below
 with the modification of   Sec. \ref{subsec6.2} which is obtained by rewriting Sec. \ref{subsec6.2}
  using Witten's heuristic surgery argument instead of the (rigorous) surgery argument for
  the Reshetikhin-Tureav invariant,  cf. Remark \ref{rm_Goal1_vs_Goal2} below.
 [Of course, if we were only interested in Goal 2 and not in Goal 1
 we could significantly reduce the amount of work and simply state and ``prove'' a heuristic continuum version
 of Theorem \ref{theorem2}. This would take only a few pages.
 In particular, Sec. \ref{sec3} could then be omitted.]
 \end{description}

  Regarding Goal 2  it  should be noted
  that there are already several quite general heuristic approaches for calculating the CS path integral expressions
  for a large class of knots \& links and base manifolds $M$, cf.  the perturbative approach in \cite{GMM,BN,Kon,AxSi1,AxSi2,BoTa} based on Lorentz gauge fixing 
  and  the approach in \cite{BeaWi,Bea} which is based on non-Abelian localization. 
     It  is expected  that  in the special case of  torus knots in $M=S^3$
    the   approach  in \cite{BeaWi,Bea} leads  to the Rosso-Jones formula
      but to my knowledge this has so far only been shown explicitly in the special case $\cG_{\bC} = sl(2,\bC)$. \par

 Apart from these   (heuristic) path integral approaches
  one should also mention the  approach in  \cite{LabLlaRam,IsLabRam,LabMa1,LabMa2,LabPer,Ste}
  where Witten's CS path integral expressions are evaluated for torus knots \& links
   using the heuristic ``knot operator'' approach introduced in \cite{LabRam}.
   The knot operator approach allows the derivation of the  Rosso-Jones formula
        for arbitrary simple complex Lie algebras but
         involves  only few genuine path integral arguments
        (i.e. arguments which deal directly/explicitly with the CS path integral).
      Instead, a variety of  several quite different heuristic  arguments, some of them from  Conformal Field Theory,   are used in the knot operator approach.
\end{comment}

\section{The heuristic Chern-Simons  path integral  in the torus gauge}
\label{sec2}

Let $G$ be a  simple,  simply-connected, compact Lie group
  and $T$ a maximal torus  of $G$.
By  $\cG$ and $\ct$ we  will denote the Lie algebras of $G$ and  $T$
and by  $\langle \cdot , \cdot \rangle$
 the unique $\Ad$-invariant scalar product
on $\cG$ such that $\langle \Check{\alpha} , \Check{\alpha} \rangle = 2$
for every short coroot $\Check{\alpha} \in \ct$.\par

Let $M$ be a compact oriented 3-manifold
of the form $M = \Sigma \times S^1$ where $\Sigma$ is a compact oriented surface.
(From Sec. \ref{sec5} on we will only consider the special case $\Sigma = S^2$.)
Finally, let $L$ be a  fixed (ordered and oriented) ``link''  in $M$,
i.e. a finite tuple $(l_1, \ldots, l_m)$, $m \in \bN$, of pairwise non-intersecting
knots $l_i$ . We equip  each $l_i$ with a ``color'', i.e. an irreducible,
 finite-dimensional, complex representation $\rho_i$ of $G$.
Recall that a ``knot'' in $M$ is an embedding $l:S^1 \to M$.
Using the surjection $[0,1] \ni t \mapsto e^{2 \pi i t} \in \{ z \in \bC \mid |z| =1 \}  \cong S^1$
we can  consider each  knot  as a loop $l:[0,1] \to M$, $l(0) = l(1)$, in the obvious way.

\subsection{Basic  spaces}
\label{subsec2.1}

As in  \cite{Ha7a,Ha7b} we will use the following notation\footnote{Here  $\Omega^p(N,V)$ denotes the space of $V$-valued  $p$-forms
on a smooth manifold $N$}
\begin{subequations} \label{eq_basic_spaces_cont}
\begin{align}
\cB & = C^{\infty}(\Sigma,\ct)  \\
\cA & =  \Omega^1(M,\cG)\\
\cA_{\Sigma} & =  \Omega^1(\Sigma,\cG) \\
\cA_{\Sigma,\ct} & = \Omega^1(\Sigma,\ct), \quad  \cA_{\Sigma,\ck}  = \Omega^1(\Sigma,\ck) \\
\cA^{\orth} & =  \{ A \in \cA \mid A(\partial/\partial t) = 0\}\\
\label{eq_part_f}
\Check{\cA}^{\orth} & = \{ A^{\orth} \in \cA^{\orth} \mid \int A^{\orth}(t) dt \in \cA_{\Sigma,\ck} \} \\
\label{eq_part_g}   \cA^{\orth}_c & = \{ A^{\orth} \in \cA^{\orth} \mid \text{ $A^{\orth}$ is constant and
 $\cA_{\Sigma,\ct}$-valued}\}
 \end{align}
\end{subequations}
where $\ck$ is the orthogonal complement of $\ct$ in $\cG$ w.r.t.
$\langle \cdot, \cdot \rangle$.
Above $dt$ denotes the normalized translation-invariant volume form on $S^1$
and $\partial/\partial t$  the vector field on $M=\Sigma \times S^1$
obtained by ``lifting''  in the obvious way
the normalized translation-invariant vector field $\partial/\partial t$ on $S^1$.
In Eqs. \eqref{eq_part_f} and \eqref{eq_part_g}
 we used the ``obvious'' identification (cf. Sec. 2.3.1 in \cite{Ha7a})
\begin{equation}
\cA^{\orth}  \cong C^{\infty}(S^1,\cA_{\Sigma})
\end{equation}
where  $C^{\infty}(S^1,\cA_{\Sigma})$ is the space of maps
$f:S^1 \to \cA_{\Sigma}$ which are ``smooth'' in the sense that
$\Sigma \times S^1 \ni (\sigma,t) \mapsto (f(t))(X_{\sigma}) \in \cG$
is smooth for every smooth vector field $X$ on $\Sigma$.
Note that we have
\begin{equation} \label{eq_cAorth_decomp}
\cA^{\orth} = \Check{\cA}^{\orth} \oplus  \cA^{\orth}_c
\end{equation}

\subsection{The original Chern-Simons path integral}
\label{subsec2.1b}

The Chern-Simons action function $S_{CS}: \cA \to \bR$
associated to $M$, $G$, and the ``level''\footnote{cf. Remark \ref{rm_shift_in_level}  below}
  $k \in \bN$ is given by
 \begin{equation} \label{eq2.2'} S_{CS}(A) = - k \pi \int_M \langle A \wedge dA \rangle
   + \tfrac{1}{3} \langle A\wedge [A \wedge A]\rangle, \quad
A \in \cA \end{equation}
Here $[\cdot \wedge \cdot]$  denotes the wedge  product associated to the
Lie bracket $[\cdot,\cdot] : \cG \times \cG \to \cG$
and    $\langle \cdot \wedge  \cdot \rangle$
   the wedge product  associated to the
 scalar product $\langle \cdot , \cdot \rangle : \cG \times \cG \to \bR$.\par

The (expectation value of the)  ``Wilson loop observable'' associated to the colored link $L=(l_1,l_2,\ldots,l_m)$
fixed above  is the informal integral expression given by
\begin{equation} \label{eq_WLO_orig}
\WLO(L) := \int_{\cA} \left( \prod_{i=1}^m  \Tr_{\rho_i}(\Hol_{l_i}(A)) \right) \exp( i S_{CS}(A)) DA
\end{equation}
where $\Hol_{l}(A) \in G$ is the holonomy of $A \in \cA$ around the loop $l = l_i$, $i \le m$,
and $DA$ is the (ill-defined) ``Lebesgue measure'' on the infinite-dimensional space $\cA$.
A useful explicit formula for  $\Hol_l(A)$ is
\begin{equation} \label{eq_Hol_heurist}
\Hol_l(A) = \lim_{n \to \infty} \prod_{j=1}^n \exp\bigl(\tfrac{1}{n}  A(l'(t))\bigr)_{| t=j/n}
\end{equation}
where $\exp:\cG \to G$ is the exponential map of $G$.

\begin{remark} In the physics literature the notation $Z(M,L)$ and $P \exp(\int_l A)$ is usually
used instead of $\WLO(L)$ and $\Hol_{l}(A)$.
\end{remark}

\subsection{The torus gauge fixed Chern-Simons path integral}
\label{subsec2.2}

Let $\pi_{\Sigma}: \Sigma \times S^1 \to \Sigma$ be the canonical projection.
For each loop $l_i$ appearing in the link $L$ we set $l^i_{\Sigma}:= \pi_{\Sigma} \circ l_i$.
Moreover, we fix  $\sigma_0 \in \Sigma$ such that
$$\sigma_0 \notin  \bigcup_i \arc(l^i_{\Sigma})$$
By applying ``abstract torus gauge fixing'' (cf. Sec. 2.2.4 in  \cite{Ha7a})
and  suitable change of variable (cf. Sec. 2.3.1 and Appendix B.3 in \cite{Ha7a})
 one can derive at a heuristic level  (cf. Eq. (2.53) in \cite{Ha7a})
\begin{multline}  \label{eq2.48_Ha7a} \WLO(L)
 \sim \sum_{y \in I}  \int_{\cA^{\orth}_c \times \cB} \biggl\{
 1_{C^{\infty}(\Sigma,\ct_{reg})}(B)  \Det_{FP}(B)\\
 \times   \biggl[ \int_{\Check{\cA}^{\orth}} \biggl( \prod_i  \Tr_{\rho_i}\bigl(
 \Hol_{l_i}(\Check{A}^{\orth} + A^{\orth}_c, B)\bigr) \biggr)
\exp(i  S_{CS}( \Check{A}^{\orth}, B)) D\Check{A}^{\orth} \biggr] \\
 \times \exp\bigl( - 2\pi i k  \langle y, B(\sigma_0) \rangle \bigr) \biggr\}
 \exp(i S_{CS}(A^{\orth}_c, B)) (DA^{\orth}_c \otimes DB)
\end{multline}
where ``$\sim$'' denotes equality up to a multiplicative ``constant''\footnote{``constant'' in the sense
 that  $C$ does not depend on $L$.
 By contrast, $C$ may  depend on $G$, $\Sigma$, and $k$.} $C$,
where $I:= \ker(\exp_{| \ct})  \subset \ct$, where $DB$ and $DA^{\orth}_c$ are the
 informal ``Lebesgue measures'' on the infinite-dimensional spaces $\cB$ and $\cA^{\orth}_c$,
 and where we have set $\ct_{reg} := \exp^{-1}(T_{reg})$,
  $T_{reg}$ being the set of ``regular'' elements\footnote{i.e. the set of all $t \in T$ which are not contained
in a different maximal torus $T'$} of $T$.
Moreover, we have set for each $B \in \cB$, $A^{\orth} \in \cA^{\orth}$
\begin{align}
S_{CS}(A^{\orth},B) & := S_{CS}(A^{\orth} + B dt),  \\
 \label{eq4.17}\Hol_{l}(A^{\orth},  B) &  := \Hol_{l}(A^{\orth}    + B dt) \nonumber \\
 & = \lim_{n \to \infty} \prod_{j = 1}^n
 \exp\bigl( \tfrac{1}{n} [ A^{\orth}(l_{S^1}(t))(l'_{\Sigma}(t)) +
 B(l_{\Sigma}(t)) \cdot dt(l'_{S^1}(t))] \bigr)_{t = j/n}
\end{align}
where $dt$ is the real-valued 1-form on  $M=\Sigma \times S^1$
obtained by  pulling back the 1-form $dt$ on $S^1$ by means of the canonical projection
$\pi_{S^1}: \Sigma \times S^1 \to S^1$ and where
$l_{S^1}: [0,1] \to S^1$ and  $l_{\Sigma}: [0,1] \to \Sigma$
are the projected loops given by
 $l_{S^1} := \pi_{S^1} \circ l$ and $l_{\Sigma} := \pi_{\Sigma} \circ l$.\par

Finally, the expression $\Det_{FP}(B)$ in Eq. \eqref{eq2.48_Ha7a} is the informal expression   given by
\begin{equation} \label{eq_DetFP} \Det_{FP}(B) :=    \det\bigl(1_{\ck}-\exp(\ad(B))_{|\ck}\bigr)
\end{equation}
where $1_{\ck}-\exp(\ad(B))_{|\ck} $
 is the linear operator on $C^{\infty}(\Sigma,\ck)$ given by
\begin{equation}
(1_{\ck}-\exp(\ad(B))_{|\ck} \cdot f)(\sigma) = (1_{\ck}-\exp(\ad(B(\sigma)))_{|\ck}) \cdot f(\sigma)
\quad \quad \forall \sigma \in \Sigma, \quad \forall f \in   C^{\infty}(\Sigma,\ck)
\end{equation}
where on the RHS $1_{\ck}$ is the identity on $\ck$.

\smallskip

It will be convenient to generalize the definition of $1_{\ck}-\exp(\ad(B))_{|\ck}$ above.
For every $p \in \{0,1,2\}$ we define
$(1_{\ck}-\exp(\ad(B))_{|\ck})^{(p)}$  to be the linear operator on $\Omega^p(\Sigma,\ck)$ given by
\begin{multline} \label{eq_def_det(p)}
\forall  \alpha \in  \Omega^p(\Sigma,\ck):  \forall  \sigma \in \Sigma:  \forall X_{\sigma} \in \wedge^p
 T_{\sigma} \Sigma: \\
 \quad \quad \bigl( \bigl(1_{\ck}-\exp(\ad(B))_{|\ck}\bigr)^{(p)} \cdot \alpha\bigr)(X_{\sigma}) = (1_{\ck}-\exp(\ad(B(\sigma))_{|\ck}) \cdot \alpha(X_{\sigma})
\end{multline}
 Note that under the identification $C^{\infty}(\Sigma,\ck) \cong \Omega^0(\Sigma,\ck)$
the operator  $(1_{\ck}-\exp(\ad(B))_{|\ck})^{(0)}$ coincides with what above we call $1_{\ck}-\exp(\ad(B))_{|\ck}$.\par

As in \cite{Ha7a} we will now fix
an  auxiliary Riemannian metric ${\mathbf g_{\Sigma}}$ on $\Sigma$.
Let    $\ll \cdot , \cdot \gg_{\cA_{\Sigma}}$ and
    $\ll \cdot , \cdot \gg_{\cA^{\orth}}$  be the scalar products on $\cA_{\Sigma}$ and
  $\cA^{\orth} \cong C^{\infty}(S^1, \cA_{\Sigma})$   induced by ${\mathbf g_{\Sigma}}$,
   and  let $\star: \cA_{\Sigma} \to \cA_{\Sigma}$ be the corresponding Hodge star operator.
  By $\star$ we will also denote the linear automorphism
 $\star: C^{\infty}(S^1, \cA_{\Sigma}) \to C^{\infty}(S^1, \cA_{\Sigma})$ given
 by $(\star A^{\orth})(t) = \star (A^{\orth}(t))$ for all $A^{\orth} \in \cA^{\orth}$ and $t \in S^1$.
 We then have (cf. Eq. (2.48) in \cite{Ha7a})
 \begin{equation} \label{eq_SCS_expl0} S_{CS}(A^{\orth},B)  =  \pi k  \ll A^{\orth},
\star  \bigl(\tfrac{\partial}{\partial t} + \ad(B) \bigr) A^{\orth} \gg_{\cA^{\orth}}
 +  2 \pi k  \ll\star  A^{\orth},  dB \gg_{\cA^{\orth}}
\end{equation}
for all $B \in \cB$ and $A^{\orth} \in \cA^{\orth}$,
 and in particular,
 \begin{align} \label{eq_SCS_expl}
S_{CS}(\Check{A}^{\orth},B) & =  \pi k  \ll \Check{A}^{\orth},
\star  \bigl(\tfrac{\partial}{\partial t} + \ad(B) \bigr) \Check{A}^{\orth} \gg_{\cA^{\orth}} \\
\label{eq_SCS_expl2}
 S_{CS}(A^{\orth}_c,B) & =   2 \pi k  \ll\star  A^{\orth}_c,  dB \gg_{\cA^{\orth}}
\end{align}
for  $B \in \cB$, $\Check{A}^{\orth} \in \Check{\cA}^{\orth}$, and $A^{\orth}_c \in \cA^{\orth}_c$.

\subsection{Ribbon version  of Eq. \eqref{eq2.48_Ha7a}}
\label{subsec2.4}

Recall that our goal is to find a rigorous realization of Witten's CS path integral expressions
 which reproduces the Reshetikhin-Turaev invariants (in the special situation described in the Introduction).
 Since the Reshetikhin-Turaev invariants are defined for ribbon links (or, equivalently\footnote{From the knot theory point of view   the framed link picture and the ribbon link picture are equivalent.
 However,  the ribbon picture seems to be better suited for the
 study of the  Chern-Simons path integral in the torus gauge},
  for framed links) we will now write down a ribbon analogue of Eq. \eqref{eq2.48_Ha7a}. \par

 A closed ribbon $R$ in $\Sigma \times S^1$ is a  smooth embedding
$R: S^1 \times [0,1] \to \Sigma \times S^1$.
A ribbon link in $\Sigma \times S^1$ is
 a finite tuple of non-intersecting closed ribbons in in $\Sigma \times S^1$.
We will replace the link $L =(l_1,l_2, \ldots, l_m)$
by a ribbon link $L_{ribb} = (R_1,R_2, \ldots, R_m)$ where each $R_i$, $i \le m$,
is chosen such that  $l_i(t)= R_i(t,1/2)$ for all $t \in S^1$. Instead of $L_{ribb}$ we will
simply write $L$ in the following.
From now on we will assume that  $\sigma_0 \in \Sigma$ was chosen such that
$$\sigma_0 \notin  \bigcup_i \Image(R^i_{\Sigma})$$
where $R^i_{\Sigma}:= \pi_{\Sigma} \circ R_i$.
For every $R \in \{R_1, R_2, \ldots, R_m\}$
we define
$$\Hol_{R}(A) :=  \lim_{n \to \infty} \prod_{j=1}^n \exp\bigl(\tfrac{1}{n} \int_{0}^1  A(l'_u(t)) du \bigr)_{| t=j/n}
 \in G$$
 where $l_u$, $u \in [0,1]$, is the knot
 $l_u := R(\cdot,u)$, considered  as a loop  $[0,1] \to \Sigma \times S^1$.
Moreover, for $A^{\orth} \in \cA^{\orth}$ and $B \in \cB$ we set
\begin{multline} \label{eq_sec3.3_1}
\Hol_{R}(A^{\orth}, B)  :=  \Hol_{R}(A^{\orth} +B dt)  \\
  = \lim_{n \to \infty} \prod_{j=1}^n \exp\bigl(\tfrac{1}{n} \int_{0}^1
   [ A^{\orth}(l^u_{S^1}(t))((l^u_{\Sigma})'(t)) +
 B(l^u_{\Sigma}(t)) \cdot dt((l^u_{S^1})'(t))] du\bigr)_{| t = j/n}
   \end{multline}
where  $l^u_{S^1} := \pi_{S^1} \circ l_u$ and $l^u_{\Sigma} := \pi_{\Sigma} \circ l_u$ for each $u \in [0,1]$.\par

We now obtain the aforementioned ribbon analogue of Eq. \eqref{eq2.48_Ha7a}
 by replacing   the expression $\Hol_{l_i}(\Check{A}^{\orth} + A^{\orth}_c,B)$ in Eq. \eqref{eq2.48_Ha7a}
 with $\Hol_{R_i}(\Check{A}^{\orth} + A^{\orth}_c, B)$:
\begin{multline}  \label{eq2.48} \WLO(L)
 \sim \sum_{y \in I}  \int_{\cA^{\orth}_c \times \cB} \biggl\{
 1_{C^{\infty}(\Sigma,\ct_{reg})}(B)  \Det(B)\\
 \times   \biggl[ \int_{\Check{\cA}^{\orth}} \left( \prod_i  \Tr_{\rho_i}\bigl(
 \Hol_{R_i}(\Check{A}^{\orth} + A^{\orth}_c, B)\bigr) \right)
d\mu^{\orth}_B(\Check{A}^{\orth}) \biggr] \\
 \times \exp\bigl( - 2\pi i k  \langle y, B(\sigma_0) \rangle \bigr) \biggr\}
 \exp(i S_{CS}(A^{\orth}_c, B)) (DA^{\orth}_c \otimes DB)
\end{multline}
where, as a preparation for Sec. \ref{subsec2.5}, we have set, for each $B \in \cB$,
\begin{align}  \label{eq_def_Z_B}
\Check{Z}(B) & := \int \exp(i  S_{CS}( \Check{A}^{\orth}, B)) D\Check{A}^{\orth},\\
 \label{eq_def_mu_B}
d\mu^{\orth}_B & := \tfrac{1}{\Check{Z}(B)} \exp(i  S_{CS}( \Check{A}^{\orth}, B)) D\Check{A}^{\orth}
\end{align}
and
\begin{equation} \label{eq_def_det(B)}
\Det(B) := \Det_{FP}(B) \Check{Z}(B)
\end{equation}

\subsection{Rewriting $\Det(B)$}
\label{subsec2.5}

Informally, we have for $B \in \cB_{reg}:= C^{\infty}(\Sigma,\ct_{reg})$
\begin{equation} \label{eq_sec2.5_1}
\Check{Z}(B) \sim \det\bigl(\tfrac{\partial}{\partial t} + \ad(B)\bigr)^{ -1/2} \overset{(*)}{\sim}
 \det\bigl(\bigl(1_{\ck}-\exp(\ad(B))_{|\ck}\bigr)^{(1)}\bigr)^{-1/2}
 \end{equation}
  where $\tfrac{\partial}{\partial t} + \ad(B):\Check{\cA}^{\orth} \to \Check{\cA}^{\orth}$
  is the operator appearing  in Eq. \eqref{eq_SCS_expl} above,
  and where  $(1_{\ck}-\exp(\ad(B))_{|\ck})^{(1)}$  is the linear operator on $\cA_{\Sigma,\ck}=
 \Omega^1(\Sigma,\ck)$ given by Eq. \eqref{eq_def_det(p)} above with $p=1$.
 Here  step $(*)$ is suggested by
 $$\det(\tfrac{\partial}{\partial t} + \ad(b)\bigr) \sim \det\bigl(\bigl(1_{\ck}-\exp(\ad(b))_{|\ck}\bigr)\bigr) \quad \forall b \in \ct_{reg}$$
 where   $\tfrac{\partial}{\partial t} + \ad(b): C^{\infty}(S^1,\ck) \to  C^{\infty}(S^1,\ck)$
 and where $\det(\tfrac{\partial}{\partial t} + \ad(b))$ is
 defined  with the help of a standard  $\zeta$-function regularization argument.   \par

 Observe also that
 $\bigl(1_{\ck}-\exp(\ad(B))_{|\ck}\bigr)^{(0)} = \star^{-1} \circ \bigl(1_{\ck}-\exp(\ad(B))_{|\ck}\bigr)^{(2)} \circ \star$  where $(1_{\ck}-\exp(\ad(B))_{|\ck})^{(2)}$  is the linear operator on
$\Omega^2(\Sigma,\ck)$ given by Eq. \eqref{eq_def_det(p)} above with $p=2$
and where  $\star: \Omega^0(\Sigma,\ck) \to \Omega^2(\Sigma,\ck)$ is the  Hodge star operator induced by ${\mathbf g_{\Sigma}}$.
Thus we obtain,  informally,
\begin{equation} \label{eq_sec2.5_2}
 \det\bigl(\bigl(1_{\ck}-\exp(\ad(B))_{|\ck}\bigr)^{(0)}\bigr) = \det\bigl(\bigl(1_{\ck}-\exp(\ad(B))_{|\ck}\bigr)^{(2)}\bigr)
\end{equation}
Combining Eq. \eqref{eq_def_det(B)}, Eq. \eqref{eq_sec2.5_1}, and Eq. \eqref{eq_sec2.5_2} we obtain
\begin{equation} \label{eq_Det(B)_rewrite}
\Det(B)  = \prod_{p=0}^2  \det\bigl(\bigl(1_{\ck}-\exp(\ad(B))_{|\ck}\bigr)^{(p)}\bigr)^{(-1)^p /2}
\end{equation}

\section{Simplicial realization  $\WLO^{disc}_{rig}(L)$ of $\WLO(L)$}
\label{sec3}

\subsection{Some polyhedral cell complexes}

Let $\cP$ be a finite oriented polyhedral cell complex (cf. Appendix C in \cite{Ha7a}).

\begin{itemize}
\item We denote by $\face_p(\cP)$, $ p \in \bN_0$, the set of $p$-faces of $\cP$.
The elements of $\face_0(\cP)$ ($\face_1(\cP)$, respectively)
  will be called the ``vertices'' (``edges'', respectively) of $\cP$.

\item For  every fixed real vector space $V$
we denote by  $C^p(\cP,V)$, $ p \in \bN_0$, the space of maps $\face_p(\cP) \to V$  (``$V$-valued
 $p$-cochains of $\cP$'').  Instead of  $C^p(\cP,\bR)$ we will often write $C_p(\cP)$.

\item We identify $\face_p(\cP)$, $ p \in \bN_0$, with a subset of $C_p(\cP) = C^p(\cP,\bR)$
     in the obvious way, i.e. each $\alpha \in \face_p(\cP)$ is identified with
     $\delta_{\alpha} \in C^p(\cP,\bR)$ given by $\delta_{\alpha}(\beta) = \delta_{\alpha\beta}$
     for all $\beta \in \face_p(\cP)$.

\item   By $d_{\cP}$, $ p \in \bN_0$, we will denote the  coboundary operator\footnote{in the special case
$p=0$, which is the only case relevant for us, $d_{\cP}: C^0(\cP,V) \to C^{1}(\cP,V)$
 is given explicitly by $d f(e) = f(end(e))-f(start(e))$ for all $f \in C^0(\cP,V)$ and $e \in \face_1(\cP)$
 where $start(e), end(e) \in \face_0(\cP)$ denote the starting/end point of the (oriented) edge $e$}
 $C^p(\cP,V) \to C^{p+1}(\cP,V)$.
\end{itemize}

i) As a discrete analogue of the Lie group $S^1$ we will
use  the finite cyclic group $\bZ_{N}$, $N \in \bN$.
 The number $N$ will be kept fixed throughout the rest of this paper.
 We will identify $\bZ_N$ with the subgroup $\{e^{2 \pi i k/N} \mid 1 \le k \le N\}$ of
 $S^1$.  The points of $\bZ_N$ induce a polyhedral cell decomposition of $S^1$. The (1-dimensional oriented\footnote{we equip each edge of $\bZ_N$ with the orientation
 induced by the orientation $dt$ of $S^1$})  polyhedral cell complex  obtained in this way
   will also be denoted by $\bZ_N$ in the following.\par

ii) We  fix a finite oriented smooth polyhedral cell decomposition
$\cC$  of $\Sigma$ .
By $\cC'$ we will denote the ``canonical dual'' of the  polyhedral cell decomposition  $\cC$ (cf. the end of
 Appendix C in \cite{Ha7a}),  again equipped with an orientation.
By $\cK$ and $\cK'$ we will denote the (oriented)
  polyhedral cell complexes associated to $\cC$ and $\cC'$,
  i.e. $\cK = (\Sigma,\cC)$ and $\cK' = (\Sigma,\cC')$.
Instead of $\cK$ and $\cK'$
 we often write $K_1$ and $K_2$   and  we set $K:= (K_1,K_2)$. \par

iii) We  introduce a joint subdivision $q\cK$ of $\cK$ and $\cK'$
which is uniquely determined by the conditions
\begin{align*} \face_0(q\cK) & = \face_0(b\cK), \\
\face_1(q\cK) & = \face_1(b\cK) \backslash \{e \in \face_1(b\cK) \mid \text{ both endpoints of $e$
lie in $\face_0(\cK) \cup  \face_0(\cK')$} \},
\end{align*}
 $b\cK$ being the barycentric subdivision of $\cK$ (cf. Sec. 4.4.3 in \cite{Ha7a} for more details).
We equip  the faces of $q\cK$ with an orientation.
 For convenience we choose the orientation on the edges of  $q\cK$  to be ``compatible''\footnote{more precisely,
 for  each $e \in \face_1(q\cK)$ we choose the orientation which is induced by orientation of the unique
 edge $e' \in \face_1(\cK) \cup \face_1(\cK')$ which contains $e$} with
 the orientation on the edges of $\cK$ and $\cK'$.\par

iv) By $\cK \times \bZ_N$ and  $q\cK \times \bZ_N$  we will denote the obvious product (polyhedral) cell complexes.

\subsection{The basic spaces}
\label{subsec3.0}

\subsubsection*{a) The spaces $\cB(q\cK)$, $\cA_{\Sigma}(q\cK)$, and $\cA^{\orth}(q\cK)$}
We first introduce the following simplicial analogues of the spaces $\cB$, $\cA_{\Sigma}$, and $\cA^{\orth}$
in Sec. \ref{subsec2.1} above:
\begin{subequations} \label{eq_basic_spaces}
\begin{align}
 \cB(qK) & :=C^0(q\cK,\ct) \\
 \cA_{\Sigma}(q\cK) & := C^1(q\cK,\cG) \\
\cA^{\orth}(q\cK) & := \Map(\bZ_N,\cA_{\Sigma}(q\cK))
 \end{align}
 \end{subequations}

\noindent
The scalar product $\langle \cdot, \cdot \rangle$ on $\cG$  induces
 scalar products $\ll \cdot, \cdot\gg_{\cB(q\cK)}$ and $\ll \cdot, \cdot\gg_{\cA_{\Sigma}(q\cK)}$
 on $\cB(q\cK)$ and $\cA_{\Sigma}(q\cK)$  in the standard way.
 We define a  scalar product
  $\ll \cdot , \cdot \gg_{\cA^{\orth}(q\cK)}$      on $\cA^{\orth}(q\cK) = \Map(\bZ_N,\cA_{\Sigma}(q\cK))$ by
   \begin{equation} \label{eq_norm_scalarprod}
\ll A^{\orth}_1 , A^{\orth}_2 \gg_{\cA^{\orth}(q\cK)}  =   \tfrac{1}{N} \sum_{t \in \bZ_N} \ll
A^{\orth}_1(t) , A^{\orth}_2(t) \gg_{\cA_{\Sigma}(q\cK)}
\end{equation} for all $A^{\orth}_1 , A^{\orth}_2 \in \cA^{\orth}(q\cK)$.

\subsubsection*{b) The subspaces $\cB(\cK)$, $\cA_{\Sigma}(K)$, and $\cA^{\orth}(K)$}

For technical reasons (cf. Remark \ref{rm_3.1}  below) we will now introduce suitable subspaces
of $\cB(q\cK)$, $\cA_{\Sigma}(q\cK)$, and $\cA^{\orth}(q\cK)$.
It will be convenient to first define these three spaces in an ``abstract'' way and
then to explain how they are embedded into the three aforementioned spaces. We set

\begin{subequations}
\begin{align}
\cB(\cK) & := C^0(\cK,\ct)\\
\cA_{\Sigma}(K) & :=C^1(K_1,\cG)  \oplus  C^1(K_2,\cG)\\
\cA^{\orth}(K) & := \Map(\bZ_N,\cA_{\Sigma}(K))
\end{align}
\end{subequations}

\begin{itemize}
\item We will  identify $\cA_{\Sigma}(K) \cong (C_1(K_1)  \oplus  C_1(K_2)) \otimes_{\bR} \cG$ with a linear subspace of $\cA_{\Sigma}(q\cK) \cong C_1(q\cK)  \otimes_{\bR} \cG$
by means of the linear injection $\psi \otimes \id_{\cG}$
where $\psi:C_1(K_1)  \oplus  C_1(K_2) \to  C_1(q\cK)$ is the linear injection given by
$\psi(e)=e_1 + e_2$  for all $e \in \face_1(K_1) \cup  \face_1(K_2)$
where $e_1, e_2 \in  \face_1(q\cK)$ are the two edges of $q\cK$
``contained'' in $e$.

\item Since $\cA_{\Sigma}(K)$ is now identified with a subspace of $\cA_{\Sigma}(q\cK) $
the space  $\cA^{\orth}(K)$ can be considered as a subspace of $\cA^{\orth}(q\cK)$ in the obvious way.

\item Finally, the space $\cB(\cK)$ will be identified with a subspace of $\cB(q\cK)$
via the linear injection   $\psi: \cB(\cK) \to \cB(q\cK)$
 which associates to each $B \in  \cB(\cK)$ the extension $\bar{B} \in \cB(q\cK)$
given by $\bar{B}(x) = \mean_{y \in S_x} B(y)$ for all $x \in \face_0(q\cK)$.
Here ``mean'' refers to the arithmetic mean and $S_x$ denotes the set of all $y \in \face_0(\cK)$ which
 lie in the closure of the unique open cell of $\cK$  containing $x$.
\end{itemize}

\begin{remark} \label{rm_3.1} \rm
 i) The reason for introducing the subspaces
$\cA_{\Sigma}(K)$ and $\cA^{\orth}(K)$ is that
these spaces allow us  to obtain a nice simplicial analogue of the Hodge star operator,
cf. Sec. \ref{subsec3.2} below.

\medskip

ii) In order to motivate the introduction of the subspace $\cB(\cK)$ of   $\cB(q\cK)$ we remark that
$\ker(\pi \circ d_{q\cK}) \neq \cB_c(q\cK)$ where
 $ \cB_c(q\cK):= \{B \in \cB(q\cK) \mid B \text{ constant}\}$ and where
\begin{equation} \label{eq_pi_proj} \pi: \cA_{\Sigma}(q\cK) \to  \cA_{\Sigma}(K)
\end{equation}
denotes the orthogonal projection. The advantage of working with the space $\cB(\cK)$ is that
\begin{equation}  \label{eq_obs1}
\ker((\pi \circ d_{q\cK})_{|\cB(\cK)}) = \cB_c(q\cK)
\end{equation}
(Observe that $\cB_c(q\cK) \subset \cB(\cK)$).
Eq. \eqref{eq_obs1} will play an important role in Step 2 in the proof of Theorem \ref{theorem1} below.
\end{remark}

\subsubsection*{c) The spaces  $\Check{\cA}^{\orth}(K)$ and $\cA^{\orth}_c(K)$}

In order to obtain  a  simplicial analogue of the decomposition $\cA^{\orth} =  \Check{\cA}^{\orth}
\oplus \cA^{\orth}_c$ in  Eq. \eqref{eq_cAorth_decomp} above
we introduce the following spaces:
\begin{subequations}
\begin{align}
  \cA_{\Sigma,\ct}(K) & := C^1(K_1,\ct) \oplus C^1(K_2,\ct)\\
    \cA_{\Sigma,\ck}(K) & := C^1(K_1,\ck) \oplus C^1(K_2,\ck)\\
\label{eq_CheckcA_disc} \Check{\cA}^{\orth}(K) & :=
 \{ A^{\orth} \in  \cA^{\orth}(K) \mid
\sum\nolimits_{t \in \bZ_{N}}  A^{\orth}(t) \in  \cA_{\Sigma,\ck}(K) \} \\
 \cA^{\orth}_c(K) & := \{ A^{\orth} \in  \cA^{\orth}(K) \mid
 \text{ $A^{\orth}(\cdot)$ is constant and $ \cA_{\Sigma,\ct}(K)$-valued}\} \overset{(*)}{\cong} \cA_{\Sigma,\ct}(K)
 \end{align}
 \end{subequations}
 where in step $(*)$ we made the obvious identification.
 Observe that  we have
\begin{equation} \label{eq4.30}
\cA^{\orth}(K) =  \Check{\cA}^{\orth}(K) \oplus \cA^{\orth}_c(K).
\end{equation}

\begin{convention}  \label{conv_EucSpaces} \rm
In the following we will always consider $\cB(\cK)$, $\cA^{\orth}(K)$
and their subspaces
 as Euclidean vector spaces in the ``obvious''\footnote{More precisely, we will assume that the space
$\cB(\cK)$ (or any subspace of $\cB(\cK)$)  is equipped with the (restriction of the) scalar product $\ll \cdot, \cdot\gg_{\cB(q\cK)}$ on $\cB(q\cK)$, and the space $\cA^{\orth}(K)$ (or any subspace of $\cA^{\orth}(K)$)
is equipped with the  restriction of the scalar product $\ll \cdot, \cdot\gg_{\cA^{\orth}(q\cK)}$,
introduced in Sec. \ref{subsec3.0} above} way.
\end{convention}

\subsection{Discrete analogue of the operator $\tfrac{\partial}{\partial t} + \ad(B)$ in Eq. \eqref{eq_SCS_expl0}}
\label{subsec3.1}

\subsubsection*{a) Discrete analogue(s) of the operator
  $\tfrac{\partial}{\partial t} + \ad(b): C^{\infty}(S^1,\cG) \to  C^{\infty}(S^1,\cG)$, $b \in \ct$}

As a preparation for the next subsection, let us introduce, for fixed $b \in \ct$, two simplicial analogues
 $\hat{L}^{(N)}(b): \Map(\bZ_N,\cG) \to \Map(\bZ_N,\cG)$ and  $\Check{L}^{(N)}(b): \Map(\bZ_N,\cG) \to \Map(\bZ_N,\cG)$ of  the continuum operator $L(b):= \tfrac{\partial}{\partial t} + \ad(b): C^{\infty}(S^1,\cG) \to  C^{\infty}(S^1,\cG)$ by
\begin{align} \label{eq_def_LOp}
 \hat{L}^{(N)}(b) & := N( \tau_1 e^{\ad(b)/N} - \tau_{0})\\
 \Check{L}^{(N)}(b) & := N( \tau_0 - \tau_{-1} e^{-\ad(b)/N})
\end{align}
where $\tau_x$, for  $x \in \bZ_N$,  denotes   the translation operator  $\Map(\bZ_N,\cG) \to  \Map(\bZ_N,\cG)$  given by $(\tau_x f)(t) = f(t +x)$ for all $t \in \bZ_N$.
I want to emphasize that  $\hat{L}^{(N)}(b)$ and  $\Check{L}^{(N)}(b)$
are indeed totally natural simplicial analogues of $L(b)$, see Sec. 5 in \cite{Ha7a} for a detailed motivation.

\begin{remark} \rm The operator $\bar{L}^{(N)}(b):\Map(\bZ_N,\cG) \to \Map(\bZ_N,\cG)$
given by
$$\bar{L}^{(N)}(b)  := \tfrac{N}{2}( \tau_1 e^{\ad(b)/N} - \tau_{-1}  e^{-\ad(b)/N} )$$
might at first look appear to be the natural candidate for a simplicial analogue of the continuum operator
$L(b)$. However, there are several problems with $\bar{L}^{(N)}(b)$.
Firstly, the properties of $\bar{L}^{(N)}(b)$ depend on whether $N$ is odd or even.
Secondly, when $N$ is odd then $\bar{L}^{(N)}(b)$  seems to have the ``wrong'' determinant.
Most probably,  part ii) of Remark \ref{rm_Sec4.3} below  will no longer be  true if we redefine the operator $L^{(N)}(B)$
given in Eq. \eqref{def_LN} below using $\bar{L}^{(N)}(b)$  instead of $\hat{L}^{(N)}(b)$  and $\Check{L}^{(N)}(b)$.
On the other hand, if $N$ is even then  $\bar{L}^{(N)}(b)$   has the ``wrong''\footnote{For example, if $b \in \ct_{reg}$ then we have
$\ker(L(b))  = \{ f \in C^{\infty}(S^1,\cG) \mid f \text{ constant and $\ct$-valued}\}$.
Similarly, we have
$\ker(\hat{L}^{(N)}(b)) = \ker(\Check{L}^{(N)}(b)) = \{f \in \Map(\bZ_N,\cG) \mid f \text{ constant and $\ct$-valued}\}$.
By contrast, $\ker(\bar{L}^{(N)}(b))$ is strictly larger than $ \{f \in \Map(\bZ_N,\cG) \mid f \text{ constant and $\ct$-valued}\}$.} kernel.
\end{remark}

\subsubsection*{b) Discrete analogue of the operator $\tfrac{\partial}{\partial t} + \ad(B):\cA^{\orth} \to \cA^{\orth}$}

For every fixed $B \in \cB(q\cK)$ we first introduce
 the linear operators
 \begin{align*}
 \hat{L}^{(N)}(B) & \text{ on } \Map(\bZ_N, C^1(K_1,\cG)) \cong \oplus_{e \in \face_1(K_1)} \Map(\bZ_N, \cG), \text{ and } \\
 \Check{L}^{(N)}(B) & \text{ on } \Map(\bZ_N, C^1(K_2,\cG)) \cong \oplus_{e \in \face_1(K_2)} \Map(\bZ_N, \cG)
\end{align*}
which are given by
\begin{subequations}
\begin{align} \label{eq_LN_ident1}
\hat{L}^{(N)}(B) & \cong \oplus_{e \in  \face_1(K_1)}
 \hat{L}^{(N)}(B(\bar{e})) \\
 \label{eq_LN_ident2} \Check{L}^{(N)}(B)) & \cong \oplus_{e \in  \face_1(K_2)}
 \Check{L}^{(N)}(B(\bar{e}))
\end{align}
\end{subequations}
where $\bar{e} \in \face_0(q\cK)$ for $e \in  \face_1(K_1) \cup  \face_1(K_2)$
is the barycenter of $e$.

\smallskip

As the simplicial analogue of the the operator $\tfrac{\partial}{\partial t} + \ad(B):\cA^{\orth} \to \cA^{\orth}$
we now take the operator  $L^{(N)}(B):\cA^{\orth}(K) \to \cA^{\orth}(K)$
 which, under the identification
 $$\cA^{\orth}(K) \cong \Map(\bZ_N, C^1(K_1,\cG)) \oplus \Map(\bZ_N, C^1(K_2,\cG)),$$
 is   given by (cf. Remark \ref{rm_Sec4.3}  below for the motivation)
 \begin{equation} \label{def_LN} L^{(N)}(B) = \left( \begin{matrix}
 \hat{L}^{(N)}(B) && 0 \\
0 && \Check{L}^{(N)}(B)
\end{matrix} \right)
\end{equation}

Note that  $L^{(N)}(B)$ leaves the subspace $\Check{\cA}^{\orth}(K)$ of $\cA^{\orth}(K)$
 invariant. The restriction of $L^{(N)}(B)$  to $\Check{\cA}^{\orth}(K)$
 will also be denoted by $L^{(N)}(B)$ in the following.

\subsection{Definition of $S^{disc}_{CS}(A^{\orth},B)$}
\label{subsec3.2}

Recall that $\cK = K_1$ and $\cK' = K_2$ are dual to each other.
As in \cite{Ad1} we can therefore introduce  simplicial Hodge star operators
$\star_{K_1}: C^1(K_1,\bR) \to C^{1}(K_2,\bR)$ and $\star_{K_2}: C^1(K_2,\bR) \to C^{1}(K_1,\bR)$.
These are the linear isomorphisms given by
\begin{equation} \label{eq_Hodge_star_concrete}
   \star_{K_j} e = \pm \Check{e}  \quad \quad  \forall  e \in \face_1(K_j)
 \end{equation}
 for $j=1,2$ where $\Check{e} \in \face_1(K_{3-j})$ is the edge dual to $e$.
The sign $\pm$ above is ``$+$'' if the orientation of $\Check{e}$ is the one induced
by the orientation of $e$ and  the orientation of $\Sigma$, and it is ``$-$'' otherwise.
The two simplicial Hodge star operators above induce
``$\cG$-valued versions'' $\star_{K_1}: C^1(K_1,\cG) \to C^{1}(K_2,\cG)$ and $\star_{K_2}: C^1(K_2,\cG) \to C^{1}(K_1,\cG)$  in the obvious way.\par

Let $\star_K$ be the linear automorphism of
 $\cA_{\Sigma}(K) =  C^1(K_1,\cG) \oplus C^1(K_2,\cG)$ which is given by
\begin{equation} \label{eq_Hodge_matrix}
\star_K := \left(\begin{matrix} 0 && \star_{K_2} \\ \star_{K_1} && 0 \end{matrix}
 \right)
\end{equation}
By $\star_K$ we will also denote the  linear automorphism
of   $\cA^{\orth}(K)$ given by
\begin{equation} \label{eq_star_K_vor_rm}
 (\star_K A^{\orth})(t) = \star_K (A^{\orth}(t)) \quad \quad \forall A^{\orth} \in \cA^{\orth}(K), t \in \bZ_N
 \end{equation}

As the  simplicial analogues of the continuum expression
$S_{CS}(A^{\orth},B)$ in Eq. \eqref{eq_SCS_expl0} above
we  use the expression
 \begin{subequations}  \label{eq_SCS_expl_disc}
\begin{equation}  S^{disc}_{CS}(A^{\orth},B) :=   \pi  k \biggl[ \ll A^{\orth},
\star_K  L^{(N)}(B)
   A^{\orth} \gg_{\cA^{\orth}(q\cK)}
 + 2 \ll  \star_K A^{\orth},  d_{q\cK}  B \gg_{\cA^{\orth}(q\cK)}  \biggr]
\end{equation}
  for $B \in \cB(q\cK)$,  $A^{\orth} \in  \cA^{\orth}(K) \subset \cA^{\orth}(q\cK)$.
Observe that this implies
 \begin{align}  \label{eq_SCS_expl_discb} S^{disc}_{CS}(\Check{A}^{\orth},B) & =   \pi  k  \ll \Check{A}^{\orth},
\star_K  L^{(N)}(B)    \Check{A}^{\orth} \gg_{\cA^{\orth}(q\cK)} \\
\label{eq_SCS_expl_discc} S^{disc}_{CS}(A^{\orth}_c,B) & =  2 \pi  k  \ll  \star_K A^{\orth}_c,  d_{q\cK}  B \gg_{\cA^{\orth}(q\cK)}
\end{align}
 \end{subequations}
for $B \in \cB(q\cK)$,  $\Check{A}^{\orth} \in  \Check{\cA}^{\orth}(K)$, $A^{\orth}_c \in  \cA^{\orth}_c(K)$.

\begin{remark}  \label{rm_Sec4.3} \rm
i) The operator  $\star_K L^{(N)}(B): \cA^{\orth}(K) \to \cA^{\orth}(K)$
is symmetric  w.r.t to the scalar product $\ll  \cdot, \cdot  \gg_{\cA^{\orth}(q\cK)}$,
cf. Proposition 5.3 in \cite{Ha7a}. This would not be the case if on the RHS
of Eq. \eqref{def_LN} we had used  $\hat{L}^{(N)}(B)$ twice (or $\Check{L}^{(N)}(B)$ twice).  \par

ii) According to Proposition 5.1 in \cite{Ha7b} we have
     $\det\bigl(\star_K L^{(N)}(B)_{| \Check{\cA}^{\orth}(K)} \bigr) \neq 0$ for all
\begin{equation}B \in \cB_{reg}(q\cK) := \{ B \in  \cB(q\cK) \mid
B(x) \in \ct_{reg} \text{ for all $x \in \face_0(q\cK)$}\}
\end{equation}
where $\star_K L^{(N)}(B)_{| \Check{\cA}^{\orth}(K)}$ is the restriction
of $\star_K L^{(N)}(B)$ to the invariant
   subspace $\Check{\cA}^{\orth}(K)$ of $\cA^{\orth}(K)$.
\end{remark}

\subsection{Definition of $\Hol^{disc}_{R}(A^{\orth},  B)$  }
\label{subsec3.3}

\subsubsection*{a) Preparation: The simplicial loop case}

A  ``simplicial curve'' in a finite oriented polyhedral cell complex
 $\cP$ is a finite sequence $c=(x^{(k)})_{0 \le k \le n}$, $n \in \bN$, of  vertices in $\cP$
 such that for every $1 \le k \le n$
the two vertices  $x^{(k)}$ and  $x^{(k-1)}$
   either coincide or are the two endpoints
 of an edge $e \in \face_1(\cP)$.
 We will call $n$ the ``length'' of the simplicial curve $c$.
   If $x^{(n)} = x^{(0)}$ we will call
$c= (x^{(k)})_{0 \le k \le n}$ a ``simplicial loop'' in $\cP$.\par

Every simplicial curve  $c= (x^{(k)})_{0 \le k \le n}$  induces a
sequence $(e^{(k)})_{1 \le k \le n}$ of ``generalized edges'', i.e. elements of
$\face_1(\cP) \cup \{0\} \cup (- \face_1(\cP))\subset
C_1(\cP)$ in a natural way. More precisely, we have
$e^{(k)} = 0$ if $x^{(k-1)} = x^{(k)}$ and $e^{(k)} = \pm e$ if $x^{(k-1)} \neq x^{(k)}$
where $e \in \face_1(\cP)$ is the unique edge connecting the vertices $x^{(k-1)}$ and $x^{(k)}$
and where the sign $\pm$ is $+$ if $x^{(k-1)}$ is the starting point of $e$ and $-$ if it is the endpoint.

\begin{convention} \label{conv_loop_pic} For a given simplicial loop  $l= (x^{(k)})_{0 \le k \le n}$
we will  usually write  $\start l^{(k)}$ instead of $x^{(k-1)}$
and $l^{(k)}$ instead of $e^{(k)}$ (for $1 \le k \le n$) where $(e^{(k)})_{1 \le k \le n}$
is the corresponding sequence of generalized edges.
\end{convention}

Let $l=   (x^{(k)})_{0 \le k \le n}$, $n \in \bN$,
 be a simplicial loop in $q\cK \times \bZ_N$
and let  $l_{q\cK}$ and  $l_{\bZ_N}$ be the ``projected'' simplicial loops in $q\cK$ and $\bZ_N$.
Instead of $l_{q\cK}$ and  $l_{\bZ_N}$ we will usually  write
$l_{\Sigma}$ and  $l_{S^1}$. (Recall that $\Sigma$ and $S^1$ are the topological spaces
underlying $q\cK$ and $\bZ_N$.)  \par

For $A^{\orth} \in \cA^{\orth}(K) \subset \cA^{\orth}(q\cK)$ and $B \in \cB(q\cK)$
we now define the following simplicial analogue of the expression $\Hol_{l}(A^{\orth},   B)$
 in Eq. \eqref{eq4.17} (cf. Convention \ref{conv_loop_pic}):
\begin{equation} \label{eq4.18}
\Hol^{disc}_{l}(A^{\orth},   B) :=    \prod_{k=1}^n \exp\biggl(
A^{\orth}(\start l^{(k)}_{S^1})(l^{(k)}_{\Sigma})  +
  B(\start l^{(k)}_{\Sigma}) \cdot dt^{(N)}(l^{(k)}_{S^1}) \biggr)
\end{equation}
with $dt^{(N)} \in C^1(\bZ_{N},\bR)$  given by
 $dt^{(N)}(e)= \tfrac{1}{N}$ for all $e \in \face_1(\bZ_N)$
 and where we  made the identification $C^1(\bZ_{N},\bR) \cong \Hom_{\bR}(C_1(\bZ_{N}),\bR)$
 and  $\cA_{\Sigma}(q\cK) = C^1(q\cK,\cG) \cong \Hom_{\bR}(C_1(q\cK),\cG)$.

\subsubsection*{b)  The simplicial ribbon case}

A ``closed  simplicial  ribbon''  in  a finite oriented polyhedral cell complex $\cP$ is a finite sequence $R = (F_i)_{i \le n}$ of 2-faces of $\cP$ such that every $F_i$ is a tetragon
and such that  $F_i \cap  F_{j} = \emptyset$ unless $i = j$ or $j= i \pm 1$ (mod n).
In the latter case  $F_i$ and $F_{j}$ intersect in  a (full) edge  (cf. Remark 4.3 in Sec. 4.3 in \cite{Ha7a}
and the paragraph before Remark 4.3 in \cite{Ha7a}). \par

From now on we will consider only the special case
where  $\cP = \cK \times \bZ_N$.
Observe that if $R  = (F_k)_{k \le \bar{n}}$, $\bar{n} \in \bN$,
is  a closed simplicial ribbon in  $\cK \times \bZ_N$
 then either all the edges $e_{ij} := F_i \cap  F_{j}$, $j= i \pm 1$ (mod $\bar{n}$),
are parallel to $\Sigma$ or they are all parallel to $S^1$.
In the first case we will call $R$ ``regular''.\par

In the following let  $R  = (F_k)_{k \le \bar{n}}$, $\bar{n} \in \bN$, be a  fixed regular
closed simplicial ribbon in $\cK \times \bZ_N$.
   Observe that $R$  induces three simplicial loops $l^j=(x^{j(k)})_{0 \le k \le n}$, $j=0,1,2$,
   (with\footnote{The common length  $n$ of the three loops is given by $n = 2n_{\Sigma} + n_{S^1}$
     where $n_{\Sigma}$ (and $n_{S_1}$, respectively) is the number of those faces appearing in $R  = (F_k)_{k \le \bar{n}}$,      which are parallel to $\Sigma$ (or are parallel to $S^1$, respectively).
     Observe that since $\bar{n} = n_{\Sigma} + n_{S^1}$ we have
     $\bar{n} \le n \le 2 \bar{n}$.}  $n \le 2 \bar{n}$) in $q\cK \times \bZ_N$   in a natural way,
         $l^1$ and $l^2$ being the two boundary loops of $R$ and $l^0$  being the loop ``inside''  $R$.
[Here we consider $R$ as a subset of $\Sigma \times S^1$ in the obvious way.
    Note that the vertices $(x^{j(k)})_{0 \le k \le n}$, $j=0,1,2$,  appearing above
    are just the elements of $R \cap \face_0(q\cK \times \bZ_N)$.
    The ``starting'' points  $x^{j(0)}$, $j=0,1,2$, of the three simplicial loops $l^j=(x^{j(k)})_{0 \le k \le n}$
     are the three elements of     $e \cap \face_0(q\cK \times \bZ_N)$ where $e \in \face_1(\cK \times \bZ_N)$ is
     the edge $e = e_{1 \bar{n}} =  F_1 \cap F_{\bar{n}}$.] \par
     By  $l^j_{\Sigma}$ and $l^j_{S^1}$, $j=0,1,2$, we will denote
the corresponding ``projected'' simplicial loops in $q\cK$ and $\bZ_N$.\par

Let  $A^{\orth} \in \cA^{\orth}(K) \subset \cA^{\orth}(q\cK)$ and $B \in \cB(q\cK)$.
As the simplicial analogue of the continuum expression
$\Hol_{R}(A^{\orth},   B)$ in Eq. \eqref{eq_sec3.3_1}  we will take
  \begin{multline} \label{eq4.21_full_kurz}
\Hol^{disc}_{R}(A^{\orth},   B) :=
  \prod_{k=1}^n \exp\biggl(  \sum_{j=0}^2 w(j) \cdot \bigl( (A^{\orth}(\start l^{j(k)}_{S^1})\bigr)(l^{j(k)}_{\Sigma}) + B(\start l^{j(k)}_{\Sigma}) \cdot dt^{(N)}(l^{j(k)}_{S^1}) \bigr) \biggr)
\end{multline}
where we use again Convention \ref{conv_loop_pic} and
where we have introduced three weight factors
  $$w(0)=1/2, \quad \quad  w(1)=1/4, \quad \quad w(2)=1/4$$

\begin{remark} \label{rm_subsec3.3} \rm
Other natural choices would be
$$(w(0),w(1),w(2)) = (1/3,1/3,1/3) \quad \quad \text{or} \quad \quad (w(0),w(1),w(2)) = (0,1/2,1/2)$$
However, these two choices would not even lead to the correct values for $\WLO^{disc}_{rig}(L)$ in the special
situation of Sec. \ref{sec5}.
\end{remark}

% From the assumption that $R$ is regular it follows that $l^{0}_{S^1} =  l^{1}_{S^1} = l^{2}_{S^1}$.

\subsection{Definition of $\Det^{disc}(B)$}
\label{subsec3.4}

 Let us first try the following
 ansatz  for the discrete analogue $\Det^{disc}(B)$ of  the heuristic expression $\Det(B)$ given by Eq.
 \eqref{eq_Det(B)_rewrite} above.   For every $B \in \cB_{reg}(q\cK)$ we set
     \begin{equation}  \label{eq_def_DetFPdisc_0}
\Det^{disc}(B) := \prod_{p=0}^2  \biggl(
  \det\bigl(\bigl(1_{{\ck}}-\exp(\ad(B))_{| {\ck}}\bigr)^{(p)}\bigr)\biggr)^{(-1)^p /2}
  \end{equation}
where
$\bigl(1_{{\ck}}-\exp(\ad(B))_{| {\ck}}\bigr)^{(p)}: C^p(\cK,\ck) \to C^p(\cK,\ck)$
is the linear operator  given by
 \begin{equation}\label{eq_3.21} (\bigl(1_{{\ck}}-\exp(\ad(B))_{| {\ck}}\bigr)^{(p)}(\alpha)\bigr)(X) = \bigl(1_{{\ck}}-\exp(\ad(B(\sigma_X)))_{| {\ck}}\bigr) \cdot \alpha(X) \quad
 \forall \alpha \in C^p(\cK,\ck), X \in  \face_p(\cK)
 \end{equation}
where  $\sigma_X \in \face_0(q\cK)$ is the barycenter of $X$.
Observe that we can rewrite Eq. \eqref{eq_def_DetFPdisc_0} in the following way:
   \begin{equation}
\Det^{disc}(B) = \prod_{p=0}^2  \biggl( \prod_{F \in \face_p(\cK)}
  \det\bigl(\bigl(1_{{\ck}}-\exp(\ad(B(\bar{F}))_{| {\ck}}\bigr)^{1/2}\biggr)^{(-1)^p}
  \end{equation}
  where $\bar{F}$ is the barycenter of $F$.

  \smallskip

It turns out however, that this ansatz would not lead to the correct values for
 $\WLO^{disc}_{rig}(L)$ defined below. This is why  we will modify our original ansatz.
 In order to do so we
  first choose a smooth  function $\det^{1/2}\bigl(1_{\ck}-\exp(\ad(\cdot))_{|\ck}\bigr): \ct \to \bR$
with the property $\forall b \in \ct: \bigl(\det^{1/2}\bigl(1_{\ck}-\exp(\ad(b))_{|\ck}\bigr)\bigr)^2 =
 \det\bigl(1_{\ck}-\exp(\ad(b))_{|\ck}\bigr)$.
 Observe that every such function will necessarily take both positive and negative values.
   Motivated by the formula
 $$ \det\nolimits\bigl(1_{{\ck}}-\exp(\ad({b}))_{|{\ck}}\bigr)
  =  \prod_{{\alpha} \in {\cR}} (1 - e^{2 \pi i \langle \alpha, b \rangle})
    = \prod_{{\alpha} \in {\cR}_+} \bigl( 4 \sin^2( \pi    \langle \alpha, b \rangle ) \bigr)$$
(with $\cR$ and $\cR_+$ as in  Sec. \ref{subsec5.2} below)
we  will make the choice
 \begin{equation} \label{eq_det_in_terms_of_roots}
 \det\nolimits^{1/2}\bigl(1_{{\ck}}-\exp(\ad({b}))_{|{\ck}}\bigr) =  \prod_{{\alpha} \in {\cR_+}} \bigl( 2 \sin( \pi    \langle \alpha, b \rangle ) \bigr)
 \end{equation}
 and then redefine $\Det^{disc}(B)$ for $B \in \cB_{reg}(q\cK)$ by
  \begin{equation} \label{eq_def_DetFPdisc}
\Det^{disc}(B) := \prod_{p=0}^2  \biggl( \prod_{F \in \face_p(\cK)}
  \det\nolimits^{1/2}\bigl(1_{{\ck}}-\exp(\ad(B(\bar{F}))_{| {\ck}})\bigr)\biggr)^{(-1)^p}
  \end{equation}

\begin{remark} \rm In the published versions of \cite{Ha7a,Ha7b}
there is a notational inaccuracy. When we write $\det\bigl(1_{{\ck}}-\exp(\ad({b}))_{|{\ck}}\bigr)^{1/2}$
in \cite{Ha7a,Ha7b} we actually mean  $\det\nolimits^{1/2}\bigl(1_{{\ck}}-\exp(\ad({b}))_{|{\ck}}\bigr)$
given as in Eq. \eqref{eq_det_in_terms_of_roots} above.
\end{remark}

\subsection{Discrete version of $1_{C^{\infty}(\Sigma,\ct_{reg})}(B)$}
\label{subsec3.6}

Let us fix a $s>0$ which is sufficiently small\footnote{\label{ft_sec3.6} $s$ needs to be smaller than the distance
between the two sets $\ct_{sing}$ and $\ct_{reg} \cap \tfrac{1}{k} \Lambda$
where $k$ is as in Sec. \ref{sec2} and $\Lambda \subset \ct$ is the weight lattice,
cf. Sec. \ref{subsec5.2} below}
 and choose $1^{\smooth}_{\ct_{reg}} \in C^{\infty}(\ct,\bR)$
such that
\begin{itemize}
\item $0 \le 1^{\smooth}_{\ct_{reg}} \le 1$
\item  $1^{\smooth}_{\ct_{reg}} = 0$ on a neighborhood of $\ct_{sing}:= \ct \backslash \ct_{reg}  $
\item $1^{\smooth}_{\ct_{reg}} = 1$ outside the $s$-neighborhood of  $\ct_{sing}$
\item $1^{\smooth}_{\ct_{reg}}$ is  invariant under the operation of the affine Weyl group $\cW_{\aff}$
 on $\ct$ (cf. Sec. \ref{subsec5.2} below).
\end{itemize}

\noindent
For fixed  $B \in \cB(q\cK)$ we will now take
 the expression
\begin{equation} \label{eq_sec3.6}
\prod_{x} 1^{\smooth}_{\ct_{reg}}(B(x)):=
\prod_{x \in \face_0(q\cK)}
1^{\smooth}_{\ct_{reg}}(B(x))
\end{equation}
as  the discrete analogue of $1_{C^{\infty}(\Sigma,\ct_{reg})}(B)$.

\subsection{Oscillatory  Gauss-type measures}
 \label{subsec3.7}

i)   An ``oscillatory
 Gauss-type measure'' on  Euclidean vector space $(V, \langle \cdot, \cdot \rangle)$
 is a  complex Borel measure $d\mu$ on $V$
 of the form
 \begin{equation} \label{eq3.1}
 d\mu(x) = \tfrac{1}{Z} e^{ - \tfrac{i}{2} \langle x - m, S (x-m) \rangle} dx
\end{equation}
with $Z \in \bC \backslash \{0\}$,
   $m \in V$, and where  $S$ is a  symmetric endomorphism of $V$  and
 $dx$  the normalized\footnote{i.e. unit hyper-cubes have volume $1$ w.r.t. $dx$}
 Lebesgue measure on $V$.
 Note that $Z$, $m$ and $S$ are uniquely determined by $d\mu$. We will often use
 the notation  $m_{\mu}$ and $S_{\mu}$ in order to refer to  $m$ and $S$.

\begin{itemize}
\item We call $d\mu$ ``centered''iff $m_{\mu}=0$.

\item  We call $d\mu$ ``degenerate'' iff $S_{\mu}$ is not invertible
 \end{itemize}

 \smallskip

 ii)   Let  $d\mu$  be an oscillatory
 Gauss-type measure on a  Euclidean vector space $(V, \langle \cdot, \cdot \rangle)$.
 A (Borel) measurable function
  $f: V \to \bC$ will be called improperly integrable w.r.t. $d\mu$
  iff\footnote{Observe that
$\int_{\ker(S_{\mu})}  e^{- \eps \|x\|^2} dx =
(\tfrac{\eps}{\pi})^{-n/2}$.
In particular, the factor
$(\tfrac{\eps}{\pi})^{n/2} $ in Eq. \eqref{eq3.2} above  ensures
that also for degenerate oscillatory
 Gauss-type measure the improper integrals  $\int\nolimits_{\sim} 1 \ d\mu$  exists}
 \begin{equation}\label{eq3.2} \int\nolimits_{\sim} f d\mu := \int\nolimits_{\sim} f(x)
   d\mu(x): =
    \lim_{\eps \to 0} (\tfrac{\eps}{\pi})^{n/2} \int f(x) e^{- \eps |x|^2} d\mu(x)
  \end{equation}
  exists. Here  we have set  $n:=\dim(\ker(S_{\mu}))$.
   Note that if $d\mu$ is non-degenerate we have $n=0$ so the factor $(\tfrac{\eps}{\pi})^{n/2}$
is then trivial.
\begin{itemize}
\item
 We call $d\mu$ ``normalized'' iff $\int\nolimits_{\sim} 1 d\mu  = 1$.
 \end{itemize}

\subsection{Simplicial versions of the two Gauss-type measures in Eq. \eqref{eq2.48}}
  \label{subsec3.8}

i) As the simplicial analogue of the heuristic complex measure $d\Check{\mu}^{\orth}_B =
\tfrac{1}{\Check{Z}(B)} \exp(i S_{CS}(\Check{A}^{\orth},B)) D\Check{A}^{\orth}$
in Eq. \eqref{eq2.48} we will take  the (rigorous) complex measure
 \begin{equation} \label{eq_def_mu_orth_disc} d\Check{\mu}^{\orth,disc}_{B}(\Check{A}^{\orth}):= \tfrac{1}{\Check{Z}^{disc}(B)} \exp(iS^{disc}_{CS}(\Check{A}^{\orth},B))  D\Check{A}^{\orth}
\end{equation}
on $\Check{\cA}^{\orth}(K)$ where
$D\Check{A}^{\orth}$ denotes the (normalized) Lebesgue measure
on $\Check{\cA}^{\orth}(K)$ and where we have set
 $\Check{Z}^{disc}(B) := \int\nolimits_{\sim} \exp(iS^{disc}_{CS}(\Check{A}^{\orth},B))  D\Check{A}^{\orth}$.
Observe that  $d\Check{\mu}^{\orth,disc}_{B}(\Check{A}^{\orth})$ is not well-defined
for all $B$. However, if $B \in \cB_{reg}(q\cK)$, which is the only case relevant for us (cf. Sec. \ref{subsec3.6} above), Eq. \eqref{eq_SCS_expl_disc} and Remark \ref{rm_Sec4.3} above imply that
 the complex measure in Eq. \eqref{eq_def_mu_orth_disc} is indeed a well-defined,
 non-degenerate,  centered, normalized  oscillatory Gauss type measure on $\Check{\cA}^{\orth}(K)$.

 \medskip

\noindent ii) As the simplicial analogue of the heuristic complex measure
$\exp(i S_{CS}(A^{\orth}_c, B)) (DA^{\orth}_c \otimes DB)$ in Eq.  \eqref{eq2.48} we will take the (rigorous) complex measure on $\cA^{\orth}_c(K) \oplus \cB(\cK)$
\begin{equation} \label{eq_def_mu2_disc}  \exp(i  S^{disc}_{CS}(A^{\orth}_c,B))    (DA^{\orth}_c \otimes  DB)
\end{equation}
where  $DA^{\orth}_c$  denotes the (normalized) Lebesgue measure on $\cA^{\orth}_c(K)$
and  $DB$  the (normalized) Lebesgue measure on $\cB(\cK)$.\par

According to Eq. \eqref{eq_SCS_expl_disc} above, the (rigorous) complex measure in Eq. \eqref{eq_def_mu2_disc}
is a  centered  oscillatory Gauss type measure on $\cA^{\orth}_c(K) \oplus \cB(\cK)$.

\subsection{Definition of $\WLO^{disc}_{rig}(L)$ and $\WLO^{disc}_{norm}(L)$}
 \label{subsec3.9}

A finite tuple $L= (R_1, R_2, \ldots, R_m)$, $m \in \bN$, of closed simplicial ribbons in
$\cK \times \bZ_{N}$ which do not intersect each other
will be called a ``simplicial ribbon link'' in $\cK \times \bZ_{N}$.
 For every such simplicial ribbon link $L= (R_1, R_2, \ldots, R_m)$ in
$\cK \times \bZ_{N}$ equipped with a tuple of ``colors''
$(\rho_1,\rho_2,\ldots,\rho_m)$, $m \in \bN$,
we now introduce  the following
 simplicial   analogue $\WLO^{disc}_{rig}(L)$
of the heuristic expression $\WLO(L)$ in Eq. \eqref{eq2.48}:
\begin{multline} \label{eq_def_WLOdisc}
\WLO^{disc}_{rig}(L)  :=
  \sum_{y \in I}\int\nolimits_{\sim}  \bigl( \prod_{x} 1^{\smooth}_{\ct_{reg}}(B(x))
  \bigr) \Det^{disc}(B)\\
\times \biggl[
\int\nolimits_{\sim}   \prod_{i=1}^m  \Tr_{\rho_i}\bigl( \Hol^{disc}_{R_i}(\Check{A}^{\orth} +  A^{\orth}_c, B)\bigr)   d\Check{\mu}^{\orth,disc}_{B}(\Check{A}^{\orth}) \biggr] \\
 \times        \exp\bigl( - 2 \pi i k  \langle y, B(\sigma_0) \rangle \bigr)
  \exp(i  S^{disc}_{CS}(A^{\orth}_c,B))    (DA^{\orth}_c \otimes  DB)
\end{multline}
{where $\sigma_0$ is an arbitrary fixed point of $\face_0(q\cK)$
which does not lie in   $\bigcup_{i \le m}  \Image(\pi_{\Sigma} \circ R_i)$.
Here we  consider  each $R_i$ as a continuous map $[0,1] \times S^1 \to \Sigma \times S^1$
in the obvious way (cf. Remark 4.3 in Sec. 4.3 in \cite{Ha7a}).

\smallskip

Apart from considering the  simplicial analogue $\WLO^{disc}_{rig}(L)$ of the heuristic expression
$\WLO(L)$ in Eq. \eqref{eq2.48} it will also
be convenient to introduce a  simplicial analogue of the normalized heuristic expression
$$\WLO_{norm}(L):= \frac{\WLO(L)}{\WLO(\emptyset)}$$
where $\emptyset$ is the ``empty link'' in $M = \Sigma \times S^1$.
Accordingly, we will now define
\begin{equation} \label{eq_def_WLO_norm}
\WLO^{disc}_{norm}(L):= \frac{\WLO^{disc}_{rig}(L)}{\WLO^{disc}_{rig}(\emptyset)}
\end{equation}
where $L$ is the colored simplicial ribbon link fixed above
and where $\WLO^{disc}_{rig}(\emptyset)$ is defined in the obvious way, i.e.
by the expression we get from the RHS of Eq. \eqref{eq_def_WLOdisc}
after replacing the product   $\prod_{i=1}^m  \Tr_{\rho_i}\bigl( \Hol^{disc}_{R_i}(\Check{A}^{\orth} +  A^{\orth}_c, B)\bigr)$  by $1$.

\medskip

We conclude this section with four important remarks.
In Remark \ref{rm_sec3.9} we compare the main aims \& results of the present paper
with those in \cite{Ha7a,Ha7b}. In Remark \ref{rm_sec3.9'} we make some comments regarding
the case of general ribbon links $L$.
In Remark \ref{rm_sec3.9b} we describe how the main result of the present paper fits
into the bigger picture of the ``simplicial program'' for Chern-Simons theory (cf. also ``Goal 1''
of the Introduction). Finally, Remark \ref{rm_sec3.9c} we clarify some points related to ``Goal 2''
of the Introduction.

\begin{remark} \label{rm_sec3.9} \rm
 In \cite{Ha7a,Ha7b} a simplicial analogue of $\WLO_{norm}(L)$
which is closely related to the simplicial analogue $\WLO^{disc}_{norm}(L)$  above was
evaluated explicitly for a simple type of simplicial ribbon links $L$, cf. Theorem 6.4 in \cite{Ha7a}.
An analogous result can be obtained for our $\WLO^{disc}_{norm}(L)$ above\footnote{or rather for
  $\WLO^{disc}_{norm}(L)$ after making the modification (M2) described in Sec. \ref{subsec3.10} below}.
More precisely,  it can be shown that
for every  simplicial ribbon link $L= (R_1, R_2, \ldots, R_m)$ in $\cK \times \bZ_{N}$
which   fulfills an analogue of Conditions (NCP)' and (NH)' in \cite{Ha7a}
we have   $\WLO^{disc}_{norm}(L) = |L|/|\emptyset|$
 where $|\cdot|$ is the shadow invariant on  $M= \Sigma \times S^1$ associated to $\cG$ and $k$ as above.
This result is interesting because it shows how major ``building blocks''
of the shadow invariant arise within the torus gauge approach to the CS path integral.
However, from a knot theoretic point of view the class of simplicial ribbon links $L$
fulfilling (the analogue of) Condition (NCP)' in \cite{Ha7a} is not very interesting.
In particular, this class of  (simplicial ribbon) links does not include any non-trivial knots. \par

One of the main aims of the present paper is to show that the torus gauge approach to the CS path integral
also allows the treatment of non-trivial knots, namely a large class of torus (ribbon)
knots in $S^2 \times S^1$, cf. Definition \ref{def5.3} and  Theorem \ref{theorem1} in Sec. \ref{sec5} below.
\end{remark}

\begin{remark} \label{rm_sec3.9'} \rm
The simplicial ribbon knots/links $L= (R_1, R_2, \ldots, R_m)$, $m \in \bN$,
  mentioned in Remark \ref{rm_sec3.9} above  have the special property that
  the projected ribbons $\pi_{\Sigma} \circ R_i$, $i \le m$, in $\Sigma$ have either no (self-)intersections
   (= the situation in \cite{Ha7a,Ha7b})    or only ``longitudinal'' self-intersections (= the situation in Definition \ref{def5.3},  Theorem \ref{theorem1} and Theorem \ref{theorem2} below).
As explained in Sec. 6  in \cite{Ha7b},  if we want to have a chance of
 obtaining the correct values for the rigorous version of $\WLO_{norm}(L)$
  for general simplicial ribbon links $L= (R_1, R_2, \ldots, R_m)$
  (where  the projected ribbons $\pi_{\Sigma} \circ R_i$, $i \le m$,
  are allowed to  have ``transversal'' intersections) we will probably have
  to modify our approach in a suitable way. One way to do so is to  make  what  in Sec. 7 in \cite{Ha7b}
 was called the ``transition to the $BF$-theory setting''.
Alternatively, one can use a ``mixed''  approach where
some of the simplicial spaces are embedded naturally into suitable continuum
spaces\footnote{For example, we can exploit the embeddings
 $C^p(\cK,V) \hookrightarrow C^p(b\cK,V) \overset{W}{\hookrightarrow} \Omega^p(\Sigma^{(2)},V)$
for $p=0,1,2$ and  $V \in \{\cG,\ct\}$, where $W$ is  the Whitney map of the simplicial complex
$b\cK$ and $\Sigma^{(2)}$ is the complement of the 1-skeleton of $b\cK$ in $\Sigma$.
Observe that $b\cK$ induces in a natural way a Riemannian metric on $\Sigma^{(2)}$,
which gives rise to a Hodge star operator $\star:\Omega^p(\Sigma^{(2)},V) \to \Omega^{2-p}(\Sigma^{(2)},V)$},
cf.  \cite{Ha10}. This leads to a greater flexibility
and allows us, for example, to work with (continuum) Hodge star operators and continuum ribbons instead
of the simplicial Hodge star operators  and simplicial ribbons mentioned above.
\end{remark}

\begin{remark} \label{rm_sec3.9b} \rm
 The longterm goal of what in Sec.  \ref{sec1} was called the ``simplicial program'' for Chern-Simons theory
 (cf. Sec. 3 in \cite{Ha7a} and see also \cite{Mn2})
      is to find for every oriented closed 3-manifold $M$   and every colored (ribbon) link $L$ in $M$
      a rigorous simplicial realization $\WLO^{disc}(L)$
      of the original or   gauge-fixed CS path integral for the WLOs associated
      to $L$  such that $\WLO^{disc}(L)$  coincides with the corresponding Reshetikhin-Turaev invariant $RT(M,L)$.

      In the present paper we are much less ambitious. Firstly, we only consider
      the special case $M= \Sigma \times S^1$ (and from Sec. \ref{sec5} on we restrict ourselves
      to the case $\Sigma = S^2$) and secondly, we only deal with a restricted
      class of simplicial ribbon links  $L$ (cf.  Theorem \ref{theorem1} and Theorem \ref{theorem2} below).
 \end{remark}

 \begin{remark} \label{rm_sec3.9c} \rm In view of ``Goal 2''
 in Comment \ref{comm1}  of the Introduction note that  $\WLO^{disc}_{norm}(L)$
can be interpreted  as a (convenient) ``lattice  regularization'' of the heuristic continuum expressions
 $\WLO_{norm}(L)$  above.
 Usually, when one  works with a lattice regularization in Quantum Gauge Field Theory
 one has to perform a suitable continuum limit.
 We can do this here as well\footnote{\label{ft_distinguished} There is, however, a major difference compared to the standard  situation in QFT where the continuum limit is usually independent of the lattice regularization.
 In the case of the  Chern-Simons path integral  (in the torus gauge) the value of the continuum limit
 will depend on the lattice regularization. In particular, only a distinguished subclass of
 lattice regularizations will  lead to the correct result, cf. \cite{Ha10}
 for an interpretation of this phenomenon}.
 So, instead of working with a fixed $\cK$ and
 $\bZ_N$ with fixed $N \in \bN$ let us now consider a sequence $(\cK^{(n)})_{n \in \bN}$ of consecutive refinements
 of $\cK$   and $(\bZ_{N^{(n)}})_{n \in \bN}$ where $N^{(n)}:= n \cdot N$.
 By doing so we can approximate every ``horizontal''\footnote{i.e. a ribbon link in $M  = \Sigma \times S^1$
 which, when considered as a framed link instead of a ribbon link,
 is ``horizontally framed'' in the sense that the framing vector field
 is   ``parallel'' to the $\Sigma$-component of  $M  = \Sigma \times S^1$}
  ribbon  link $L$   in $M  = \Sigma \times S^1$
 by a suitable sequence $(L^{(n)})_{n \in \bN}$  of simplicial ribbon links $L^{(n)}$ in $\cK^{(n)} \times \bZ_{N^{(n)}}$.\par

 Let us now restrict our attention to horizontal ribbon links $L$ in $M  = \Sigma \times S^1$
 which are analogous  to the simplicial ribbon links
appearing in Theorem \ref{theorem1} and Theorem \ref{theorem2} below
and let $(L^{(n)})_{n \in \bN}$ be a suitable approximating sequence as above.
Then we obtain, informally\footnote{\label{ft_distinguished2}and under the assumption that the simplicial framework we use in Sec. \ref{sec3}
indeed belongs to the ``distinguished subclass'' of lattice regularizations mentioned in Footnote \ref{ft_distinguished} above},
\begin{equation} \label{eq_WLO_appr}
\WLO_{norm}(L) = \lim_{n \to \infty} \WLO^{disc}_{norm}(L^{(n)})
\end{equation}
where $\WLO^{disc}_{norm}(L^{(n)})$ is defined in an analogous way
as $\WLO^{disc}_{norm}(L)$ above (with $\cK^{(n)}$  playing the role of
$\cK$ and $\bZ_{N^{(n)}}$ playing the role of $\bZ_{N}$).\par

But from (the proof of) Theorem \ref{theorem1} and Theorem \ref{theorem2}
it follows that  $\WLO^{disc}_{norm}(L^{(n)})$
does not depend on $n$ (provided that $\cK$ was chosen  fine enough and $N$  large enough).
Accordingly, the $n \to \infty$-limit in Eq. \eqref{eq_WLO_appr} is trivial
and we simply obtain
$$\WLO_{norm}(L) =  \WLO^{disc}_{norm}(L^{(1)})$$
So in order to evaluate the heuristic expression $\WLO_{norm}(L)$
(for the special type of continuum ribbon links $L$ we are considering here) it is enough to compute
$\WLO^{disc}_{norm}(L^{(1)})$. And this is exactly what is done in
Theorem \ref{theorem1} and Theorem \ref{theorem2} (with $\cK^{(1)}$ replaced by $\cK$).
\end{remark}

\subsection{Modification of the definition of $\WLO^{disc}_{rig}(L)$ and $\WLO^{disc}_{norm}(L)$}
\label{subsec3.10}

As we will see later the definition of  $\WLO^{disc}_{rig}(L)$ and of $\WLO^{disc}_{norm}(L)$
above need\footnote{I consider this to be a purely technical issue which can probably be resolved by using an alternative  way for making rigorous sense
of the RHS of Eq. \eqref{eq2.48}, cf.  Remark  \ref{rm_subsec3.11} below}
 to be modified slightly if we want to obtain the correct values for $\WLO^{disc}_{norm}(L)$.
Without such a modification a ``wrong'' factor $1_{\ct_{reg}}(B(Z_0))$ will appear at the end of the computations
in Step 4 in Sec. \ref{subsec5.4} below.
 Here are two modifications of the current approach for each of which this extra factor does not appear and one indeed obtains the correct values  for  $\WLO_{norm}^{disc}(L)$:
\begin{enumerate}

\item[(M1)] Instead of working with closed simplicial ribbons in $\cK \times \bZ_N$
  we could work with closed simplicial ribbons in $q\cK \times \bZ_N$.
  In fact, this is exactly what was done in \cite{Ha7a,Ha7b} in the situation studied there.
  The disadvantage of this kind of modification is that the space $\cB(\cK)$ above needs to be replaced by a
  less natural space. Moreover, the proof of the analogue of  Lemma \ref{lem2} in Sec. \ref{subsec5.4}
  will become unnaturally complicated.

 \item[(M2)] We regularize the  RHS of  Eq. \eqref{eq_def_WLOdisc}  in a suitable way. In order to do so we
  first choose a fixed vector $v \in \ct$ which is not orthogonal
     to any of the roots $\alpha \in \cR$. Then we define $B_{displace} \in \cB(q\cK)$ by
            $$ B_{displace}(x) =
             \begin{cases} 0 & \text{ if } x \in \face_0(\cK) \\
                          v & \text{ if } x \in \face_0(q\cK) \backslash \face_0(\cK)
             \end{cases} $$
       and set, for each $\beta >0$ and each $B \in \cB(\cK) \subset \cB(q\cK)$,
             $$B(\beta) = B + \beta B_{displace}$$
          After that we replace $B$ by $B(\beta)$  in each of the  three terms
          $\prod_{x} 1^{\smooth}_{\ct_{reg}}(B(x))  \bigr)$,
          $\Det^{disc}(B)$, and $d\Check{\mu}^{\orth,disc}_{B}(\Check{A}^{\orth})$
         appearing  on the RHS of   Eq. \eqref{eq_def_WLOdisc}.
       Finally, we let\footnote{without letting $s\to 0$ first,
      the $\beta \to 0$ limit has no effect}  $s \to 0$ and later $\beta \to 0$.
     More precisely, we add $\lim_{\beta \searrow 0}  \lim_{s \searrow 0}    \cdots $
     in front of the (modified) RHS of Eq.   \eqref{eq_def_WLOdisc}.
($\WLO^{disc}_{norm}(L)$ is again defined by Eq. \eqref{eq_def_WLO_norm}.)
 \end{enumerate}

 During the proof of Theorem \ref{theorem1} below, which will be given in Sec. \ref{subsec5.4} below
 we will first  work with the original definition of  $\WLO^{disc}_{rig}(L)$ in Sec. \ref{subsec3.9}
 until the end of ``Step 4''. This is instructive because we see how the factor $1_{\ct_{reg}}(B(Z_0))$
 arises. After that we switch to the modified definition of  $\WLO^{disc}_{rig}(L)$
 (using either of the two options (M1) and (M2)  above) and complete the proof.

\begin{remark}  \label{rm_subsec3.11} \rm
The simplicial approach described above for obtaining a  rigorous realization
of $\WLO(L)$ is  simple and  fairly natural and it will be  sufficient
for the goals  of the present paper, cf. ``Goal 1'' and ``Goal 2'' in Comment \ref{comm1}
in the Introduction. \par

That  said I want to emphasize that
even though  the approach above is probably one of the simplest ways
for making rigorous sense of the RHS of Eq. \eqref{eq2.48}
(for the special simplicial ribbon links we are interested in in the present paper)
I do not claim that it is the best way of obtaining such a rigorous realization. It is  likely that
 several improvements are possible and that, in particular,
there is an alternative to modification (M1) or modification (M2) which is more natural.
\end{remark}

 \section{Some useful results on oscillatory Gauss-type measures}
\label{sec4}

In the present section we will review (without proof and in a slightly modified form)
 some of the basic definitions and results in \cite{Ha7b}
 on oscillatory Gauss-type measures.

 \smallskip

In the following let  $(V, \langle \cdot, \cdot \rangle)$ be a Euclidean vector space
and $d\mu$  an  oscillatory  Gauss-type measure on $(V, \langle \cdot, \cdot \rangle)$,
cf. Sec. \ref{subsec3.7} above.

\begin{proposition}  \label{obs1}
  If $d\mu$ is normalized and  non-degenerate then we have for all $v, w \in V$
\begin{equation} \label{eq3.4}
 \int\nolimits_{\sim} \langle v, x \rangle \ d\mu(x)  =   \langle v, m \rangle  , \quad \quad
 \int\nolimits_{\sim} \langle v, x \rangle \langle w, x \rangle \ d\mu(x)
  = \tfrac{1}{i} \langle v, S^{-1}  w \rangle +  \langle v, m \rangle \langle w, m \rangle
\end{equation}
where  $m = m_{\mu}$ and $S = S_{\mu}$.
\end{proposition}

\begin{definition} \label{def3.3}
 By $\cP_{exp}(V)$ we denote the subalgebra
 of $\Map(V,\bC)$
which is generated by the polynomial functions  $f:V \to \bC$
and  all  functions $f:V \to \bC$ of the form
$f = \theta \circ \exp_{\End(\bC^n)} \circ \varphi$, $n \in \bN$, where
 $\theta: \End(\bC^n) \to \bC$ is linear, $\varphi: V \to \End(\bC^n)$ is affine, and
$\exp_{\End(\bC^n)}:\End(\bC^n) \to \End(\bC^n)$ is the exponential map of
the (normed) $\bR$-algebra $\End(\bC^n)$.
\end{definition}

\begin{proposition} \label{prop3.1} For every  $f \in  \cP_{exp}(V)$ the  improper integral  $\int\nolimits_{\sim} f \ d\mu \in \bC$ exists.
\end{proposition}

\begin{proposition} \label{prop3.3} If $d\mu$ is  normalized and  non-degenerate and if
 $(Y_k)_{k \le n}$, $n \in \bN$, is  a sequence of affine maps $V \to \bR$ such that
\begin{equation} \label{eq3.10}  \int\nolimits_{\sim} Y_i Y_j d\mu  = \bigl( \int\nolimits_{\sim} Y_i d\mu  \bigr) \bigl(
\int\nolimits_{\sim} Y_j d\mu \bigr) \quad \forall i,j \le n
\end{equation}
 then  we have for every  $\Phi \in \cP_{exp}(\bR^n)$
\begin{equation} \label{eq3.11} \int\nolimits_{\sim}  \Phi((Y_k)_{k}) d\mu
 =\Phi\bigl(\bigl(  \int\nolimits_{\sim}  Y_k  d\mu \bigr)_{k} \bigr)
\end{equation}

A totally analogous statement holds in the situation where
 instead of the sequence $(Y_k)_{k \le n}$,
$n \in \bN$, we have  a family $(Y^{a}_k)_{k \le n,
a \le d}$ of affine maps fulfilling the obvious analogue of Eq.
\eqref{eq3.10} and where $\Phi \in \cP_{exp}(\bR^{n \times d})$.
\end{proposition}

\begin{definition} \label{conv3.1} \rm
Let $f:V \to \bC$ be  a continuous function, let $d:= \dim(V)$,
and let $dx$ be the normalized Lebesgue measure on $V$.
 We set
\begin{equation}
 \int^{\sim}_{V} f(x) dx := \tfrac{1}{\pi^{d/2}} \lim_{\eps \to 0} \eps^{d/2}   \int_{V} e^{-\eps |x|^2} f(x) dx
\end{equation}
whenever  the expression on the RHS  of the previous equation is well-defined.
\end{definition}

\begin{remark}\label{rm_last_sec3}  \rm
Let $\Gamma$ be a lattice in $V$ and  $f:V \to \bC$ a $\Gamma$-periodic continuous function.
Then  $\int^{\sim}_{V} f(x) dx$ exists and we have
 \begin{equation}  \label{eq2_lastrmsec3}
\int^{\sim}_{V} f(x) dx = \frac{1}{vol(Q)} \int_{Q}  f(x) dx
 \end{equation}
 with  $Q :=\{\sum_i x_i e_i \mid 0 \le x_i \le 1 \forall i \le d\}$
 where $(e_i)_{i \le d}$ is  an arbitrary basis of the lattice $\Gamma$
 and where  $vol(Q)$ denotes the volume of $Q$.
 Clearly, Eq. \eqref{eq2_lastrmsec3} implies
 \begin{equation}  \label{eq1_lastrmsec3}  \forall y \in   V: \quad
 \int^{\sim}_{V} f(x) dx =  \int^{\sim}_{V} f(x + y) dx
 \end{equation}
\end{remark}

\begin{proposition} \label{prop3.5}
Assume that $V= V_0 \oplus V_1 \oplus V_2$ where $V_0$, $V_1$, $V_2$ are pairwise
orthogonal subspaces of $V$. (We will denote  the $V_j$-component of $x \in V$ by  $x_j$ in the following.)
Assume also that $d\mu$ is a  (centered) normalized  oscillatory  Gauss-type measure on  $(V, \langle \cdot, \cdot \rangle)$ of the form $d\mu(x)=  \tfrac{1}{Z} \exp(i \langle x_2,M x_1) dx$
for some linear isomorphism $M: V_1 \to V_2$.
Then, for every  $v \in V_2$
and every   bounded uniformly  continuous function $F:V_0 \oplus V_1 \to \bC$
the LHS of the following equation exists  iff the RHS  exists
and in this case we have
\begin{equation} \label{eq3.18}
\int_{\sim} F(x_0 + x_1) \exp(i \langle x_2,v \rangle ) d\mu(x) =  \int^{\sim}_{V_0} F(x_0  - M^{-1} v) dx_0,
\end{equation}
where $dx_0$ is the normalized Lebesgue measure on $V_0$.
\end{proposition}

\section{Evaluation of $\WLO^{disc}_{rig}(L)$ for torus ribbon knots $L$ in $S^2\times S^1$}
\label{sec5}

From now on we will only consider the special case $\Sigma = S^2$.

\subsection{A certain class of torus (ribbon) knots in $S^2\times S^1$}
\label{subsec5.1}

Recall that a torus knot  in $S^3$ is a knot
which is contained in an unknotted torus $\tilde{\cT}  \subset  S^3$.
Motivated by this definition we will now introduce an analogous notion
for knots in the manifold $M = S^2 \times S^1$.

\begin{definition} \label{def5.1}
A torus knot in $S^2 \times S^1$ of standard type is a knot in $S^2 \times S^1$
which is contained in a torus  $\cT$ in $ S^2 \times S^1$  fulfilling
the following condition
\begin{description}
\item[(T)]  $\cT$ is of the form $\cT = \psi(\cT_0)$ with $\cT_0 := C_0 \times S^1$
 where $C_0$ is an embedded circle in $S^2$ and  $\psi:S^2 \times S^1 \to S^2 \times S^1$ is a  diffeomorphism.
\end{description}
\end{definition}

\begin{remark} \rm \label{rm_sec5.1} Note that every unknotted torus $\tilde{\cT}$ in $S^3$ can be obtained
from a  torus  $\cT$ in $ S^2 \times S^1$ fulfilling condition (T) by performing a suitable
Dehn surgery on a separate  knot in  $ S^2 \times S^1$.
Consequently, every torus knot $\tilde{K}$ in $S^3$ can be obtained from a torus knot $K$ in $S^2 \times S^1$ of standard type by performing such a Dehn surgery.
Moreover, even if we restrict ourselves to the special situation where
$K$ lies in $\cT_0 = C_0 \times S^1$ for $C_0$ as above
we can still obtain all torus knots in $S^3$ up to equivalence by performing a suitable Dehn surgery.
 We will exploit this fact in Sec. \ref{subsec6.2} below.
\end{remark}

Let us now go back to the simplicial setting introduced in Sec. \ref{sec3}.
Recall that in Sec. \ref{sec3} we fixed two polyhedral cell complexes  $\cK$ and $\bZ_N$
and considered also their product $\cK \times \bZ_N$.
The topological space underlying $\cK \times \bZ_N$ is $\Sigma \times S^1 = S^2 \times S^1$.
We want to find a ``simplicial analogue'' of Definition \ref{def5.1} above.
In view of Remark \ref{rm_sec5.1} we will work with the following definition:

\begin{definition} \label{def5.2} Let $l$ be a simplicial loop in $\cK \times \bZ_N$
(which we will consider as a continuous map $S^1 \to S^2 \times S^1$ in the obvious way).
We say that $l$ is a simplicial torus knot of standard type iff  $l: S^1 \to S^2 \times S^1$ is an embedding
and also the following condition is fulfilled:
\begin{description}
\item[(TK)] $\Image(l)$ is contained in  $\cT_0 := C_0 \times S^1$
where $C_0$ is some embedded  circle in $S^2$ which lies on the 1-skeleton of  $\cK$.
\end{description}
By ${\mathbf p}(l)$ and  ${\mathbf q}(l)$ we will denote the
winding numbers of $\pi_i \circ l:S^1 \to S^1$, $i = 1,2$, where
$\pi_1$ and $\pi_2$ are the two canonical projections $\cT_0 = C_0 \times S^1 \cong S^1 \times S^1 \to S^1$
where for the identification $C_0 \cong S^1$ we picked an  orientation on $C_0$.
(Observe that ${\mathbf p}(l)$ and  ${\mathbf q}(l)$ will always be coprime.)
\end{definition}

\begin{definition} \label{def5.3} Let $R$ be a closed simplicial ribbon in $\cK \times \bZ_N$
(which we will consider as a continuous map $S^1 \times [0,1] \to S^2 \times S^1$
in the obvious way). We say that $R$  is a simplicial torus ribbon knot of standard type
iff it is regular (cf. Sec. \ref{subsec3.3} above) and also the following condition is fulfilled:
\begin{description}
\item[(TRK)] Each of the two simplicial loops $l_1$ and $l_2$ on the boundary of $R$  fulfills   condition (TK) above.
\end{description}
The two integers  ${\mathbf p}:= {\mathbf p}(l_1) = {\mathbf p}(l_2)$ and
${\mathbf q}:= {\mathbf q}(l_1) = {\mathbf q}(l_2)$  will be called the winding numbers of $R$.
\end{definition}

\begin{definition} \label{def5.4} Let $R=(F_i)_{i \le n}$ be a closed simplicial ribbon in $\cK \times \bZ_N$.
 We say that $R$  is vertical
 iff $R$ is regular and moreover,  every 2-face $F_i \in \face_2(\cK \times \bZ_N)$ is ``parallel'' to $S^1$.
 In this case the three simplicial loops $l^{j}$ ($j=0,1,2$) in $q\cK \times \bZ_N$,
 associated to $R$ (cf. Sec. \ref{subsec3.3} above) will  be ``parallel'' to $S^1$ as well.
 More precisely,  for each $l^{j}$ the image of the projected simplicial loop  $l^{j}_{\Sigma}$ in $q\cK$
  will simply consist of  a single point  $\sigma^j \in \face_0(q\cK)$.\par

 Observe that every  vertical closed simplicial ribbon is a simplicial torus ribbon knot of standard type with ${\mathbf p}=0$ and ${\mathbf q}=\pm1$.
 If ${\mathbf q}=1$ we say that $R$ has standard orientation.
 \end{definition}

\subsection{Some notation}
\label{subsec5.2}

Recall that in Sec. \ref{sec2}, above we have fixed a scalar product $\langle \cdot, \cdot \rangle$
on $\ct$. Using this scalar product we will now make  the obvious identification $\ct \cong \ct^*$.

\begin{itemize}
\item $\cR \subset \ct^*$ will denote the set of real roots associated to ($\cG, \ct)$

\item $\Check{\cR}$ denotes the set of  real coroots, i.e. $\Check{\cR} := \{\Check{\alpha} \mid \alpha \in \cR\} \subset \ct$
       where $\Check{\alpha}: = \frac{ 2 \alpha}{\langle \alpha,  \alpha \rangle}$.

\item $\Lambda \subset \ct^*$ denotes the real weight lattice associated to $(\cG,\ct)$.

\item $\Gamma  \subset \ct$ will denote the lattice generated by the set of real coroots.

 \item  A Weyl alcove  is a connected component
 of the set

 \smallskip

 $\ct_{\reg} = \exp^{-1}(T_{reg}) = \ct \backslash \bigcup_{\alpha \in \cR, k \in \bZ} H_{\alpha,k}
 \quad \text{ where $H_{\alpha,k}:= \alpha^{-1}(k)$.} $

 \item $\cW$ will denote the Weyl group associated to $(\cG,\ct)$

\item  $\cW_{\aff}$ will denote the affine Weyl group associated to $(\cG,\ct)$, i.e.
  the group of isometries of $\ct \cong \ct^*$
   generated by the orthogonal reflections on the hyperplanes $H_{\alpha,k}$, $\alpha \in \cR$, $k \in \bZ$,
    defined  above.
    Equivalently, one can define $\cW_{\aff}$ as the group of isometries of $\ct \cong \ct^*$ generated by
      $\cW$ and  the translations associated to the coroot lattice $\Gamma$.
 For $\tau \in \cW_{\aff}$ we will denote the sign
  of $\tau$ by  $(-1)^{\tau}$.
\end{itemize}

Recall that in Sec. \ref{sec2} above we fixed $k \in \bN$.
Let us now also fix a Weyl chamber $\CW$.

\begin{itemize}
\item $\cR_+$ denotes the set of positive (real) roots
associated to $(\cG,\ct)$ and $\CW$.

 \item  $\Lambda_+$ denotes the set of dominant (real) weights
 associated to $(\cG,\ct)$ and $\CW$.

\item  $\rho$  denotes the half-sum of the positive (real) roots
\item $\theta$  denotes the unique  long (real) root in  $\overline{\CW}$.
\item We set $ \cg:= 1 + \langle \theta,\rho \rangle$ ($\cg$ is the dual Coxeter number of $\cG$)

\item For $\lambda \in  \Lambda_+$ let $\lambda^* \in
\Lambda_+$ denote the weight conjugated to $\lambda$ and
$\bar{\lambda} \in \Lambda_+ $ the weight conjugated to $\lambda$
``after applying a shift by $\rho$''. More precisely, $\bar{\lambda}$
is given by $\bar{\lambda} + \rho = (\lambda + \rho)^*$.

\item  We set
$ \Lambda_+^k :=  \{ \lambda \in \Lambda_+  \mid  \langle \lambda + \rho ,\theta \rangle < k \}
=  \{ \lambda \in \Lambda_+  \mid  \langle \lambda,\theta \rangle\leq k - \cg\}$.
\end{itemize}

\begin{remark}  \label{rm_shift_in_level} \rm In Sec. \ref{sec1} I mentioned that for a given
oriented closed 3-manifold $M$ the Reshetikhin-Turaev invariants associated to $(M,\cG_{\bC},q)$
are widely believed to be equivalent to  Witten's heuristic path integral expressions
   based on the Chern-Simons action function     associated to $(M,G,k)$  where
 $G$ is the simply connected, compact Lie group corresponding to the compact real form
  $\cG$ of $\cG_{\bC}$   and $k \in \bN$ is chosen suitably.
 It it commonly believed that this relationship between $q$ and $k$ is given by
 $$q = e^{2 \pi i/(k+\cg)}, \quad \quad k \in \bN$$
The appearance of  $k + \cg$ instead of $k$ (i.e. the replacement $k \to k + \cg$)
is the famous ``shift of the level'' $k$.
However,  several  authors have argued (cf., e.g., \cite{GMM2}) that the occurrence
(and magnitude) of such a shift in the level depends on the regularization procedure and
renormalization prescription which is used for making sense of the heuristic path integral.
Accordingly, it should not be surprising that there  are several papers (cf. the references in \cite{GMM2}) where the shift $k \to k + \cg$ is not observed and one is therefore led to
 the following relationship  between $q$ and\footnote{In view of the definition of the set $\Lambda_+^k$
   above it is clear that the situation $k \le \cg$ is not interesting} $k$:
$$q = e^{2 \pi i/k}, \quad \quad k \in \bN \text{ with\ } k > \cg$$
This is also the case in \cite{Ha7a,Ha7b} and the present paper\footnote{by contrast, in \cite{HaHa}
 a shift $k \to k + \cg$ was inserted by hand into several formulas. Accordingly, several definitions
in the present paper differ from the definitions in \cite{HaHa}}.
\end{remark}

Let $C$ and $S$  be the $\Lambda_+^k \times \Lambda_+^k$ matrices with complex entries    given by
\begin{subequations} \label{eq_def_C+S}
\begin{align}
  C_{\lambda \mu} & := \delta_{\lambda \bar{\mu}}, \\
\label{eq_def_S} S_{\lambda \mu} & :={i^{\# \cR _{+}}\over k^{\dim(\ct)/2}}  \frac{1}{|\Lambda/\Gamma|^{1/2}}
\sum_{\tau \in \cW} (-1)^{\tau} e^{- {2\pi i\over k} \langle \lambda + \rho , \tau \cdot
(\mu + \rho) \rangle }
\end{align}
\end{subequations}
for all $\lambda, \mu \in \Lambda_+^k$
where $\# \cR _{+}$ is the number of elements of $\cR _{+}$.
We have
\begin{equation} \label{eq_S2=C} S^2 = C
\end{equation}

It will be convenient to generalize the definition of $S_{\lambda \mu}$ to the situation
of general $\lambda, \mu \in \Lambda$ using again Eq. \eqref{eq_def_S}. \par

 Let $\theta_{\lambda}$ and $d_{\lambda}$ for  $\lambda \in \Lambda$
 be  given by\footnote{\label{ft_warning}For $r \in \bQ$ we will write $\theta_{\lambda}^r$ instead of
$e^{r \cdot\frac{\pi i}{k} \langle \lambda,\lambda+2\rho \rangle}$.
Note that this notation is somewhat dangerous since $\theta_{\lambda_1} = \theta_{\lambda_2}$
does, of course, in general not imply $\theta_{\lambda_1}^r = \theta_{\lambda_2}^r$}
\begin{subequations} \label{eq_def_th+d}
\begin{align}
\label{eq_def_th}
 \theta_{\lambda}  & := e^{\frac{\pi i}{k} \langle \lambda,\lambda+2\rho \rangle}\\
 \label{eq_def_d}
d_{\lambda} & := \frac{S_{\lambda 0}}{S_{00}}
\overset{(*)}{=} \prod_{\alpha \in \cR_+}
\frac{\sin(\frac{\pi}{k} \langle \lambda+\rho,\alpha \rangle) }{\sin(\frac{\pi}{k} \langle \rho,\alpha \rangle)}
\end{align}
\end{subequations}
where step $(*)$ follows, e.g., from part iii) in Theorem 1.7 in Chap. VI in \cite{Br_tD}.  \par

For every $\lambda \in \Lambda_+$ we denote by $\rho_{\lambda}$
the (up to equivalence) unique
irreducible, finite-dimensional, complex representation of $G$
with highest weight $\lambda$.
For every $\mu \in \Lambda$ we will denote by
$m_{\lambda}(\mu)$  the   multiplicity   of $\mu$
as a weight in $\rho_{\lambda}$.
It will be convenient to introduce $  \bar{m}_{\lambda}: \ct \to \bZ$  by
 \begin{equation}\label{eq_mbar_def}
  \bar{m}_{\lambda}(b) =
\begin{cases} m_{\lambda}(b) & \text{ if } b \in \Lambda\\
0 & \text{ otherwise }
\end{cases}
\end{equation}
Instead of  $\bar{m}_{\lambda}$ we will simply write  $m_{\lambda}$ in the following. \par

Finally, let us define $\ast: \cW_{\aff} \times \ct \to \ct$  by
\begin{equation} \label{eq_def_ast}
\tau \ast b = k \bigl( \tau \cdot \tfrac{1}{k} (b+\rho)\bigr) - \rho, \quad \quad
\text{for all $\tau \in \cW_{\aff}$ and $b \in \ct$}
\end{equation}
and set, for all $\lambda \in \Lambda_+$,  $\mu, \nu \in \Lambda$, $\mathbf p \in \bZ \backslash \{0\}$,
and  $\tau \in \cW_{\aff}$
\begin{equation} \label{eq_def_plethysm}
  m^{\mu \nu }_{\lambda, \mathbf p}(\tau)  :=  (-1)^{\tau} m_{\lambda}\bigl(\tfrac{1}{\mathbf p} (\mu - \tau \ast \nu)\bigr) \in \bZ
\end{equation}
and
\begin{equation} \label{eq_def_plethysm_org}
 M^{\mu \nu }_{\lambda, \mathbf p}  :=  \sum_{\tau \in \cW_{\aff}}
  m^{\mu \nu }_{\lambda, \mathbf p}(\tau) \in \bZ
\end{equation}

\subsection{The two main results}
\label{subsec5.3}

From now on we will always assume that $k > \cg$, cf. Remark \ref{rm_shift_in_level} above.

\begin{theorem} \label{theorem1} Let $L=(R_1)$ be a simplicial ribbon link in $\cK \times \bZ_N$
colored with $\rho_1$ where $R_1$ is simplicial torus ribbon knot
of standard type   with winding numbers ${\mathbf p} \in \bZ \backslash \{0\}$ and $\mathbf q \in \bZ$ (cf. (TRK) in Sec. \ref{subsec5.1}).
Assume that $\lambda_1 \in \Lambda_+^k$ where  $\lambda_1$ is the highest weight of $\rho_1$.
 Then  $\WLO^{disc}_{norm}(L)$
is well-defined and we have
\begin{equation} \label{eq_theorem1}
 \WLO^{disc}_{norm}(L)   = S_{00}^2
 \sum_{\eta_1, \eta_2 \in \Lambda_+^k} \sum_{\tau \in \cW_{\aff}} m^{\eta_1\eta_2}_{\lambda_1,\mathbf p}(\tau) \
d_{\eta_1} d_{\eta_2} \ \theta_{\eta_1}^{ \frac{{\mathbf q} }{{\mathbf p}}} \theta_{\tau \ast \eta_2}^{- \frac{{\mathbf q} }{{\mathbf p}}}
\end{equation}
\end{theorem}

The following  generalization of Theorem \ref{theorem1} will play a crucial role in Sec. \ref{subsec6.2} below.

\begin{theorem} \label{theorem2}
Let $L=(R_1, R_2)$ be  simplicial ribbon link in $\cK \times \bZ_N$ colored with $(\rho_1,\rho_2)$
where  $R_1$ is a simplicial torus ribbon knot of standard type
 with winding numbers ${\mathbf p} \in \bZ \backslash \{0\}$ and $\mathbf q \in \bZ$
and $R_2$ is vertical with standard orientation.
 Let us assume that $R_1$ winds around $R_2$ in the ``positive direction''\footnote{cf. Sec. \ref{subsec5.5}
 below for a precise definition}
 and  that $\lambda_1, \lambda_2 \in \Lambda_+^k$
 where $\lambda_1$ and  $\lambda_2$ are the highest weights of $\rho_1$ and $\rho_2$.
 Then  $\WLO^{disc}_{norm}(L)$ is well-defined and we have
\begin{equation}  \label{eq_theorem2}
 \WLO^{disc}_{norm}(L)   =  S_{00} \sum_{\eta_1, \eta_2 \in \Lambda_+^k} \sum_{\tau \in \cW_{\aff}} m^{\eta_1\eta_2}_{\lambda_1,\mathbf p}(\tau)  \
d_{\eta_1} S_{\lambda_2 \eta_2} \ \theta_{\eta_1}^{ \frac{{\mathbf q} }{{\mathbf p}}} \theta_{\tau \ast \eta_2}^{- \frac{{\mathbf q} }{{\mathbf p}}}
\end{equation}
\end{theorem}

\begin{remark} \label{rm_theorems} \rm
In the special case where $\mathbf p  = 1$ we have
$\theta_{\tau \ast \eta_2}^{- \frac{{\mathbf q} }{{\mathbf p}}} = \theta_{\tau \ast \eta_2}^{- {\mathbf q}} =
\theta_{\eta_2}^{- {\mathbf q}} $ (cf. Eq. \eqref{eq_theta_inv} below)
and Eq. \eqref{eq_theorem1} can be rewritten as
$$ \WLO^{disc}_{norm}(L)   = S_{00}^2
 \sum_{\eta_1, \eta_2 \in \Lambda_+^k}  M^{\eta_1\eta_2}_{\lambda_1,1}
d_{\eta_1} d_{\eta_2} \ \theta_{\eta_1}^{ {\mathbf q} } \theta_{\eta_2}^{- {\mathbf q} }$$
where  $M^{\mu \nu}_{\lambda,1}$ is as in Eq. \eqref{eq_def_plethysm_org} above.
(A totally analogous remark applies to Eq. \eqref{eq_theorem2}). But
\begin{equation}  \label{eq_rm5.8_1}
M^{\mu \nu}_{\lambda,1 } = \sum_{\tau \in \cW_{\aff}} (-1)^{\tau} m_{\lambda}\bigl(\mu - \tau \ast \nu\bigr) \overset{(*)}{= } N_{\lambda \nu}^{\mu}
\end{equation}
 where $N_{\lambda \nu}^{\mu}$, $\lambda, \mu, \nu \in \Lambda_+^k$,
  are the so-called fusion coefficients, see e.g., \cite{Saw} for the
 definition of $N_{\lambda \nu}^{\mu}$ and for the proof of the equality $(*)$ (called
  the ``quantum Racah formula'' in \cite{Saw}).\par
Eq. \eqref{eq_rm5.8_1} implies that the RHS of both theorems can be rewritten in terms of Turaev's shadow invariant
(or, equivalently, the Reshetikhin-Turaev invariant), cf. Remark \ref{rm_sec3.9} above.
Accordingly, in this case it is clear\footnote{Note, for example, that
if $\mathbf p  = 1$  then the simplicial ribbon link
 $L=(R_1)$  appearing in Theorem \ref{theorem1} will fulfill the analogue of
 Conditions (NCP)' and (NH)' in \cite{Ha7a} mentioned
in  Remark \ref{rm_sec3.9} above.}  that both theorems give the results expected in the literature, i.e. Conjecture \ref{conj0} below is indeed true if $\mathbf p  = 1$.
\end{remark}

 \subsection{Proof of Theorem \ref{theorem1}}
\label{subsec5.4}

Let $L=(R_1)$  where $R_1$ is a simplicial torus ribbon knot
of standard type  in $\cK \times \bZ_N$, colored with $\rho_1$ and with winding numbers ${\mathbf p} \in \bZ \backslash \{0\}$ and $\mathbf q \in \bZ$. (In the following we sometimes write $R$ instead of $R_1$.)
 Let $n \in \bN$  be the length of the three simplicial loops $l^j$, $j =0,1,2$, in $q\cK \times \bZ_N$ associated
to $R=R_1$, cf. Sec. \ref{subsec3.3} above.\par

The symbol $\sim$ will denote equality up to a multiplicative (non-zero) constant
which is allowed to depend on $G$, $k$, $\cK$ and $N$ but not
on the colored ribbon knot considered\footnote{in particular,  it will  depend neither on $\mathbf p$ nor on  $\mathbf q$ nor on $\rho_1$}.\par

Recall that, as mentioned in Sec. \ref{subsec3.10} above, until the end of ``Step 4'' below
we will work with the original definition of  $\WLO^{disc}_{rig}(L)$.
 Then we will explain how Steps 1--4 need to be modified if the new definition is used.
In Step 5--6 we then work with the new definition of $\WLO^{disc}_{rig}(L)$.

\subsubsection*{a) Step 1: Performing the $\int\nolimits_{\sim}  \cdots
 d\Check{\mu}^{\orth,disc}_{B}(\Check{A}^{\orth}) $
integration in Eq. \eqref{eq_def_WLOdisc}}

We will prove below that under the assumptions  on  $L=(R_1)$
 made above we have for every fixed $A^{\orth}_c \in  \cA^{\orth}_c(K)$ and $B \in \cB_{reg}(q\cK)$
\begin{equation}  \label{eq5.1}\int\nolimits_{\sim}    \Tr_{\rho_1}\bigl( \Hol^{disc}_{R_1}(\Check{A}^{\orth} + A^{\orth}_c,   B)\bigr) d\Check{\mu}^{\orth,disc}_{B}(\Check{A}^{\orth}) =   \Tr_{\rho_1}\bigl(
\Hol^{disc}_{R_1}(A^{\orth}_c,   B)\bigr)
\end{equation}
By taking into account that
  $ \prod_{x} 1^{\smooth}_{\ct_{reg}}(B(x))
\neq 0$ for $B \in \cB(\cK) \subset  \cB(q\cK)$ implies $B \in \cB_{reg}(q\cK)$
 we then obtain from  Eq. \eqref{eq5.1} and  Eq. \eqref{eq_def_WLOdisc}
\begin{multline} \label{eq5.14}
\WLO^{disc}_{rig}(L) =  \sum_{y \in I} \int\nolimits_{\sim} \biggl\{
 \bigl( \prod_{x} 1^{\smooth}_{\ct_{reg}}(B(x)) \bigr)
    \Tr_{\rho_1}\bigl(\Hol^{disc}_{R_1}( A^{\orth}_c, B)\bigr) \Det^{disc}(B)  \biggr\}\\
  \times \exp\bigl( -  2 \pi i k   \langle y, B(\sigma_0) \rangle \bigr)
  \exp(i  S^{disc}_{CS}(A^{\orth}_c,B))    (DA^{\orth}_c \otimes  DB)
\end{multline}

{\it Proof of Eq. \eqref{eq5.1}:}  Let  $A^{\orth}_c \in \cA^{\orth}_c(K)$ and   $B \in \cB_{reg}(q\cK)$ be fixed.
 We will prove Eq. \eqref{eq5.1} by applying Proposition \ref{prop3.3}
 to the special situation where
 \begin{itemize}
 \item   $V= \Check{\cA}^{\orth}(K)$ \quad and \quad $d\mu=  d\Check{\mu}^{\orth,disc}_{B}$,

\item  $(Y^{a}_{k})_{k \le n,a \le \dim(\cG)}$  is the
 family   of maps $Y^{a}_{k}:  \Check{\cA}^{\orth}(K) \to \bR$ given by
\begin{multline}  \label{eq5.4}
Y^{a}_k(\Check{A}^{\orth})  =  \bigl\langle e_a,
\sum_{j=0}^2 w(j) \bigl( \Check{A}^{\orth}(\start l^{j(k)}_{S^1})(l^{j(k)}_{\Sigma})
+  A^{\orth}_c(l^{j(k)}_{\Sigma})
+ B(\start l^{j(k)}_{\Sigma})  \cdot dt^{(N)}(l^{j(k)}_{S^1}) \bigr) \bigr\rangle
\end{multline}
for all $\Check{A}^{\orth} \in \Check{\cA}^{\orth}(K)$,

\item $\Phi: \bR^{n \times \dim(\cG)} \to \bC$  is given by
\begin{equation}  \label{eq5.5} \Phi((x^{a}_k)_{k,a}) =    \Tr_{\rho_1}(\prod_{k=1}^n \exp(\sum_{a=1}^{\dim(\cG)} e_a x^{a}_k))
\quad \quad \text{for all $(x^{a}_k)_{k,a} \in \bR^{n \times \dim(\cG)}$}
\end{equation}
Here  $(e_a)_{a \le \dim(\cG)}$ is an arbitrary but fixed
$\langle \cdot, \cdot  \rangle$-orthonormal basis of $\cG$.
\end{itemize}
Note that
$d\Check{\mu}^{\orth,disc}_{B}$ is a  well-defined normalized,  non-degenerate, centered oscillatory Gauss-type measure.    (Since by assumption $B \in \cB_{reg}(q\cK)$ this follows from
   the remarks in Sec. \ref{subsec3.8}).
Moreover, we have
\begin{equation}
 \Tr_{\rho_1}\bigl( \Hol^{disc}_{R_1}(\Check{A}^{\orth} + A^{\orth}_c, B)\bigr)   =   \Tr_{\rho_1}(\prod_{k=1}^n \exp(\sum_{a=1}^{\dim(\cG)} e_a Y^{a}_k(\Check{A}^{\orth}))) \quad \forall \Check{A}^{\orth} \in \Check{\cA}^{\orth}(K)
 \end{equation}
Finally, since   $ d\Check{\mu}^{\orth,disc}_{B}$ is centered
  and  normalized  we have for every  $k \le n$ and $a \le \dim(\cG)$
\begin{equation}  \label{eq5.7}
\int\nolimits_{\sim}  Y^{a}_k \  d\Check{\mu}^{\orth,disc}_{B} = Y^{a}_k(0)
\end{equation}
  Consequently, we obtain
\begin{multline}  \label{eq5.11}  \int\nolimits_{\sim}
   \Tr_{\rho_1}\bigl( \Hol^{disc}_{R_1}(\Check{A}^{\orth} + A^{\orth}_c,
  B)\bigr)   d\Check{\mu}^{\orth,disc}_{B}(\Check{A}^{\orth})\\
= \int\nolimits_{\sim}    \Tr_{\rho_1}(\prod_k \exp(\sum_a e_a Y^{a}_k))
d\Check{\mu}^{\orth,disc}_{B} =  \int\nolimits_{\sim}  \Phi((Y^{a}_k)_{k,a}) \ d\Check{\mu}^{\orth,disc}_{B} \overset{(*)}{=}
 \Phi((\int\nolimits_{\sim} Y^{a}_k \ d\Check{\mu}^{\orth,disc}_{B})_{k,a}) \\
=  \Phi((Y^{a}_k(0))_{k,a}) =   \Tr_{\rho_1}( \prod_k
\exp( \sum_a e_a Y^{a}_k(0))) = \Tr_{\rho_1}\bigl(
\Hol^{disc}_{R_1}(A^{\orth}_c,   B)\bigr)
\end{multline}
where  step $(*)$ follows from Proposition \ref{prop3.3}   above.
The following remarks show that the assumptions of Proposition \ref{prop3.3} are indeed fulfilled.
\begin{enumerate}

\item We have $\Phi \in \cP_{exp}(\bR^{n \times \dim(\cG)})$. In order to see this note  first that
\begin{multline}  \label{eq5.6}
\Phi((x^{a}_k)_{k,a}) =    \Tr_{\rho_1}(\prod_k
\exp(\sum_a e_a x^{a}_k))
  =  \Tr_{\End(V_1)} \bigl( \prod_k  \rho_1(\exp( \sum_a e_a x^{a}_k )) \bigr) \\
 =    \Tr_{\End(V_1)} \bigl( \prod_k \exp_{\End(V_1)}( \sum_a
 (\rho_1)_* (e_a) x^{a}_k ) \bigr)
 \end{multline}
where $V_1$ is the representation space of $\rho_1$,   $\exp_{\End(V_1)}$ is the  exponential map of the associative algebra $\End(V_1)$, and  $(\rho_1)_* : \cG \to \gl(V_1)$  is the Lie algebra representation
 induced by $\rho_1: G \to \GL(V_1)$. Without loss of generality we can assume that
 $V_1 = \bC^d$ where $d=\dim(V_1)$. From Definition \ref{def3.3} it now easily follows that we have indeed $\Phi \in \cP_{exp}(\bR^{n \times \dim(\cG)})$.

\item For all  $k,k' \le n$, $a,a' \le \dim(\cG)$ we have
\begin{equation}  \label{eq5.8} \int\nolimits_{\sim}  Y^{a}_k Y^{a'}_{k'} \  d\Check{\mu}^{\orth,disc}_{B}
 =   \int\nolimits_{\sim}  Y^{a}_k \  d\Check{\mu}^{\orth,disc}_{B}
 \int\nolimits_{\sim} Y^{a'}_{k'} \ d\Check{\mu}^{\orth,disc}_{B}
 \end{equation}
 This follows from Eq. \eqref{eq5.7} above and
\begin{align}  \label{eq5.9}
& \int\nolimits_{\sim}  (Y^{a}_k - Y^{a}_k(0))
 (Y^{a'}_{k'} - Y^{a'}_{k'}(0)) d\Check{\mu}^{\orth,disc}_{B} =
 \int\nolimits_{\sim}  \ll \cdot , f \gg \ll \cdot ,
 f' \gg d\Check{\mu}^{\orth,disc}_{B} \nonumber \\
& \quad \overset{(*)}{\sim}  \quad
\ll f, \bigl(\star_K L^{(N)}(B)_{| \Check{\cA}^{\orth}(K)} \bigr)^{-1}
f' \gg \overset{(**)}{=} 0
\end{align}
where $\ll \cdot, \cdot \gg$ is the scalar product on\footnote{which, according to
Convention \ref{conv_EucSpaces} above, is the the scalar product  induced by
 $\ll  \cdot , \cdot \gg_{\cA^{\orth}(q\cK)}$} $\Check{\cA}^{\orth}(K)$
and, for given $k,k' \le n$, $a,a' \le \dim(\cG)$,
 $f, f' \in \Check{\cA}^{\orth}(K)$ are chosen  such that $Y^{a}_k(\Check{A}^{\orth}) - Y^{a}_k(0)
=   \ll  \Check{A}^{\orth} ,f \gg$
and $Y^{a'}_{k'}(\Check{A}^{\orth}) - Y^{a'}_{k'}(0)
=   \ll  \Check{A}^{\orth} ,f' \gg$ for all  $\Check{A}^{\orth} \in \Check{\cA}^{\orth}(K)$.
Here in step $(*)$  we have used  Proposition \ref{obs1}  (cf. also Remark \ref{rm_Sec4.3}),
and in step $(**)$ we have used that  for all non-trivial
$l^{j(k)}_{\Sigma}$ and $l^{j'(k')}_{\Sigma}$, $k, k' \le n, j, j' \in \{0,1,2\}$,
 appearing on the RHS of Eq. \eqref{eq5.4} we have
\begin{equation} \label{eq_Ende_Step1}
 \star_K \pi(l^{j(k)}_{\Sigma}) \neq \pm \pi(l^{j'(k')}_{\Sigma})
 \end{equation}
where  $\pi: C^1(q\cK,\bR) \to  C^1(K,\bR)$
is the real analogue of the orthogonal projection given in Eq. \eqref{eq_pi_proj}.
Eq. \eqref{eq_Ende_Step1} follows from  Eq. \eqref{eq_Hodge_star_concrete} in Sec. \ref{subsec3.2} above
and from our  assumption that  $R_1$ is a simplicial torus ribbon knot of standard type  in $\cK \times \bZ_N$.
\end{enumerate}

\subsubsection*{b) Step 2: Performing the $\int\nolimits_{\sim}  \cdots
 \exp(i  S^{disc}_{CS}(A^{\orth}_c,B))    (DA^{\orth}_c \otimes  DB)$-integration in \eqref{eq5.14}}

Note that the  remaining fields $A^{\orth}_c$ and  $B$  in Eq.  \eqref{eq5.14}
take values in the Abelian Lie algebra  $\ct$.
For  fixed  $A^{\orth}_c$ and $B$ we can therefore rewrite  $\Hol^{disc}_{R_1}(A^{\orth}_c,B)$  as
\begin{equation}  \label{eq5.16} \Hol^{disc}_{R_1}( A^{\orth}_c, B)
  = \exp\bigl(  \Psi(B) +   \sum_{k=1}^n  \sum_{j=0}^2 w(j) A^{\orth}_c(l^{j(k)}_{\Sigma})   \bigr)
\end{equation}
where we  have set
\begin{equation} \label{eq5.17} \Psi(B)  :=  \sum_k  \sum_{j=0}^2 w(j)
 B(\start l^{j(k)}_{\Sigma})  dt^{(N)}(l^{j(k)}_{S^1}) \quad \in \ct
\end{equation}

\begin{observation} \label{obs_red_loops} \rm From the assumption that $R = R_1$ is a simplicial torus ribbon knot in $\cK \times \bZ_N$
of standard type with first winding number ${\mathbf p} \neq 0$
 it follows that
 for each $j=0,1,2$ there is a simplicial
loop ${\mathfrak l}^j_{\Sigma} = ({\mathbf x}^{j(k)}_{\Sigma})_{0 \le k \le \mathbf n}$, $\mathbf n \in \bN$,
 in $q\cK$ which, considered as a continuous map  $S^1 \to S^2$, is an embedding and fulfills
 (cf. Convention \ref{conv_loop_pic})
$$ \sum_{k=1}^n  l^{j(k)}_{\Sigma} =  {\mathbf p}  \sum_{k=1}^{\mathbf n}  {\mathfrak l}^{j(k)}_{\Sigma} $$
\end{observation}

Since  $\rho_1 = \rho_{\lambda_1}$ it follows from the definitions
in Sec. \ref{subsec5.2} above that
\begin{equation} \label{eq5.18}
\Tr_{\rho_1}(\exp(b))  = \sum_{\alpha \in
\Lambda} m_{\lambda_1}(\alpha) e^{2 \pi i \langle \alpha, b \rangle} \quad \forall b \in \ct
\end{equation}
 Combining Eqs. \eqref{eq5.16} -- \eqref{eq5.18}  with Observation \ref{obs_red_loops} we obtain
\begin{align} \label{eq5.19}
&  \Tr_{\rho_1}\bigl( \Hol^{disc}_{R_1}( A^{\orth}_c,   B) \bigr) \nonumber \\
 & =   \sum_{\alpha \in \Lambda}  m_{\lambda_1}(\alpha)
\bigl(  \exp(2 \pi i \langle \alpha, \Psi(B) \rangle )  \bigr)
 \exp\bigl(  2 \pi i  \ll  A^{\orth}_c ,   \alpha  \, {\mathbf p} \, (w{\mathfrak l})_{\Sigma} \gg_{\cA^{\orth}(q\cK)}  \bigr)
\end{align}
where  we have set
\begin{equation} \label{eq_mathfrak_l_def}
 (w{\mathfrak l})_{\Sigma} :=  \sum_{k=1}^{\mathbf n}  \sum_{j=0}^2 w(j)  {\mathfrak l}^{j(k)}_{\Sigma}    \in C_1(q\cK)
 \end{equation}

Let us now introduce  for each     $y \in I$,
$\alpha \in \Lambda$ and $B \in \cB(\cK) \subset \cB(q\cK)$:
\begin{multline}  \label{eq5.22}
F_{\alpha,y}(B)
 :=  \bigl( \prod_{x}
1^{\smooth}_{\ct_{reg}}(B(x)) \bigr)
 \bigl(  \exp(2 \pi i \langle \alpha, \Psi(B) \rangle )  \bigr) \Det^{disc}(B)  \exp\bigl( - 2 \pi i k    \langle y, B(\sigma_0)  \rangle \bigr)
 \end{multline}
After doing so we can rewrite  Eq. \eqref{eq5.14} as
\begin{multline} \label{eq5.21} \WLO^{disc}_{rig}(L)
=  \sum_{\alpha \in \Lambda}   m_{\lambda_1}(\alpha)
 \sum_{y \in I}\\
 \times \int_{\sim} F_{\alpha,y}(B)
  \exp\bigl(  2 \pi i  \ll  A^{\orth}_c ,   \alpha {\mathbf p} (w{\mathfrak l})_{\Sigma}  \gg_{\cA^{\orth}(q\cK)}  \bigr) \exp(i  S^{disc}_{CS}(A^{\orth}_c,B))    (DA^{\orth}_c \otimes  DB)
\end{multline}

For each fixed  $y \in I$ and  $\alpha \in \Lambda$
we will now  evaluate the corresponding integral in Eq. \eqref{eq5.21} by
applying Proposition \ref{prop3.5} above in the special situation where
\begin{itemize}
\item $V  = \cA^{\orth}_c(K) \oplus  \cB(\cK)$. For $V$ we use the decomposition
   $V = V_0 \oplus V_1 \oplus V_2$ given by
\begin{align*}
V_0 & := \cB_c(q\cK) \oplus (\Image(\star_K \circ \pi \circ d_{q\cK}) )^{\orth}\\
V_1 & := (\ker(\star_K \circ
\pi \circ d_{q\cK}))^{\orth}  = (\ker(\pi \circ d_{q\cK}))^{\orth} \overset{(*)}{=} (\cB_c(q\cK))^{\orth} \subset \cB(\cK)\\
V_2 & := \Image(\star_K \circ \pi \circ d_{q\cK}) \subset  \cA^{\orth}_c(K)
\end{align*}
where   $\pi:\cA_{\Sigma,\ct}(q\cK) \to  \cA_{\Sigma,\ct}(K) \cong \cA^{\orth}_c(K)$ is the orthogonal
 projection of Eq. \eqref{eq_pi_proj},  $d_{q\cK}$ is a short notation for  $(d_{q\cK})_{| \cB(\cK)}$,
 and $(\cdot)^{\orth}$ denotes the orthogonal complement in $\cA^{\orth}_c(K)$ and $\cB(\cK)$, respectively.
Note that step $(*)$ follows from  Eq. \eqref{eq_obs1}.

\item $d\mu  = d\nu^{disc} := \tfrac{1}{Z^{disc}} \exp(iS^{disc}_{CS}(A^{\orth}_c,B))
(DA^{\orth}_c  \otimes DB)$ where
$$Z^{disc}:= \int_{\sim} \exp(i  S^{disc}_{CS}(A^{\orth}_c,B))    (DA^{\orth}_c \otimes  DB)$$

\item  $F = F_{\alpha,y} \circ p$
where $p: V_0 \oplus V_1 =  \cB(\cK) \oplus (V_2)^{\orth} \to \cB(\cK)$
is the canonical projection.

\item $v =  2 \pi    \alpha   {\mathbf p} (w{\mathfrak l})_{\Sigma}$
\end{itemize}

\noindent Before we continue we need to verify  that  the assumptions of Proposition \ref{prop3.5}  above are indeed fulfilled.

\begin{enumerate}
\item[1.]      $d\nu^{disc}$ is  a normalized, centered oscillatory Gauss type measure on  $\cA^{\orth}_c(K) \oplus \cB(q\cK)$. In order to see this we rewrite $d\nu^{disc}$  as\footnote{recall Eq. \eqref{eq_SCS_expl_discc}
 and observe that  $ \ll  \star_K A^{\orth}_c,  d_{q\cK}  B \gg_{\cA^{\orth}(q\cK)}
  =  \ll  \star_K A^{\orth}_c,  \pi( d_{q\cK}  B) \gg_{\cA^{\orth}(q\cK)}
 = - \ll   A^{\orth}_c, \star_K ( \pi( d_{q\cK}  B)) \gg_{\cA^{\orth}(q\cK)}$}
 $$ d\nu^{disc}   = \tfrac{1}{Z^{disc}} \exp(  i
 \ll  A^{\orth}_c,  -   2 \pi  k (\star_K \circ \pi \circ d_{q\cK}) B \gg_{\cA^{\orth}(q\cK)}  ) (DA^{\orth}_c \otimes DB)$$
 Accordingly,  $d\nu^{disc}$ has the form as in Proposition \ref{prop3.5}
    with $V_0$, $V_1$, and $V_2$ as above and where
     $M: V_1 \to V_2$  is the well-defined   linear isomorphism given by
   \begin{equation} \label{eq_def_M} M  =  -   2 \pi  k (\star_K \circ \pi \circ d_{q\cK})_{| V_1}
   \end{equation}

\item[2.]  $F = F_{\alpha,y} \circ p$ is a bounded and  uniformly continuous function.

 \item[3.] $v$ is an element of $V_2$. In order to prove this
  we introduce a linear map
  $$m_{\bR}:  C^0(q\cK,\bR) \to  C^1(K,\bR)$$
   by  $m_{\bR} := \star_K \circ \pi \circ  d_{q\cK}$
  where $\pi: C^1(q\cK,\bR) \to C^1(K,\bR)$, $d_{q\cK}:C^0(q\cK,\bR) \to  C^1(q\cK,\bR)$,
    and  $\star_{K}:C^1(K,\bR) \to  C^1(K,\bR)$
 are the ``real analogues'' of the three maps appearing on the RHS of Eq. \eqref{eq_def_M} above.
    From Lemma \ref{lem2_pre} in Step 3 below it follows that
 there is a unique\footnote{here $C^0(\cK,\bR)$ is embedded into $C^0(q\cK,\bR)$
 in the same way as $\cB(\cK)$ is embedded into $\cB(q\cK)$, cf. Sec. \ref{subsec3.0}}
  $f \in C^0(\cK,\bR) \subset C^0(q\cK,\bR)$
 such that
\begin{equation} \label{eq5.28}
  (w{\mathfrak l})_{\Sigma} =   m_{\bR} \cdot f
 \end{equation}
 and also
 \begin{equation} \label{eq5.28''} \sum_{x \in \face_0(q\cK)} f(x) = 0
\end{equation}
  That $v \in V_2$ now follows from
\begin{equation} \label{eq_def_v} v =  2 \pi    \alpha  {\mathbf p} (w{\mathfrak l})_{\Sigma} =
   2 \pi    \alpha  {\mathbf p} ( m_{\bR} \cdot f)   = (\star_K \circ \pi \circ d_{q\cK}) \cdot \bigl( 2 \pi    \alpha {\mathbf p} f\bigr)
   \end{equation}
\end{enumerate}

\noindent
By applying Proposition \ref{prop3.5}  we now obtain
\begin{align} \label{eq5.33} & \tfrac{1}{ Z^{disc}} \int_{\sim} F_{\alpha,y}(B)
  \exp\bigl(  2 \pi i \ll  A^{\orth}_c ,   \alpha   {\mathbf p} (w{\mathfrak l})_{\Sigma}  \gg_{\cA^{\orth}(q\cK)}  \bigr) \exp(iS^{disc}_{CS}(A^{\orth}_c,B)) (DA^{\orth}_c  \otimes DB)   \nonumber \\
  & = \int_{\sim} F(B)
  \exp\bigl(   i \ll  A^{\orth}_c , v  \gg_{\cA^{\orth}(q\cK)}  \bigr)
    d\nu^{disc}(A^{\orth}_c,B)   \nonumber \\
 & \quad \quad \sim     \int^{\sim}_{V_0}
  F( x_0 - M^{-1} v)  dx_0  \overset{(*)}{\sim}
   \int^{\sim}_{\ct}    F_{\alpha,y}(b - M^{-1} v) db
\end{align}
Here $\int^{\sim} \cdots dx_0$ and
$\int^{\sim} \cdots db$ are improper integrals
 defined  according to Definition \ref{conv3.1}  above.
(Remark \ref{rm_last_sec3} above  and a ``periodicity argument'', which will be given in  Step 4 below,
imply that these improper integrals  are well-defined.)
Step $(*)$ above follows because $F( x_0 - M^{-1} v)$ depends only on the $\cB_{c}(q\cK)$-component of
$x_0 \in V_0 =  \cB_{c}(q\cK) \oplus (V_2)^{\orth} \cong \ct  \oplus (V_2)^{\orth}$.

\smallskip

According to Eq. \eqref{eq5.28''} we have $\alpha  f  \in V_1 = (\cB_c(q\cK))^{\orth}$, so Eq. \eqref{eq_def_v}  implies that
\begin{equation} \label{eq5.34} M^{-1} v =    - \tfrac{1}{k}  \alpha  {\mathbf p} f
\end{equation}
\noindent  Combining  Eqs. \eqref{eq5.22}, \eqref{eq5.21}, \eqref{eq5.33}, and  \eqref{eq5.34}
 we obtain
\begin{multline} \label{eq5.35}
\WLO^{disc}_{rig}(L) \sim  \sum_{\alpha \in \Lambda}   m_{\lambda_1}(\alpha)
 \sum_{y \in I} \\
 \times      \int^{\sim}_{\ct} db \ \biggl[ \exp\bigl( - 2 \pi i k    \langle y, B(\sigma_0)  \rangle \bigr)  \bigl( \prod_{x}
1^{\smooth}_{\ct_{reg}}(B(x)) \bigr) \\
\times  \bigl(  \exp(2 \pi i \langle \alpha, \Psi(B) \rangle )  \bigr) \Det^{disc}(B)
   \biggr]_{| B  = b + \tfrac{1}{k}     \alpha   {\mathbf p} f }
    \end{multline}

 It is not difficult to see that
Eq. \eqref{eq5.35} also holds   if we redefine $f$
using  the  normalization condition
 \begin{equation} \label{eq5.28'}  f(\sigma_0)=0
\end{equation}
instead of the normalization condition \eqref{eq5.28''} above\footnote{this follows from Eq. \eqref{eq1_lastrmsec3}
 in Remark \ref{rm_last_sec3} above
 and the periodicity properties of the integrand in
 $\int^{\sim}_{\ct} \cdots db$ (for fixed $y$ and $\alpha$),
 cf.  Step 4 below}.

\subsubsection*{c) Step 3: Rewriting Eq. \eqref{eq5.35}}

\begin{lemma} \label{lem2_pre}
Assume that $ (w{\mathfrak l})_{\Sigma}  \in C_1(q\cK) \cong C^1(q\cK,\bR)$ is as in Eq. \eqref{eq_mathfrak_l_def} above. Then  there is a $f \in C^0(\cK,\bR) \subset  C^0(q\cK,\bR)$, unique up to an additive constant,
such that $ (w{\mathfrak l})_{\Sigma}  = m_{\bR} \cdot f $.
\end{lemma}

\begin{proof}   That $f$ is unique up to an additive constant follows by combining the definition of $m_{\bR}$
 with   the real analogue  of Eq. \eqref{eq_obs1} in Sec. \ref{subsec3.0} above and
 the fact that  $\star_{K}:C^1(K,\bR) \to  C^1(K,\bR)$   is a bijection.\par

In order to show the existence of $f$ we observe first that the assumption that
$R = R_1$ is a simplicial torus ribbon knot of standard type  in $\cK \times \bZ_N$ implies that
$\Sigma \backslash (\arc({\mathfrak l}^1_{\Sigma}) \cup \arc({\mathfrak l}^2_{\Sigma}))$
has three connected components.
Let us denote these three connected components by
 $Z_0, Z_1, Z_2$. The enumeration is chosen such that
   $Z_0$ is the connected component containing $\arc({\mathfrak l}^0_{\Sigma})$
 and $Z_1$ is the other connected component having $ \arc({\mathfrak l}^1_{\Sigma})$ on its boundary.
Let $f: \face_0(\cK) \to \bR$ be given by
\begin{equation} \label{eq_def_f}
f(\sigma):=
 \begin{cases} c & \text{ if } \sigma \in \overline{Z_1} \\
 c \pm \tfrac{1}{2} & \text{ if } \sigma \in Z_0 \\
c \pm 1 & \text{ if } \sigma \in \overline{Z_2} \\
\end{cases}
\end{equation}
for all $\sigma \in \face_0(\cK)$
where   $\overline{Z_i}$ is the closure of  $Z_i$ in $\Sigma$,
$c \in \bR$ is an arbitrary constant,
and  the sign  $\pm$ is ``$+$'' if for any $k \le \mathbf n$
the edge $\star_K ( \pi ( {\mathfrak l}^{0(k)}_{\Sigma}))$
points from $Z_2$ to $Z_1$ and ``$-$'' otherwise.
In order to conclude the proof of   Lemma \ref{lem2_pre}
we have to show that
\begin{equation} \label{eq_lemma3_crucial}
 (w{\mathfrak l})_{\Sigma}  = m_{\bR} \cdot f = \star_K(\pi(d_{q\cK}f))
 \end{equation}

 Let $S $  be the set of those $e \in \face_1(q\cK)$ which are contained in $Z_0$
but do  not lie on $\arc({\mathfrak l}^0_{\Sigma})$.
Now observe that $(d_{q\cK} f)(e) = 0$   if $e \notin S$.
On the other hand if $e \in S$ we have
$$(d_{q\cK} f)(e) =  \pm \tfrac{1}{2} \sgn(e) $$
where the sign $\pm$ is the same as in Eq. \eqref{eq_def_f} and
 where $\sgn(e) = 1$ if the (oriented) edge  $e$   ``points from'' the region $Z_1$ to the region $Z_2$ and $\sgn(e) = -1$ otherwise.  \par

From the definition of $q\cK$ it follows that every $e \in \face_1(q\cK)$ has
exactly one endpoint in $ \face_0(K_1) \cup \face_0(K_2) \subset \face_0(q\cK)$
and one endpoint in  $(\face_0(K_1) \cup \face_0(K_2))^c := \face_0(q\cK) \backslash (\face_0(K_1) \cup \face_0(K_2))$.
On the other hand, if $e \in S$ then both endpoints of $e$ will be in $\bigcup_{j=0}^2 \arc({\mathfrak l}^{j}_{\Sigma})$.
Accordingly, for every  $e \in S$ we can distinguish between exactly three types:
 $$\text{ $e$ is of type $j$ ($j=0,1,2$) if $e$ has an  endpoint in $\arc({\mathfrak l}^{j}_{\Sigma}) \cap (\face_0(K_1) \cup \face_0(K_2))^c$ }$$
 Next we observe that  for each fixed $e \in S$  of type $j$ there exist exactly two indices $k \le \mathbf n$
   such that
\begin{equation} \label{eq_lemma_step3}  \star_K (\pi(\sgn(e) \cdot e)) =  \pi({\mathfrak l}^{j(k)}_{\Sigma})
\end{equation}
Moreover, if we let $e \in S$ vary then for  $j =1,2$  every $k \le \mathbf n$ arises exactly once on the RHS
of Eq. \eqref{eq_lemma_step3}
and for $j=0$ every $k \le \mathbf n$ arises exactly twice
Taking this into account we arrive at
\begin{multline*} \star_K(\pi(d_{q\cK}f)) = \sum_{e \in S} \tfrac{1}{2}  \star_K (\pi(\sgn(e) \cdot e))
=  \sum_{k=1}^{\mathbf n} \bigl( \tfrac{1}{4} \pi({\mathfrak l}^{1(k)}_{\Sigma}) + \tfrac{1}{4} \pi({\mathfrak l}^{2(k)}_{\Sigma})  + 2 \tfrac{1}{4} \pi({\mathfrak l}^{0(k)}_{\Sigma}) \bigr) \\
 =  \sum_{k=1}^{\mathbf n} \bigl( \tfrac{1}{4} {\mathfrak l}^{1(k)}_{\Sigma} + \tfrac{1}{4} {\mathfrak l}^{2(k)}_{\Sigma}  +  \tfrac{1}{2} {\mathfrak l}^{0(k)}_{\Sigma} \bigr) =
  \sum_{k=1}^{\mathbf n} \sum_{j=0}^2 w(j)  {\mathfrak l}^{j(k)}_{\Sigma} =
  (w{\mathfrak l})_{\Sigma}
\end{multline*}
\end{proof}

In the following we assume that  $B \in \cB(\cK) \subset \cB(q\cK) = C^0(q\cK,\ct) $  is of the form
\begin{equation} \label{eq5.37}
  B= b + \tfrac{1}{k}   \alpha {\mathbf p} f
 \end{equation}
with $b \in \ct$, $\alpha \in \Lambda$ and where $f$ is given by
Eq. \eqref{eq5.28} above in combination with \eqref{eq5.28'}.

 \begin{observation} \label{obs3}  Let $Z_0, Z_1, Z_2$ be as in the proof of Lemma \ref{lem2_pre} above.
 Then the restriction of $B: \face_0(q\cK) \to \ct$
to $\face_0(q\cK) \cap Z_i$ is constant for each $i = 0,1,2$.
Moreover, if $i=1,2$ then also the restriction of
$B: \face_0(q\cK) \to \ct$
to $\face_0(q\cK) \cap \overline{Z_i}$ is constant.
 \end{observation}

 We  set  $B(Z_i) :=  B(x)$ for any $x \in Z_i \cap \face_0(q\cK)$, which
 -- according to Observation \ref{obs3} -- is well-defined.

\begin{lemma} \label{lem2} For every $B \in \cB(\cK)$
of the form \eqref{eq5.37} and fulfilling  $\prod_{x}
1^{\smooth}_{\ct_{reg}}(B(x)) \neq 0$   we have\footnote{Note that the expression on the RHS of Eq. \eqref{eq5.70} is well-defined since by assumption $\prod_{x} 1^{\smooth}_{\ct_{reg}}(B(x)) \neq 0$, which implies that
$B(x) \in \ct_{reg}$ and therefore  $\det\bigl(1_{\ck}-\exp(\ad(B(x)))_{|\ck}\bigr) \neq 0$
(recall the definition of $\ct_{reg}$ and cf. Eq \eqref{eq_det_in_terms_of_roots} above)}
\begin{align} \label{eq5.70}
\Det^{disc}(B) &  = \prod_{i=0}^2   \det\nolimits^{1/2}\bigl(1_{\ck}-\exp(\ad(B(Z_i)))_{|\ck}\bigr)^{\chi(Z_i)}
 \end{align}
 where $Z_i$, $i=0,1,2$ are as in the proof of Lemma \ref{lem2_pre}
and $\chi(Z_i)$  is the Euler characteristic of $Z_i$.
 \end{lemma}
\begin{proof} We set $\face_p(\overline{Z_i}) :=
 \{F \in \face_p(\cK) \mid \bar{F} \in  \overline{Z_i}\} = \{F \in \face_p(\cK) \mid F \subset \overline{Z_i}\}$,
 for $i=0,1,2$, and $\face_p(Z_0) := \{F \in \face_p(\cK) \mid \bar{F} \in  Z_0\}$.
Since $\Sigma = S^2$ is the disjoint union of the three sets   $Z_0$, $\overline{Z_1}$, and $\overline{Z_2}$
we obtain from  Eq. \eqref{eq_def_DetFPdisc}
      in Sec. \ref{subsec3.4} and Observation  \ref{obs3}
\begin{align}  \label{eq_5.44b}
 \Det^{disc}(B) & = \prod_{p=0}^2 \biggl[ \biggl( \prod_{i=1}^2
 \det\nolimits^{1/2}\bigl(1_{\ck}-\exp(\ad(B(Z_i)))_{|\ck}\bigr)^{(-1)^p} \biggr)^{\#
\face_p(\overline{Z_i})} \nonumber \\
& \quad \quad \quad \times  \biggl(\det\nolimits^{1/2}\bigl(1_{\ck}-\exp(\ad(B(Z_0)))_{|\ck}\bigr)^{(-1)^p} \biggr)^{\# \face_p(Z_0)} \biggr]  \nonumber \\
& \overset{(*)}{=} \prod_{i=0}^2  \det\nolimits^{1/2}\bigl(1_{\ck}-\exp(\ad(B(Z_i)))_{|\ck}\bigr)^{\sum_{p=0}^2 (-1)^p \# \face_p(\overline{Z_i})}
\end{align}
where in step $(*)$ we have used that
$\sum_{p=0}^2 (-1)^p \# \face_p(Z_0) = \sum_{p=0}^2 (-1)^p \# \face_p(\overline{Z_0})$.
 The assertion of the lemma now follows by combining Eq. \eqref{eq_5.44b} with
$$ \chi(Z_i)= \chi(\overline{Z_i}) = \sum_{p=0}^2 (-1)^p \# \face_p(\overline{Z_i}), \quad i = 0,1,2 $$
where we have used that   $\overline{Z_i}$ is a  subcomplex of the CW complex $\cK$.
\end{proof}

Taking into account that $\arc(l^j_{\Sigma}) \subset \overline{Z_j}$ for $j=1,2$ and
  $\arc(l^{0}_{\Sigma}) \subset Z_0$   we see that
Observation  \ref{obs3} implies that
\begin{equation} \label{eq5.38} B(Z_j) =
B(\start l^{j(k)}_{\Sigma}) \quad \forall k \le n
\end{equation}
for every $B$ of the form in Eq. \eqref{eq5.37}. Moreover, for such $B$ we have
 \begin{equation} B(Z_0) = \tfrac{1}{2} ( B(Z_1) + B(Z_2))
\end{equation}
 Finally, note that for $j=0,1,2$ we have
\begin{equation} \label{eq5.39}  {\mathbf q}   = \wind(l^j_{S^1}) =  \sum_k dt^{(N)}(l^{j(k)}_{S^1})
\end{equation}
In order to see this recall that
 ${\mathbf q}$ is the second winding number of the simplicial torus ribbon knot of standard type $R_1$,
 which coincides with the winding number
$\wind(l^j_{S^1})$ of $l^j_{S^1}$, considered as a continuous map $S^1 \to S^1$.\par

 Combining  Eq. \eqref{eq5.35} and Eq. \eqref{eq5.17}   with
 Eqs. \eqref{eq5.38} -- \eqref{eq5.39} and Lemma \ref{lem2}, and taking into account
Eq. \eqref{eq5.28'} above (cf. the paragraph before Observation \ref{obs3})
\begin{equation}  \label{eq5.50} \WLO^{disc}_{rig}(L)
 \sim  \sum_{\alpha \in \Lambda}
   m_{\lambda_1}(\alpha) \   \sum_{y \in I}
  \int^{\sim}_{\ct} db \   e^{ - 2\pi  i k \langle y,  b \rangle}   F_{\alpha}(b)
\end{equation}
where for $b \in \ct$ and $\alpha \in \Lambda$ we have set
 \begin{multline} \label{eq5.49}
 F_{\alpha}(b) :=  \bigl[ \bigl(  \prod_{x}
1^{\smooth}_{\ct_{reg}}(B(x)) \bigr) \bigl(  \exp(   \pi i {\mathbf q}   \langle \alpha,
 B(Z_1) + B(Z_2)  \rangle )   \\
 \times \prod_{i=0}^2  \det\nolimits^{1/2}\bigl(1_{\ck}-\exp(\ad(B(Z_i)))_{|\ck}\bigr)^{\chi(Z_i)}  \bigr) \bigr]_{| B  = b + \tfrac{1}{k}     \alpha {\mathbf p} f }
\end{multline}

\subsubsection*{d) Step 4:  Performing  $\int^{\sim} \cdots db$ and  $\sum_{y \in I}$  in Eq. \eqref{eq5.50}}

For all $y \in I$ and $\alpha \in \Lambda$ the function
$\ct \ni b \mapsto e^{ - 2\pi  i k \langle y,  b \rangle}   F_{\alpha}(b)
\in \bC$ is invariant under all translations of the form
$b \mapsto b + x$ where $x \in I = \ker(\exp_{| \ct}) \cong \bZ^{\dim(\ct)}$.
In order to prove this it is enough to show that
for all $b \in \ct$ and $x \in I$ we have
\begin{subequations}  \label{eq5.51}
\begin{align}
1^{\smooth}_{\ct_{reg}}(b + x) & = 1^{\smooth}_{\ct_{reg}}(b)\\
e^{ 2\pi  i \eps \langle \alpha,  b + x \rangle} & = e^{ 2\pi  i \eps \langle \alpha,  b \rangle}
\quad \text{  for all $\alpha \in \Lambda$, $\eps \in \bZ$} \\
\det\nolimits^{1/2}(1_{\ck} - \exp(\ad(b + x))_{| \ck})  & =  \det\nolimits^{1/2}(1_{\ck} - \exp(\ad(b))_{| \ck})\\
e^{ - 2\pi  i k \langle y,  b + x \rangle} & = e^{ - 2\pi  i k \langle y,  b\rangle} \quad \text{
 for all $y \in I$}
\end{align}
\end{subequations}
Note that  because of the assumption that
 $G$ is simply-connected we have $\Gamma = I$.
The first of the four equations above
 therefore  follows from the assumption in Sec. \ref{subsec3.6} that $1^{\smooth}_{\ct_{reg}}$
 is invariant under $\cW_{\aff}$.
    The second  equation  follows because  by definition,
$\Lambda$ is the lattice dual to $\Gamma = I$.
 The third  equation follows from  Eq. \eqref{eq_det_in_terms_of_roots} above
 by taking into account that $(-1)^{ \sum_{\alpha \in \cR_+} \langle \alpha, x \rangle}
= (-1)^{ 2  \langle \rho, x \rangle} = 1$ for $x \in  \Gamma = I$
because  $\rho = \tfrac{1}{2} \sum_{\alpha \in \cR_+}$
is an element of the weight lattice $\Lambda$.
 Finally, in order to see that the fourth equation holds, observe
that   it is enough to show that
\begin{equation} \label{eq_CartanMatrix} \langle \Check{\alpha}, \Check{\beta} \rangle \in \bZ \quad \text{ for all coroots  $\Check{\alpha}, \Check{\beta}$}
\end{equation}
According to the general theory of semi-simple Lie algebras we have
 $2\tfrac{\langle \Check{\alpha}, \Check{\beta} \rangle}{\langle \Check{\alpha}, \Check{\alpha} \rangle}
\in \bZ$. Moreover,  there are at most two different (co)roots lengths
and  the quotient between the square lengths of the long and short coroots is either 1, 2, or 3.
Since  the normalization of $\langle
\cdot,  \cdot \rangle$ was chosen such that
$\langle \Check{\alpha}, \Check{\alpha} \rangle = 2$ holds
if $\Check{\alpha}$ is a short coroot we therefore have
$\langle \Check{\alpha}, \Check{\alpha} \rangle/2 \in \{1,2,3\}$
and \eqref{eq_CartanMatrix} follows.

  \medskip

From Eqs. \eqref{eq5.49} and \eqref{eq5.51}  we conclude that
 $\ct \ni b \mapsto e^{ - 2\pi  i k \langle y,  b \rangle}   F_{\alpha}(b)
\in \bC$ is indeed $I$-periodic
and we can therefore apply  Eq. \eqref{eq2_lastrmsec3} in Remark \ref{rm_last_sec3} above  and  obtain
\begin{equation}  \label{eq5.57}
  \int^{\sim} db \
 e^{ - 2\pi  i k \langle y,  b \rangle}  F_{\alpha}(b)
 \sim  \int_{Q} db \  e^{ - 2\pi  i k \langle y,  b \rangle}  F_{\alpha}(b)
 = \int db \   e^{ - 2\pi  i k \langle y,  b \rangle} 1_{Q}(b) F_{\alpha}(b)
\end{equation}
where on the RHS  $\int_{Q} \cdots db$ and $\int \cdots db$ are now  ordinary integrals
and where  we have set
\begin{equation}  \label{eq5.54}
Q:= \{ \sum_i \lambda_i e_i \mid \lambda_i \in (0,1) \text{ for all $i \le m$}   \} \subset \ct,
\end{equation}
Here  $(e_i)_{i \le m}$ is an (arbitrary) fixed basis of $I$.\par

According to Eq. \eqref{eq5.57} we can now rewrite Eq. \eqref{eq5.50} as
\begin{align}  \label{eq5.63} %\label{eq5.50b}
 \WLO^{disc}_{rig}(L)&\sim  \sum_{\alpha \in \Lambda}
   m_{\lambda_1}(\alpha)    \sum_{y \in I}
  \int db \   e^{ - 2\pi  i k \langle y,  b \rangle}  1_{Q}(b)  F_{\alpha}(b) \nonumber \\
&  \overset{(*)}{\sim}  \sum_{\alpha \in \Lambda}
 m_{\lambda_1}(\alpha) \sum_{b \in \tfrac{1}{k} \Lambda} 1_{Q}(b)  F_{\alpha}(b)
\end{align}
where in step $(*)$ we have used,
for each  $\alpha \in \Lambda$, the Poisson summation formula
\begin{equation} \label{eq5.61}
 \sum_{y \in I}   e^{ - 2\pi  i k \langle y,  b \rangle}
  = c_{\Lambda} \sum_{x \in \tfrac{1}{k} \Lambda} \delta_x(b)
  \end{equation}
where $\delta_x$ is the  delta distribution in $x \in \ct$ and $c_{\Lambda}$ a constant depending on the lattice $\Lambda$. (Recall that the lattice $\Lambda$ is dual to $\Gamma = I$.)
Observe also that $1_{Q}  F_{\alpha}$ clearly has compact support
and that $1_{Q}  F_{\alpha}$ is smooth  because  $ \partial Q \subset \ct_{sing} = \ct \backslash \ct_{reg}$ and
$F_{\alpha}$ vanishes\footnote{note that according to the definition of  $F_{\alpha}$ and Eq. \eqref{eq5.28'}
   there is a factor $1^{(s)}_{\ct_{reg}}(b)$ appearing in $F_{\alpha}(b)$
  which vanishes on a neighborhood of $\ct_{sing}$, cf. Sec. \ref{subsec3.6}}
 on a neighborhood of  $\ct_{sing}$. \par

 Finally, note that since $s>0$ above was chosen small enough (cf. Footnote \ref{ft_sec3.6} in Sec. \ref{subsec3.6})
 we have for every  $B \in \cB(q\cK)$ of the form Eq. \eqref{eq5.37}  with $b \in \tfrac{1}{k} \Lambda$
\begin{equation} \label{eq_lem3}
 \prod_{x \in \face_0(q\cK)} 1^{\smooth}_{\ct_{reg}}(B(x)) =
  \prod_{x \in \face_0(q\cK)} 1_{\ct_{reg}}(B(x)) \overset{(+)}{=} \prod_{i=0}^2 1_{\ct_{reg}}(B(Z_i))
\end{equation}
where step $(+)$ follows from Observation  \ref{obs3}.\par
Using this we  obtain from Eq.  \eqref{eq5.63} and Eq. \eqref{eq5.49} after  the change of variable
$b \to k b =: \alpha_0$  and writing $\alpha_1$ instead of $\alpha$
(and by taking into account that    $\chi(Z_0) = 0$ and
 $\chi(Z_1) = \chi(Z_2) = 1$):
\begin{align}  \label{eq5.89_org} \WLO^{disc}_{rig}(L)  &  \sim
 \sum_{\alpha_0, \alpha_1  \in \Lambda} 1_{k Q}(\alpha_0)
  m_{\lambda_1}(\alpha_1)  \nonumber \\
&  \quad \quad \quad \times \biggl[ \biggl( \prod_{i=0}^2 1_{\ct_{reg}}(B(Z_i))\biggr)
\biggl(\prod_{i=1}^2 \det\nolimits^{1/2}\bigl(1_{\ck}-\exp(\ad(B(Z_i)))_{|\ck}\bigr)   \biggr) \nonumber \\
&  \quad \quad \quad \times    \exp(   \pi i {\mathbf q}   \langle \alpha_1,
 B(Z_1) + B(Z_2)  \rangle )  \biggr]_{|
B  =  \tfrac{1}{k} ( \alpha_0 +    \alpha_1 {\mathbf p} f)}
\end{align}

Recall from the paragraph at the beginning of Sec. \ref{subsec5.4} that so far we have been working with
the original definition of $\WLO^{disc}_{rig}(L)$ given in Sec. \ref{subsec3.9}. \par

By examining the calculations above it it becomes clear\footnote{for example, if one works with the
first modification (M1) on the list in Sec. \ref{subsec3.10}
then this point is obvious after examining the proof of Theorem 3.5
in \cite{Ha3b} where simplicial ribbons
in $q\cK \times \bZ_N$ are used. For modification (M2) this is also not difficult to see}
that if one  modifies the definition of $\WLO^{disc}_{rig}(L)$
in one of the possible ways listed in Sec. \ref{subsec3.10} then instead of Eq. \eqref{eq5.89_org} one arrives at
\begin{align}  \label{eq5.89} \WLO^{disc}_{rig}(L)  &  \sim
 \sum_{\alpha_0, \alpha_1  \in \Lambda} 1_{k Q}(\alpha_0)
  m_{\lambda_1}(\alpha_1)  \nonumber \\
&  \quad \quad \quad \times \biggl(  \prod_{i=1}^2 1_{\ct_{reg}}(B(Z_i)) \biggr)
\biggl( \prod_{i=1}^2 \det\nolimits^{1/2}\bigl(1_{\ck}-\exp(\ad(B(Z_i)))_{|\ck}\bigr)
  \biggr) \nonumber \\
&  \quad \quad \quad \times    \exp(   \pi i {\mathbf q}   \langle \alpha_1,
 B(Z_1) + B(Z_2)  \rangle )  \biggr]_{|
B  =  \tfrac{1}{k} ( \alpha_0 +    \alpha_1 {\mathbf p} f)}
\end{align}
Note that the only difference between Eq. \eqref{eq5.89} and \eqref{eq5.89_org}
is that the $1_{\ct_{reg}}(B(Z_0))$-factor appearing in Eq. \eqref{eq5.89_org}
no longer appears in Eq. \eqref{eq5.89}.

\subsubsection*{e) Step 5: Some algebraic/combinatorial arguments}

For each $\alpha_0, \alpha_1 \in \Lambda$ we define
$$\eta_{(\alpha_0,\alpha_1)}: \{1, 2\} \to \Lambda$$
 by
\begin{equation}\eta_{(\alpha_0,\alpha_1)}(i)= k B(Z_i) - \rho \quad \quad i =1,2
\end{equation}
where $B= \tfrac{1}{k}(\alpha_0 + \alpha_1 {\mathbf p} f)$.
Observe that for $\eta=\eta_{(\alpha_0,\alpha_1)}$ and $B= \tfrac{1}{k}(\alpha_0 + \alpha_1 {\mathbf p} f)$
 we have
\begin{subequations}
 \begin{equation}
 \det\nolimits^{1/2}\bigl(1_{\ck}-\exp(\ad(B(Z_i)))_{|\ck}\bigr)) =  \det\nolimits^{1/2}\bigl(1_{\ck}-\exp(\ad(\tfrac{1}{k}(\eta(i)+\rho))_{|\ck}\bigr)\bigr) \overset{(*)}{\sim} d_{\eta(i)},
 \end{equation}
 where in step $(*)$ we used Eq. \eqref{eq_det_in_terms_of_roots}  and Eq. \eqref{eq_def_th+d}. \par

In the following we will assume, without loss of generality\footnote{one can check easily that
the final result for $\WLO_{rig}^{disc}(L)$ does not depend on this assumption  }
 that $\sigma_0 \in Z_1$ and therefore
 \begin{equation} \label{eq_sigma_0_in_Z1}
  B(\sigma_0) =  B(Z_1)
 \end{equation}
 Moreover, we will assume, without loss of generality\footnote{again one can check easily that
the final result for $\WLO_{rig}^{disc}(L)$ does not depend on this assumption  }
 that the enumeration of $Z_1$ and $Z_2$ was chosen such that  we have the situation where
  in Eq.  \eqref{eq_def_f} in the proof of
 Lemma \ref{lem2_pre}  the ``$-$''-sign appears.  Then we have
 \begin{equation}
 \alpha_0 = k B(\sigma_0) = k B(Z_1) = \eta(1)+\rho,
 \end{equation}
 \begin{equation}
 \alpha_1 = \tfrac{1}{\mathbf p} (\eta(1)-\eta(2))
 \end{equation}
\begin{multline} \label{eq_prep_diagonal_argument}
 \exp( \pi i {\mathbf q}   \langle \alpha_1,B(Z_1) + B(Z_2)\rangle
 = \exp( \pi i {\mathbf q}   \langle  \tfrac{1}{\mathbf p} (\eta(1)-\eta(2)),\tfrac{1}{k}(\eta(1)+\eta(2) + 2 \rho)\rangle  \\
 = \exp\bigl(   \tfrac{\pi i}{k} \tfrac{\mathbf q}{\mathbf p}
 \langle \eta(1),\eta(1) + 2\rho \rangle -  \langle \eta(2),\eta(2) + 2\rho \rangle \bigr)
 = \theta_{\eta(1)}^{\tfrac{\mathbf q}{\mathbf p}} \theta_{\eta(2)}^{-\tfrac{\mathbf q}{\mathbf p}}
 \end{multline}
\end{subequations}
 In view of  the previous equations
 it is clear that we can rewrite Eq. \eqref{eq5.89} in the following form
\begin{multline}  \label{eq5.89b} \WLO^{disc}_{rig}(L)    \sim
 \sum_{\alpha_0, \alpha_1  \in \Lambda} \biggl[
  m_{\lambda_1}\bigl(\tfrac{1}{\mathbf p} (\eta(1)-\eta(2))\bigr)
 1_{k Q}(\eta(1)+\rho) \bigl( \prod_{i=1}^2  1_{\ct_{reg}}(\tfrac{1}{k}(\eta(i)+\rho)\bigr) \bigr)  \\
  \quad \quad \quad \times  d_{\eta(1)} d_{\eta(2)}      \theta_{\eta(1)}^{\tfrac{\mathbf q}{\mathbf p}} \theta_{\eta(2)}^{-\tfrac{\mathbf q}{\mathbf p}}  \biggr]_{|\eta  =  \eta_{( \alpha_0,\alpha_1)}}\
\end{multline}
But
\begin{multline*}
1_{k Q}(\eta(1)+\rho) \ \bigl( \prod_{i=1}^2  1_{\ct_{reg}}(\tfrac{1}{k}(\eta(i)+\rho)) \bigr)\\
= 1_{k Q}(\eta(1)+\rho) 1_{k \ct_{reg}}(\eta(1)+\rho)  1_{k \ct_{reg}}(\eta(2)+\rho)
=  1_{k (Q \cap \ct_{reg})}(\eta(1)+\rho)  1_{k \ct_{reg}}(\eta(2) + \rho)
\end{multline*}
Combining this with Eq. \eqref{eq5.89b}  we obtain
\begin{align}  \label{eq5.89c} \WLO^{disc}_{rig}(L)   & \sim
 \sum_{\eta_1, \eta_2  \in \Lambda } \biggl[
  m_{\lambda_1}\bigl(\tfrac{1}{\mathbf p} (\eta_1-\eta_2)\bigr)
1_{k (Q \cap \ct_{reg})}(\eta_1 + \rho)  1_{k \ct_{reg}}(\eta_2 + \rho)
d_{\eta_1} d_{\eta_2}      \theta_{\eta_1}^{\tfrac{\mathbf q}{\mathbf p}} \theta_{\eta_2}^{-\tfrac{\mathbf q}{\mathbf p}}  \biggr] \nonumber \\
& \sim  \sum_{\eta_1 \in k (Q \cap \ct_{reg} - \rho) \cap \Lambda,  \eta_2 \in
(k \ct_{reg} - \rho) \cap \Lambda} \biggl[
  m_{\lambda_1}\bigl(\tfrac{1}{\mathbf p} (\eta_1-\eta_2)\bigr)
d_{\eta_1} d_{\eta_2}      \theta_{\eta_1}^{\tfrac{\mathbf q}{\mathbf p}} \theta_{\eta_2}^{-\tfrac{\mathbf q}{\mathbf p}}  \biggr]
\end{align}
Let $P$ be the unique Weyl alcove which is contained in the Weyl chamber $\CW$ fixed above
and which has $0 \in \ct$ on its  boundary. Explicitly, $P$ is given by
\begin{equation} \label{eq_P_formula} P =   \{b  \in \CW \mid \langle b,\theta \rangle < 1 \}
\end{equation}
Note that the map
\begin{subequations}
\begin{equation}\cW_{\aff} \times P \ni (\tau,b) \mapsto \tau \cdot b \in \ct_{reg}
\end{equation}
is a well-defined bijection.
Moreover, there is a finite subset $W$ of $\cW_{\aff}$ such that
\begin{equation}W \times P \ni (\tau,b) \mapsto \tau \cdot b \in Q \cap \ct_{reg}
\end{equation}
\end{subequations}
is a bijection, too.
Clearly, these two bijections above  induce two other bijections
\begin{subequations}
\begin{equation}\cW_{\aff} \times (k P - \rho) \ni (\tau,b) \mapsto \tau \ast b \in k \ct_{reg} - \rho
\end{equation}
\begin{equation}W \times (k P - \rho) \ni (\tau,b) \mapsto \tau \ast b \in k (Q \cap \ct_{reg}) - \rho
\end{equation}
\end{subequations}
where $\ast$ is given as in Eq. \eqref{eq_def_ast} in Sec. \ref{subsec5.2} above.\par

Observe that for  $\eta \in \Lambda$  and $\tau \in \cW_{\aff}$ we have
\begin{subequations} \label{eq_d+th_inv}
\begin{align}
 \label{eq_theta_inv}
\theta_{\tau \ast \eta} &= \theta_{\eta}\\
\label{eq_d_inv}
 d_{\tau \ast \eta} & = (-1)^{\tau} d_{\eta}
\end{align}
\end{subequations}
[Since
$\cW_{\aff}$ is generated by $\cW$ and the translations associated to the lattice
$\Gamma$ it is enough to check Eq. \eqref{eq_d_inv} and Eq. \eqref{eq_theta_inv} for elements of $\cW$ and the aforementioned translations. If    $\tau \in \cW $  then $\tau \ast \eta = \tau \cdot \eta + \tau \cdot \rho - \rho$. On the other hand if $\tau$ is the translation by $y \in \Gamma$ we have
$\tau \ast \eta = \eta + k y$. Using this\footnote{and taking into account the relations
$\rho \in \Lambda$, $\forall x,y \in \Gamma: \langle x, y \rangle \in \bZ$, and
$\forall x \in \Gamma: \langle x, x \rangle \in 2\bZ$ (cf. Eq. \eqref{eq_CartanMatrix} above and the paragraph following Eq. \eqref{eq_CartanMatrix} } Eq. \eqref{eq_theta_inv} follows from
Eq. \eqref{eq_def_th} and  Eq. \eqref{eq_d_inv} follows from Eq. \eqref{eq_def_d}
  and Eq. \eqref{eq_def_S}].\par

In the special case where $\tau \in \cW$ we also have
$\theta_{\tau \ast \eta}^{\tfrac{\mathbf q}{\mathbf p}} = \theta_{\eta}^{\tfrac{\mathbf q}{\mathbf p}}$.
However, if ${\mathbf p} \neq \pm 1$ we cannot expect the last relation to hold
 for a general element $\tau$ of $\cW_{\aff}$ (cf. Footnote \ref{ft_warning} in Sec. \ref{subsec5.2} above). \par
 On the other hand, by taking into account\footnote{this is relevant only in the special
 case where $\tau_1$ is a translation by  $y \in \Gamma$.
 If $\tau_1 \in \cW$
 then the validity of Eq. \eqref{eq_diagonal_inv} follows from the $\cW$-invariance of $m_{\lambda_1}(\cdot)$
 and the relations mentioned above}
   Eq. \eqref{eq_prep_diagonal_argument} above and\footnote{recall that -- according to the conventions    made above we have    $m_{\lambda_1}(\cdot)  = \bar{m}_{\lambda_1}(\cdot)$} Eq. \eqref{eq_mbar_def} above
    we see that we always have
 \begin{equation} \label{eq_diagonal_inv}
  m_{\lambda_1}\bigl(\tfrac{1}{\mathbf p} (\tau_1 \ast \eta_1- \tau_1 \ast \eta_2)\bigr)
      \theta_{\tau_1 \ast \eta_1}^{\tfrac{\mathbf q}{\mathbf p}} \theta_{\tau_1 \ast \eta_2}^{-\tfrac{\mathbf q}{\mathbf p}} =
  m_{\lambda_1}\bigl(\tfrac{1}{\mathbf p} (\eta_1-  \eta_2)\bigr)
      \theta_{\eta_1}^{\tfrac{\mathbf q}{\mathbf p}} \theta_{\eta_2}^{-\tfrac{\mathbf q}{\mathbf p}}
 \end{equation}
 for all $\tau_1 \in \cW_{\aff}$ and $\eta_1,\eta_2 \in \Lambda$.
Combining this with Eq.  \eqref{eq5.89c} we finally obtain
\begin{align} \label{eq5.89d}
 & \WLO^{disc}_{rig}(L)  \nonumber \\
 &  \sim  \sum_{\eta_1, \eta_2 \in (k P - \rho) \cap \Lambda} \sum_{\tau_1 \in W, \tau_2 \in \cW_{\aff}}
\biggl[  m_{\lambda_1}\bigl(\tfrac{1}{\mathbf p} (\tau_1 \ast \eta_1- \tau_2 \ast \eta_2)\bigr)
d_{\tau_1 \ast \eta_1} d_{\tau_2 \ast \eta_2}      \theta_{\tau_1 \ast \eta_1}^{\tfrac{\mathbf q}{\mathbf p}} \theta_{\tau_2 \ast \eta_2}^{-\tfrac{\mathbf q}{\mathbf p}}  \biggr] \nonumber \\
& = \sum_{\eta_1, \eta_2 \in (k P - \rho) \cap \Lambda} \sum_{\tau_1 \in W, \tau_2 \in \cW_{\aff}}
\biggl[  m_{\lambda_1}\bigl(\tfrac{1}{\mathbf p} (\tau_1 \ast \eta_1- \tau_2 \ast \eta_2)\bigr)
(-1)^{\tau_1} (-1)^{\tau_2} d_{\eta_1} d_{\eta_2}      \theta_{\tau_1 \ast \eta_1}^{\tfrac{\mathbf q}{\mathbf p}} \theta_{\tau_2 \ast \eta_2}^{-\tfrac{\mathbf q}{\mathbf p}}  \biggr] \nonumber \\
& \overset{(*)}{=} \sum_{\eta_1, \eta_2 \in (k P - \rho) \cap \Lambda} \sum_{\tau_1 \in W, \tau \in \cW_{\aff}}
\biggl[ m_{\lambda_1}\bigl(\tfrac{1}{\mathbf p} ( \eta_1- \tau \ast \eta_2)\bigr)  (-1)^{\tau}
 d_{\eta_1} d_{\eta_2}      \theta_{\eta_1}^{\tfrac{\mathbf q}{\mathbf p}} \theta_{\tau \ast \eta_2}^{-\tfrac{\mathbf q}{\mathbf p}}  \biggr]  \nonumber \\
& \overset{(**)}{\sim}
\sum_{\eta_1, \eta_2 \in \Lambda_+^k} \sum_{\tau \in \cW_{\aff}}
\biggl[ (-1)^{\tau}  m_{\lambda_1}\bigl(\tfrac{1}{\mathbf p} ( \eta_1- \tau \ast \eta_2)\bigr)
  d_{\eta_1} d_{\eta_2}      \theta_{\eta_1}^{\tfrac{\mathbf q}{\mathbf p}}  \theta_{\tau \ast \eta_2}^{-\tfrac{\mathbf q}{\mathbf p}} \biggr] \nonumber\\
  & = \sum_{\eta_1, \eta_2 \in \Lambda_+^k} \sum_{\tau \in \cW_{\aff}}
m^{\eta_1\eta_2}_{\lambda_1,\mathbf p}(\tau)
  d_{\eta_1} d_{\eta_2}      \theta_{\eta_1}^{\tfrac{\mathbf q}{\mathbf p}}  \theta_{\tau \ast \eta_2}^{-\tfrac{\mathbf q}{\mathbf p}}
\end{align}
where in step $(*)$ we applied Eq. \eqref{eq_diagonal_inv}
and made the change of  variable $\tau_2 \to \tau := \tau_1^{-1} \tau_2$
 and where in  step $(**)$ we have used that (cf.  \eqref{eq_P_formula} above)
 \begin{align*}
 (k P - \rho)
 & =     \{k b - \rho  \mid b \in \CW \text{ and } \langle b ,\theta \rangle < 1 \}\\
& =      \{\bar{b} \in \ct \mid \bar{b} + \rho \in \CW \text{ and } \langle \bar{b} + \rho,\theta\rangle < k \}
\end{align*}
and therefore
\begin{align*}
 \Lambda  \cap (k P - \rho)
& =    \{\lambda \in \Lambda  \mid \lambda + \rho \in \CW \text{ and } \langle \lambda + \rho,\theta\rangle < k \}\\
& \overset{(*)}{=}  \{\lambda \in \Lambda \cap \overline{\CW}  \mid \langle \lambda + \rho,\theta\rangle < k \}  = \Lambda^{k}_+
\end{align*}
where step $(*)$ follows because for each $\lambda \in  \Lambda$,
    $\lambda + \rho $ is in the open Weyl chamber $\CW$ iff $\lambda$
    is in the closure $\overline{\CW}$ (cf. the last remark   in Sec. V.4 in \cite{Br_tD}).

\subsubsection*{f) Step 6: Final step}

In Steps 1--5 we showed that for $L$  as in Theorem \ref{theorem1} we have
\begin{equation} \label{eq_WLO_value_gen}
\WLO^{disc}_{rig}(L)  \sim \sum_{\eta_1, \eta_2 \in \Lambda_+^k} \sum_{\tau \in \cW_{\aff}} m^{\eta_1\eta_2}_{\lambda_1,\mathbf p}(\tau) \
d_{\eta_1} d_{\eta_2} \ \theta_{\eta_1}^{ \frac{{\mathbf q} }{{\mathbf p}}} \theta_{\tau \ast \eta_2}^{- \frac{{\mathbf q} }{{\mathbf p}}}
\end{equation}
For the empty simplicial ribbon link $L = \emptyset$ the computations in Step 1--5 simplify drastically and we obtain
\begin{equation} \label{eq_WLO_value_empty}  \WLO^{disc}_{rig}(\emptyset) \sim \frac{1}{S^2_{00}}
\end{equation}
where the multiplicative (non-zero) constant represented by $\sim$ is the same as that in Eq. \eqref{eq_WLO_value_gen}
above.
Combining Eq. \eqref{eq_WLO_value_gen} and Eq. \eqref{eq_WLO_value_empty} and recalling the meaning of $\sim$ we conclude
$$ \WLO^{disc}_{norm}(L)  = \frac{\WLO^{disc}_{rig}(L)}{\WLO^{disc}_{rig}(\emptyset)}
= S^2_{00} \sum_{\eta_1, \eta_2 \in \Lambda_+^k} \sum_{\tau \in \cW_{\aff}} m^{\eta_1\eta_2}_{\lambda_1,\mathbf p}(\tau) \
d_{\eta_1} d_{\eta_2} \ \theta_{\eta_1}^{ \frac{{\mathbf q} }{{\mathbf p}}} \theta_{\tau \ast \eta_2}^{- \frac{{\mathbf q} }{{\mathbf p}}}
$$

\begin{remark} \rm An alternative  (and more explicit) derivation of
 Eq. \eqref{eq_WLO_value_empty} can be obtained as follows.
 We can apply Eq. \eqref{eq_WLO_value_gen} to the special situation
 $L = L':=(R'_1)$   where $R'_1$ is a simplicial torus ribbon knot in $\cK \times \bZ_N$
of standard type   with winding numbers ${\mathbf p} =1$ and $\mathbf q =0$
and which is a colored with the trivial representation  $\rho_0$.
 Since  $\WLO^{disc}_{rig}(\emptyset) =  \WLO^{disc}_{rig}(L')$  we obtain
 \begin{multline}
 \WLO^{disc}_{rig}(\emptyset) =  \WLO^{disc}_{rig}(L') \sim \sum_{\eta_1, \eta_2 \in \Lambda_+^k}
 \sum_{\tau \in \cW_{\aff}} m^{\eta_1\eta_2}_{0,1}(\tau) \
d_{\eta_1} d_{\eta_2} \ \theta_{\eta_1}^{ \frac{{0} }{{1}}} \theta_{\tau \ast \eta_2}^{- \frac{{0} }{{1}}} \\
 \overset{(*)}{=}
 \sum_{\eta_1 \in \Lambda_+^k} d_{\eta_1}^2 = \sum_{\eta_1 \in \Lambda_+^k} \tfrac{S_{\eta_1 0}  S_{\eta_1 0}}{S_{00} S_{00}} = C_{0 0} \tfrac{1}{S^2_{00}} =
 \delta_{0 \bar{0}} \tfrac{1}{S^2_{00}} = \tfrac{1}{S^2_{00}}
\end{multline}
where in step $(*)$ we used  $m_{0}(\alpha) =  \delta_{0 \alpha}$
which implies that  $m^{\eta_1\eta_2}_{0,1}(\tau)$ vanishes
unless $\tau = 1$ and $\eta_2 = \eta_1$.
\end{remark}

\subsection{Proof of Theorem \ref{theorem2}}
\label{subsec5.5}

We will now sketch how a  proof of Theorem \ref{theorem2} can be obtained
by a  straightforward modification of the proof of Theorem \ref{theorem1}.

\smallskip

Recall that   $R_2$ appearing in Theorem \ref{theorem2} is a closed simplicial ribbon which is vertical
and has standard orientation.
Let $\sigma_2^j \in \face_0(q\cK)$, $j=0,1,2$, be the three points associated to $R_2$ as explained
in Definition \ref{def5.4}. We then have
\begin{equation} \label{eq_sec5.5} \Tr_{\rho_2}\bigl( \Hol^{disc}_{R_2}(\Check{A}^{\orth} + A^{\orth}_c,   B)\bigr)
= \Tr_{\rho_2}\bigl(\exp({\mathbf q_2}\sum_{j=0}^2 w(j) B(\sigma_2^j))\bigr)
\end{equation}
where $w(j)$, $j=0,1,2$, is as in Sec. \ref{subsec3.3} and where ${\mathbf q_2}$ is the second winding number of $R_2$.
Since $R_2$ has -- by assumption --
standard orientation (cf. again Definition \ref{def5.4}) we have ${\mathbf q_2} =1$. \par

Clearly, the last expression in Eq. \eqref{eq_sec5.5} above does not depend on $\Check{A}^{\orth}$ and $ A^{\orth}_c$.
Thus when proving Theorem \ref{theorem2} we can repeat the  Steps 1--3 in the proof of
Theorem \ref{theorem1} almost without modifications,
the only difference being that now
 an extra factor  $\Tr_{\rho_2}\bigl(\exp(\sum_{j=0}^2 w(j) B(\sigma_2^j))\bigr)$ appears in several equations.
  For example, we obtain again Eq. \eqref{eq5.50} at the end of Step 3 where this time the function $F_{\alpha}(b)$
contains an extra factor $\Tr_{\rho_2}\bigl(\exp(\sum_{j=0}^2 w(j) B(\sigma_2^j)))\bigr)$
inside the $[\cdots]$ brackets.
According to Observation \ref{obs3} in Step 3 for all $B$ appearing in $F_{\alpha}(b)$
we have $B(\sigma_2^0) = B(\sigma_2^1) = B(\sigma_2^2)$,  which implies that
$\Tr_{\rho_2}\bigl(\exp(\sum_{j=0}^2 w(j) B(\sigma_2^j))\bigr)
 =  \Tr_{\rho_2}\bigl(\exp( B(\sigma_2^0))\bigr)$ for the relevant $B$.
This extra factor appears later,  in Eq. \eqref{eq5.89_org}  (and in the modification Eq.
\eqref{eq5.89}) at the end of Step 4.
(In order to arrive at the modified version of Eq. \eqref{eq5.89_org}
we need to add the equation    $\Tr_{\rho_2}\bigl(\exp(b + x)\bigr) = \Tr_{\rho_2}\bigl(\exp(b)\bigr)$ for all $b \in \ct$, $x \in I$  to the list of equations in Eqs \eqref{eq5.51}).\par

As in  Sec. \ref{subsec5.4} above let us assume again
(without loss of generality)
 that the enumeration of $Z_1$ and $Z_2$ was chosen such that in Eq.  \eqref{eq_def_f} in the proof of
 Lemma \ref{lem2_pre}  we have the situation where the ``$-$''-sign appears.
Then it follows from our assumption that $R_1$ winds around $R_2$ in ``positive direction''\footnote{by
which we meant that the winding number of the projected ribbon
 $(R_1)_{\Sigma}:= \pi_{\Sigma} \circ R_1: S^1 \times [0,1] \to S^2$ around $\sigma_2^0$ is positive}
 that $\sigma_2^0 \in Z_2$.
So if we now, in Step 5  replace the variable $B$ by $\eta: \{1,2\} \to \Lambda$
given by $\eta(i) = k B(Z_i) - \rho$ then
 $B(\sigma_2^0)$ gets replaced by
$\tfrac{1}{k} (\eta(2) + \rho)$  and the term $\Tr_{\rho_2}\bigl(\exp( B(\sigma_2^0))\bigr)$
is replaced by $\Tr_{\rho_2}\bigl(\exp(\tfrac{1}{k} (\eta(2) + \rho))\bigr)$.
From  Weyl's character formula it follows that
$$\Tr_{\rho_2}\bigl(\exp(\tfrac{1}{k} (\eta(2) + \rho))\bigr) = \frac{S_{\lambda_2  \eta(2)}}{S_{0 \eta(2)}}$$
 where $\lambda_2$ is the highest weight of $\rho_2$.\par
 Accordingly, at a later stage  in Step 5 an extra factor
$ \frac{S_{\lambda_2 \eta_2}}{S_{0 \eta_2}}$ appears in several equations
where $\eta_2$ is one of the two summation variables.
Taking into account that apart from Eq. \eqref{eq_d+th_inv} we also
have $S_{\lambda_2 (\tau \ast \eta)} = (-1)^{\tau} S_{\lambda_2 \eta}$   we  finally arrive at
\begin{align} \label{eq_5.73}
 \WLO^{disc}_{rig}(L)  &  \sim
 \sum_{\eta_1, \eta_2 \in \Lambda_+^k} \sum_{\tau \in \cW_{\aff}}
  m^{ \eta_1 \eta_2}_{\lambda_1,\mathbf p}(\tau)
d_{\eta_1} d_{\eta_2} \frac{S_{\lambda_2 \eta_2}}{S_{0 \eta_2}} \theta_{\eta_1}^{ \frac{{\mathbf q} }{{\mathbf p}}} \theta_{\tau \ast \eta_2}^{- \frac{{\mathbf q} }{{\mathbf p}}} \nonumber \\
& = \sum_{\eta_1, \eta_2 \in \Lambda_+^k} \sum_{\tau \in \cW_{\aff}} m^{ \eta_1 \eta_2}_{\lambda_1,\mathbf p}(\tau)
d_{\eta_1} \tfrac{1}{S_{00}} S_{\lambda_2 \eta_2}  \theta_{\eta_1}^{ \frac{{\mathbf q} }{{\mathbf p}}} \theta_{\tau \ast \eta_2}^{- \frac{{\mathbf q} }{{\mathbf p}}}
\end{align}
 Combining this with Eq. \eqref{eq_WLO_value_empty} above (where $\sim$ represents equality up to the same multiplicative constant as in Eq. \eqref{eq_5.73}) we  arrive at Eq. \eqref{eq_theorem2}.

\section{Comparison of $\WLO^{disc}_{norm}(L)$ with the Reshetikhin-Turaev invariant $RT(L)$}
\label{sec6}

\subsection{Conjecture \ref{conj0}}
\label{subsec6.1}

For every $\cG$  as in Sec. \ref{sec2}, every   $k \in \bN$ with $k > \cg$,
and every colored ribbon link $L$ in $S^2 \times S^1$
let us denote by $RT(S^2 \times S^1,L)$, or simply by $RT(L)$,
 the corresponding Reshetikhin-Turaev invariant
associated to $U_{q}(\cG_{\bC})$ where $\cG_{\bC} := \cG \otimes_{\bR} \bC$
and $q = e^{2 \pi i/k}$ (cf. Remark \ref{rm_shift_in_level} above).\par

According to Remark \ref{rm_sec3.9c} above it is plausible\footnote{in view of Footnote \ref{ft_distinguished}
and Footnote \ref{ft_distinguished2} in Remark \ref{rm_sec3.9c} above we do not have the guarantee that
this really is the case} to expect that the values of
 $\WLO^{disc}_{norm}(L)$ computed above coincide with the
 corresponding values of  Witten's heuristic path integral expressions\footnote{\label{ft_Z=1}
 Here we used Witten's heuristic argument that $Z(S^2 \times S^1) = 1$}
  $\WLO_{norm}(L) = Z(S^2 \times S^1,L)/Z(S^2 \times S^1) = Z(S^2 \times S^1,L)$  in the special case considered.
 In view of the expected equivalence between  Witten's heuristic path integral expressions $Z(S^2 \times S^1,L)$
 and the rigorously defined Reshetikhin-Turaev invariants $RT(L)$
we arrive at the following  (rigorous) conjecture:

\begin{conjecture} \label{conj0} For  every colored simplicial
ribbon link $L$ as in Theorem \ref{theorem1} or  Theorem \ref{theorem2}
  we have
  $$\WLO^{disc}_{norm}(L)  = RT(L)$$
  where on the RHS we consider $L$ as a continuum ribbon link in $S^2 \times S^1$
  in the obvious way.
\end{conjecture}

As mentioned already in Remark \ref{rm_theorems} above in
the special case ${\mathbf p} =1$ Conjecture \ref{conj0} is true.
 Since  I have not found  a concrete  formula for $RT(L)$ where  $L$ is as in  Theorem \ref{theorem1} or  Theorem \ref{theorem2} with ${\mathbf p} > 1$
 in the standard literature  at present I cannot prove Conjecture \ref{conj0} in general
 (cf. also Sec. \ref{subsec6.3} below).
We can, however, make a ``consistency check'' and
compute -- assuming the validity of Conjecture \ref{conj0} --
the value of  $RT(S^3,\tilde{L})$ for an arbitrary colored torus knot $\tilde{L}$ in $S^3$.
We will do this in Sec. \ref{subsec6.2}  below with the help of a standard surgery argument.
It turns out that we indeed obtain the correct value for $RT(S^3,\tilde{L})$
(which is given by the  Rosso-Jones formula, cf. Eq. \eqref{eq_my_RossoJones3} below).

\subsection{Derivation of the Rosso-Jones formula}
\label{subsec6.2}

Let us now combine Theorem \ref{theorem2} with a simple surgery  argument
in order to derive\footnote{As explained in Remark \ref{rm_Goal1_vs_Goal2} below, we will do this in two different ways.  Firstly, in a rigorous way in order to obtain a consistency check of  Conjecture \ref{conj0} above
and, secondly, in a heuristic way (where we  do not need Conjecture \ref{conj0})
in order to obtain a heuristic derivation of the  Rosso-Jones formula}
the  Rosso-Jones formula for general colored  torus knots in $S^3$. \par

In the following it will be convenient to switch forth and back between
the framed link picture and the ribbon link picture.

\smallskip

Let us first recall Witten's (heuristic) surgery formula.
For our purposes it will be sufficient to consider the following special case
of Witten's surgery formula\footnote{here we use a notation which is very
similar to Witten's notation; one important difference is that we write $(C,\rho_{\alpha})$ where Witten  writes $R_{\alpha}$}
\begin{equation} \label{eq_surgery_formula0}
Z(S^3,\tilde{L}) = \sum_{\alpha \in \Lambda_+^k} K_{\alpha 0} \ Z(S^2 \times S^1, L, (C,\rho_{\alpha}))
\end{equation}
where
\begin{itemize}
\item $L$ is a colored, framed link  in $S^2 \times S^1$,

\item $\tilde{L}$ is the colored, framed link in $S^3$ obtained from $L$ by performing a surgery on a separate
     (framed) knot $C$ in $S^2 \times S^1$,

\item $\rho_{\alpha}$ is the irreducible,  finite-dimensional, complex representation of $G$ with highest weight
$\alpha \in \Lambda_+^k$ (we assume that $C$ is colored with $\rho_{\alpha}$),

\item $(K_{\mu \nu})_{\mu,\nu \in \Lambda_+^k}$ is the matrix associated to the surgery mentioned
above.
\end{itemize}

Let us now restrict to the special case where $L$ is the colored knot
$L = (T_{{\mathbf p},{\mathbf q}},\rho_{\lambda})$
 where $\lambda \in \Lambda_+^k$ and where   $T_{{\mathbf p},{\mathbf q}}$
is a  torus knot of standard type in $S^2 \times S^1$ with winding numbers ${\mathbf p} \in \bZ \backslash \{0\}$ and ${\mathbf q} \in \bZ$  (cf. Definition \ref{def5.1} and Definition \ref{def5.2}) and equipped with a ``horizontal'' framing\footnote{a special case of $T_{{\mathbf p},{\mathbf q}}$ is any simplicial torus ribbon knot of standard type (cf. Definition \ref{def5.3})
  when considered as a framed knot instead of a  ribbon knot}, i.e. a normal vector field on
$T_{{\mathbf p},{\mathbf q}}$  which is parallel to the $S^2$-component of  $S^2 \times S^1$.
Moreover,  let $C$ be a vertical loop in $S^2 \times S^1$ (equipped with a horizontal framing).
Let  $\tilde{T}_{{\mathbf p},{\mathbf q}}$ be the  framed torus  knot in
$S^3$ which is obtained from $T_{{\mathbf p},{\mathbf q}}$
by performing the  surgery on $C$ which transforms $S^2 \times S^1$ into $S^3$
and has the matrix $K=S$ associated to it (cf. p. 389 in \cite{Wi}).

\begin{remark} \label{rm_sec6.2_0} \rm Note that up to equivalence and a change of framing
every framed torus knot in $S^3$ can be obtained in this way.
\end{remark}

In the special situation described above
 formula \eqref{eq_surgery_formula0} reads
\begin{equation} \label{eq_surgery_formula1}
Z(S^3,(\tilde{T}_{{\mathbf p},{\mathbf q}},\rho_{\lambda})) = \sum_{\alpha \in \Lambda_+^k} S_{\alpha 0} \ Z(S^2 \times S^1, (T_{{\mathbf p},{\mathbf q}},\rho_{\lambda}), (C,\rho_{\alpha}))
\end{equation}

Clearly,  Eq. \eqref{eq_surgery_formula1} is not rigorous.
 We can  obtain a rigorous version of Eq. \eqref{eq_surgery_formula1}  by replacing
 the  two heuristic path integral expressions
 $Z(S^3,(\tilde{T}_{{\mathbf p},{\mathbf q}},\rho_{\lambda}))$ and $Z(S^2 \times S^1, (T_{{\mathbf p},{\mathbf q}},\rho_{\lambda}), (C,\rho_{\alpha}))$ with the corresponding Reshetikhin-Turaev invariants. Doing so we arrive at

\begin{equation} \label{eq_surgery_formula2}
RT(S^3,(\tilde{T}_{{\mathbf p},{\mathbf q}},\rho_{\lambda})) = \sum_{\alpha \in \Lambda_+^k} S_{\alpha 0} \ RT(S^2 \times S^1, (T_{{\mathbf p},{\mathbf q}},\rho_{\lambda}), (C,\rho_{\alpha}))
\end{equation}

\begin{remark} \rm \label{rm_Goal1_vs_Goal2} Even though Eq.  \eqref{eq_surgery_formula1}
is only heuristic  it is sufficient/appropriate for achieving ``Goal 2'' of Comment \ref{comm1} in the
 Introduction.  It is straightforward to rewrite the next paragraphs
 using   Eq.  \eqref{eq_surgery_formula1}
instead of Eq.  \eqref{eq_surgery_formula2} and   using\footnote{cf. the argument at the beginning of Sec. \ref{subsec6.1} where we used Witten's  heuristic equation $Z(S^2 \times S^1)=1$, cf. Footnote \ref{ft_Z=1} above.}
 $Z(S^2 \times S^1, (T_{{\mathbf p},{\mathbf q}},\rho_{\lambda}), (C,\rho_{\alpha})) = \WLO_{norm}(L) = \WLO^{disc}_{norm}(L)$ where $L$ is given as in the paragraph after the present remark. Doing so we arrive at
\begin{equation} Z(S^3,(\tilde{T}_{{\mathbf p},{\mathbf q}},\rho_{\lambda}))  = S_{00}  \sum_{\mu \in \Lambda_+}
 c^{ \mu}_{\lambda,\mathbf p} d_{\mu} \theta_{\mu}^{ \frac{{\mathbf q} }{{\mathbf p}}},
 \end{equation}
 which is the (heuristic) ``Chern-Simons path integral version'' of the
  Rosso-Jones formula, cf.  Eq. \eqref{eq_my_RossoJones3} below. \par

On the other hand, for ``Goal 1''  we  should use the rigorous formula Eq. \eqref{eq_surgery_formula2}
 in order to make the aforementioned ``consistency check'' of Conjecture \ref{conj0} above.
Since Goal 1 is our main goal we will now work with Eq. \eqref{eq_surgery_formula2}.
\end{remark}

Let us now consider the special case where $ T_{{\mathbf p},{\mathbf q}}$
``comes from''\footnote{i.e. $T_{{\mathbf p},{\mathbf q}}$ agrees
with $R_1$ when $R_1$ is considered as a framed knot instead of a ribbon knot}
 a simplicial torus ribbon knot $R_1$ in $\cK \times \bZ$ of standard type
 and  $C$ ``comes from'' a vertical closed  simplicial  ribbon $R_2$ in $\cK \times \bZ$.
 Moreover, set  $\rho_1 := \rho_{\lambda}$ and $\rho_2 := \rho_{\alpha}$ where $\alpha \in \Lambda_+^k$
 is fixed (temporarily).
 Finally, assume that for  the colored simplicial ribbon link $L := ((R_1,R_2), (\rho_1,\rho_2))$  in $\cK \times \bZ$ the  assumptions of Theorem \ref{theorem2} above are fulfilled.
If Conjecture \ref{conj0} is true we have
\begin{equation} \label{eq_RT_WLO_spec} RT(S^2 \times S^1, (T_{{\mathbf p},{\mathbf q}},\rho_{\lambda}), (C,\rho_{\alpha})) = \WLO^{disc}_{norm}(L)
\end{equation}
 Combining  Eq. \eqref{eq_RT_WLO_spec} with Theorem \ref{theorem2} (for every $\alpha \in \Lambda_+^k$)
 and using Eq. \eqref{eq_surgery_formula2} we obtain
\begin{align} \label{eq_surgery_formula}
 RT(S^3,(\tilde{T}_{{\mathbf p},{\mathbf q}},\rho_{\lambda}))
 & = \sum_{\alpha \in \Lambda_+^k} S_{\alpha 0} \biggl(S_{00} \sum_{\eta_1,\eta_2 \in \Lambda_+^k}
  \sum_{\tau \in \cW_{\aff}} m^{ \eta_1 \eta_2}_{\lambda,\mathbf p}(\tau)
d_{\eta_1} S_{\alpha \eta_2}  \theta_{\eta_1}^{ \frac{{\mathbf q} }{{\mathbf p}}} \theta_{\tau \ast \eta_2}^{- \frac{{\mathbf q} }{{\mathbf p}}}  \biggr) \nonumber\\
  & = S_{00} \sum_{\eta_1,\eta_2 \in \Lambda_+^k}  \sum_{\tau \in \cW_{\aff}}
   \biggl(\sum_{\alpha \in \Lambda_+^k} S_{\alpha 0}  S_{\alpha \eta_2} \biggr)
  m^{ \eta_1 \eta_2}_{\lambda,\mathbf p}(\tau)
d_{\eta_1} \theta_{\eta_1}^{ \frac{{\mathbf q} }{{\mathbf p}}} \theta_{\tau \ast \eta_2}^{- \frac{{\mathbf q} }{{\mathbf p}}}   \nonumber\\
   & \overset{(*)}{=} S_{00} \sum_{\eta_1,\eta_2 \in \Lambda_+^k}  \sum_{\tau \in \cW_{\aff}}  \bigl(C_{0 \eta_2} \bigr) m^{ \eta_1 \eta_2}_{\lambda,\mathbf p}(\tau)
d_{\eta_1} \theta_{\eta_1}^{ \frac{{\mathbf q} }{{\mathbf p}}} \theta_{\tau \ast \eta_2}^{- \frac{{\mathbf q} }{{\mathbf p}}}   \nonumber\\
  & \overset{(**)}{=} S_{00} \sum_{\eta_1\in \Lambda_+^k}  \sum_{\tau \in \cW_{\aff}}
  m^{ \eta_1 0}_{\lambda,\mathbf p}(\tau)
d_{\eta_1} \theta_{\eta_1}^{ \frac{{\mathbf q} }{{\mathbf p}}} \theta_{\tau \ast 0}^{- \frac{{\mathbf q} }{{\mathbf p}}}
\end{align}
Here  in Step $(*)$ we used $S^2 = C$  and the fact that $S$ is a symmetric matrix (cf. Sec. \ref{subsec5.2}),
and in Step $(**)$  we used $C_{0 \mu} = \delta_{\bar{0} \mu} = \delta_{0 \mu}$. \par

By renaming the index $\eta_1$ as $\mu$  we obtain from Eq. \eqref{eq_surgery_formula}
\begin{equation} \label{eq_RT_rewrite0}
RT(S^3,(\tilde{T}_{{\mathbf p},{\mathbf q}},\rho_{\lambda})) =
 S_{00} \sum_{\mu \in \Lambda_+^k}   \sum_{\tau \in \cW_{\aff}}
 m^{ \mu 0}_{\lambda,\mathbf p}(\tau) d_{\mu} \theta_{\mu}^{ \frac{{\mathbf q} }{{\mathbf p}}}
 \theta_{\tau \ast 0}^{- \frac{{\mathbf q} }{{\mathbf p}}}
\end{equation}

For simplicity we will now assume that $k$ is ``large'' (cf. Remark \ref{rm_general_k} below for the case of general $k > \cg$). If $k$ is ``large enough''\footnote{i.e. $k \ge k(\lambda,{\mathbf p})$ where
 $k(\lambda,{\mathbf p})$ is a constant depending only on $\lambda$ and ${\mathbf p}$}
 the sum $\sum_{\tau \in \cW_{\aff}} \cdots $ appearing in Eq. \eqref{eq_RT_rewrite0}
can be replaced by $\sum_{\tau \in \cW} \cdots $. On the other hand, for $\tau \in \cW$
we do have  $\theta_{\tau \ast 0}^{- \frac{{\mathbf q} }{{\mathbf p}}} =  \theta_{0}^{- \frac{{\mathbf q} }{{\mathbf p}}} = 1$ and so Eq. \eqref{eq_RT_rewrite0} simplifies and we obtain
\begin{align} \label{eq_RT_rewrite}
RT(S^3,(\tilde{T}_{{\mathbf p},{\mathbf q}},\rho_{\lambda})) & =
 S_{00} \sum_{\mu \in \Lambda_+^k}   \sum_{\tau \in \cW}
 m^{ \mu 0}_{\lambda,\mathbf p}(\tau) d_{\mu} \theta_{\mu}^{ \frac{{\mathbf q} }{{\mathbf p}}} \nonumber \\
& =  S_{00} \sum_{\mu \in \Lambda_+^k} \bar{M}^{ \mu 0}_{\lambda,\mathbf p} d_{\mu} \theta_{\mu}^{ \frac{{\mathbf q} }{{\mathbf p}}}
\end{align}
where we have set
\begin{equation} \label{eq_def_Mbar}
\bar{M}^{ \mu 0}_{\lambda,\mathbf p} := \sum_{\tau \in \cW}
 m^{ \mu 0}_{\lambda,\mathbf p}(\tau) = \sum_{\tau \in \cW} (-1)^{\tau} m_{\lambda}\bigl(\tfrac{1}{\mathbf p} (\mu - \tau \cdot \rho + \rho)\bigr)
\end{equation}

Observe that (for fixed $\lambda$ and ${\mathbf p}$)   the coefficients  $\bar{M}^{ \mu 0}_{\lambda,\mathbf p}$
are non-zero only for a finite number of values of $\mu$.
So if $k$ is large enough
we can replace the index set $\Lambda_+^k$ in Eq. \eqref{eq_RT_rewrite} by $\Lambda_+$ and obtain
\begin{equation}  \label{eq_my_RossoJones2}
RT(S^3,(\tilde{T}_{{\mathbf p},{\mathbf q}},\rho_{\lambda}))  = S_{00}  \sum_{\mu \in \Lambda_+}
 \bar{M}^{ \mu 0}_{\lambda,\mathbf p} d_{\mu} \theta_{\mu}^{ \frac{{\mathbf q} }{{\mathbf p}}}
\end{equation}

\begin{remark} \rm \label{rm_general_k} For simplicity, we considered here (and in the paragraph before Eq. \eqref{eq_RT_rewrite})  the case where $k$ is ``large''. However,
it is not too difficult to see that  this restriction on $k$ can be dropped, i.e.
assuming the  validity of Conjecture \ref{conj0} we can actually derive
Eq. \eqref{eq_my_RossoJones2}  for all $k > \cg$.
\end{remark}

According to Lemma 2.1 in \cite{GaVu} we have
\begin{equation} \label{eq_cond1}
\forall \mu, \lambda \in \Lambda_+:  \forall {\mathbf p} \in \bN: \quad
 \bar{M}^{ \mu 0}_{\lambda,\mathbf p} = c^{\mu}_{\lambda,{\mathbf p}}
 \end{equation}
where $(c^{\mu}_{\lambda,{\mathbf p}})_{\mu, \lambda \in \Lambda_+,{\mathbf p} \in \bN }$ are the ``plethysm coefficients'' appearing in
the Rosso-Jones formula, cf. \cite{RoJo} and Eq. (10)
in \cite{GaMo}. Accordingly, we can rewrite Eq. \eqref{eq_my_RossoJones2} as
\begin{equation}  \label{eq_my_RossoJones3}
RT(S^3,(\tilde{T}_{{\mathbf p},{\mathbf q}},\rho_{\lambda}))  = S_{00}  \sum_{\mu \in \Lambda_+}
 c^{ \mu}_{\lambda,\mathbf p} d_{\mu} \theta_{\mu}^{ \frac{{\mathbf q} }{{\mathbf p}}},
\end{equation}
which is a version of the Rosso-Jones formula.
(Note that the original Rosso-Jones formula deals with unframed torus knots
rather than framed torus knots. In Appendix \ref{appA} below we will show
that Eq. \eqref{eq_my_RossoJones3} above is indeed equivalent to the original Rosso-Jones formula).

\subsection{Reformulation of Conjecture \ref{conj0}}
\label{subsec6.3}

Note that Eq. \eqref{eq_surgery_formula2} above can be generalized  to
\begin{equation} \label{eq_surgery_formula2_gen}
RT(S^3,(\tilde{T}_{{\mathbf p},{\mathbf q}},\rho_{\lambda}), (\tilde{C},\rho_{\beta})) = \sum_{\alpha \in \Lambda_+^k} S_{\alpha \beta} \ RT(S^2 \times S^1, (T_{{\mathbf p},{\mathbf q}},\rho_{\lambda}), (C,\rho_{\alpha})) \quad \forall \beta \in \Lambda_+^k
\end{equation}
where $((\tilde{T}_{{\mathbf p},{\mathbf q}},\rho_{\lambda}), (\tilde{C},\rho_{\beta}))$ on the LHS is the
(framed, colored) two-component-link in $S^3$ obtained from the two-component-link
 $((T_{{\mathbf p},{\mathbf q}},\rho_{\lambda}), (C,\rho_{\alpha}))$ in $S^2 \times S^1$ after applying the same surgery operation as the one described in the paragraph before Remark \ref{rm_sec6.2_0}   in Sec. \ref{subsec6.2} above. \par

 By modifying the arguments and calculations after  Remark \ref{rm_Goal1_vs_Goal2}
  in Sec. \ref{subsec6.2} above in the obvious way
 we can show that Conjecture \ref{conj0} implies that  for all $\lambda, \beta \in \Lambda_+^k$ we have
\begin{equation} \label{eq_RT_rewrite0_gen}
RT(S^3,(\tilde{T}_{{\mathbf p},{\mathbf q}}, \rho_{\lambda}), (\tilde{C},\rho_{\beta})) =
 S_{00} \sum_{\mu \in \Lambda_+^k}   \sum_{\tau \in \cW_{\aff}}
 m^{ \mu \bar{\beta}}_{\lambda,\mathbf p}(\tau) d_{\mu}  \theta_{\mu}^{ \frac{{\mathbf q} }{{\mathbf p}}}  \theta_{\tau \ast \bar{\beta}}^{- \frac{{\mathbf q} }{{\mathbf p}}}
\end{equation}
The converse is also true: If Eq. \eqref{eq_RT_rewrite0_gen} holds for every (framed) colored two-component-link $((\tilde{T}_{{\mathbf p},{\mathbf q}}, \rho_{\lambda}), (\tilde{C},\rho_{\beta}))$, $\lambda, \beta \in \Lambda_+^k$  in $S^3$ obtained as above
 then Conjecture \ref{conj0} will be true.
  I expect that Eq. \eqref{eq_RT_rewrite0_gen} (and therefore Conjecture \ref{conj0})
can  be proven for arbitrary\footnote{note that according to  Sec. \ref{subsec6.2} above 
  in the special case where $\beta = 0$ (and $\lambda \in \Lambda_+^k$ is arbitrary)
  Eq. \eqref{eq_RT_rewrite0_gen} is indeed true.} $\beta \in \Lambda_+^k$ by using similar techniques as  the ones used in \cite{RoJo}.

\section{Conclusions}
\label{sec7}

In the present paper we introduced and studied -- for every  simple, simply-connected, compact Lie group $G$,
and a large class of colored torus (ribbon) knots\footnote{In the present paper we have restricted ourselves to the case of torus knots but it  should not be difficult to generalize
our main results to torus links.} $L$ in $S^2 \times S^1$ --
 a rigorous realization $\WLO^{disc}_{norm}(L)$ of the torus gauge fixed version
 of Witten's heuristic CS path integral expressions $Z(S^2\times S^1,L)$.
 Moreover, we computed the values of $\WLO^{disc}_{norm}(L)$ explicitly, cf. Theorem \ref{theorem1}. \par

As a by-product  we  obtained an elementary, heuristic
derivation\footnote{\label{ft_sec7_1} Our derivation is ``almost'' a pure path
integral derivation. Essentially all our arguments
 are based on (the rigorous realization of) the  Chern-Simons path integral
 in the torus gauge introduced in Sec. \ref{sec3}.
 The only two exceptions are the arguments involving
 Witten's heuristic surgery formula and
 Witten's formula $Z(S^2 \times S^1)=1$ mentioned in Remark \ref{rm_Goal1_vs_Goal2}  above.
 (Both formulas  were  derived by Witten  using arguments from Conformal Field Theory.)}
of the original Rosso-Jones formula for arbitrary colored torus knots $\tilde{L}$ in $S^3$
(and arbitrary simple complex Lie algebras $\cG_{\bC}$).
This means that we have achieved ``Goal 2'' of Comment \ref{comm1} in the Introduction. \par

Apart from achieving ``Goal 2''
we have also made progress towards  achieving ``Goal 1'' of Comment \ref{comm1}.
The rigorous computation in Sec. \ref{subsec6.2}
provides  strong evidence in  favor of Conjecture \ref{conj0} above, i.e. the conjecture that the explicit values
 of  $\WLO^{disc}_{norm}(L)$ obtained in Theorem \ref{theorem1} and Theorem \ref{theorem2} indeed  coincide
with the values of the corresponding Reshetikhin-Turaev invariants $RT(L)$.
If this is indeed the case
Theorem \ref{theorem1} can be considered as a  step forward in the simplicial program for Chern-Simons theory, cf.  Sec. 3 of \cite{Ha7a} and cf. also Remark  \ref{rm_sec3.9b}  in Sec. \ref{sec3} of the present paper.

 \renewcommand{\thesection}{\Alph{section}}
\setcounter{section}{0}

\section{Appendix: The original Rosso-Jones formula for unframed torus knots in $S^3$}
 \label{appA}

In this appendix we will recall the original Rosso-Jones formula (which deals with unframed torus knots in $S^3$
rather than framed torus knots) and show that it is equivalent
to  Eq. \eqref{eq_my_RossoJones3}  in Sec. \ref{subsec6.2} above.

\smallskip

Recall that above we set $\cG_{\bC} = \cG \otimes_{\bR} \bC$  and $q = e^{ 2 \pi i/k}$
and denoted by $RT(M,\cdot)$ (for $M=S^3$ and $M=S^2 \times S^1$) the
Reshetikhin-Turaev invariant associated to $U_{q}(\cG_{\bC})$.
Let us write $RT(\cdot)$ instead of $RT(S^3,\cdot)$. \par

We will now compare $RT(\cdot)$ with  $QI(\cdot)$
where $QI(\cdot)$ is the  $U_{q}(\cG_{\bC})$-analogue\footnote{Note that
in \cite{Tu0,RoJo} the letter $q$ refers to a complex variable or, equivalently,
a generic element of $\bC$. We now replace this variable by the root of unity $e^{ 2 \pi i/k}$.
Doing so we obtain a complex valued  topological invariant  of all colored links $L$ in $S^3$
 whose colors $\rho_i$  fulfill the condition $\lambda_i \in \Lambda_+^k$ where $\lambda_i$
 is the highest weight of $\rho_i$.
Recall that if  $\lambda_i \notin \Lambda_+^k$
 then $RT(L)$  is not defined and  $QI(L)$ need not be  defined either since division by $0$ may occur.}
  of the topological invariant for colored links in $S^3$
 which appeared in \cite{Tu0,RoJo}. \par

The two invariants $RT(\cdot)$ and ${QI}(\cdot)$ are very closely related
but there are some important  differences:
\begin{itemize}

\item ${QI}(\cdot)$ is an invariant of (unframed) colored links.
      More precisely, it is an ambient isotopy invariant (i.e. invariant under\footnote{here we consider every link $L \subset S^3$ as a link in $\bR^3$ and project it down to a suitable fixed plane $P$} Reidemeister I, II and III moves)

\item ${QI}(\cdot)$  is normalized such that ${QI}(U_{\lambda}) = 1$
where $U_{\lambda}$ is the unknot colored with (the irreducible complex representation $\rho$ with highest weight) $\lambda \in \Lambda_+^k$.
\end{itemize}

By contrast we have
\begin{itemize}

\item $RT(\cdot)$ is an invariant of framed, colored links.
More precisely, it is a regular isotopy invariant (i.e. invariant  only under Reidemeister  II and III moves)

\item For every framed knot $L$ with color $\rho_{\lambda}$, $\lambda \in \Lambda_+^k$, the value of $RT(L)$ changes by a factor
 ${\theta}_{\lambda}^{\pm 1}$ when we perform a  Reidemeister I move on $L$.

\item $RT(\cdot)$ is normalized such that $RT(U_{\lambda}) = {S}_{\lambda 0}$  where
$U_{\lambda}$ is the 0-framed  unknot colored with $\lambda$.

 \end{itemize}

Taking this into account one can deduce
 the following relation between $RT(\cdot) = RT(S^3,\cdot)$ and ${QI}(\cdot)$
\begin{equation} \label{eq_QI_RT}
{QI}(L^{0}) = \frac{1}{{S}_{\lambda 0}}  {\theta}_{\lambda}^{- \writhe(D(L))} RT(S^3,L)
\end{equation}
for every framed knot $L$ with color $\rho_{\lambda}$, $\lambda \in \Lambda_+^k$, where $L^0$ is the unframed, colored knot obtained from $L$ by forgetting the framing and where $D(L)$ is any
``admissible''\footnote{Here by ``admissible'' I mean a knot diagram which is obtained in the following way
(here we use the ribbon picture, i.e. we consider the framed knot $L$ as a ribbon in the obvious way):
 We press the ribbon $L$ flat onto a fixed plane $P$.
  After that the ribbon width is sent to zero}    knot diagram of $L$.
  [Note that the writhe is also a regular isotopy invariant and the effect of a Reidemeister I move
on the exponential on the RHS of Eq. \eqref{eq_QI_RT} cancels out the effect of the move on the factor $RT(L)$.
Accordingly, the RHS of Eq \eqref{eq_QI_RT} will be invariant under Reidemeister I-III moves.] \par

Let us now consider the special case  $L = (\tilde{T}_{{\mathbf p}, {\mathbf q}},\rho_{\lambda})$
where $\tilde{T}_{{\mathbf p}, {\mathbf q}}$ is the framed torus knot in $S^3$
that appeared in Sec. \ref{subsec6.2} above.
It can be shown that for a torus knot $\tilde{T}_{{\mathbf p}, {\mathbf q}}$
obtained by surgery from a  torus knot $T_{{\mathbf p}, {\mathbf q}}$ in $S^2 \times S^1$ of standard type
with horizontal framing we have
$$\writhe(D(\tilde{T}_{{\mathbf p}, {\mathbf q}})) = {\mathbf p}{\mathbf q}$$
for every admissible knot diagram $D(\tilde{T}_{{\mathbf p}, {\mathbf q}})$ of $\tilde{T}_{{\mathbf p}, {\mathbf q}}$
(in the sense above).
 Accordingly, Eq. \eqref{eq_QI_RT} specializes to
\begin{equation} \label{eq_A.4} {QI}((\tilde{T}^0_{{\mathbf p}, {\mathbf q}},\rho_{\lambda}))
 = \frac{1}{{S}_{\lambda0}}  {\theta}_{\lambda}^{-{\mathbf p}{\mathbf q}}
RT(S^3,( \tilde{T}_{{\mathbf p}, {\mathbf q}},\rho_{\lambda}))
\end{equation}

In view of  Eq. \eqref{eq_A.4} it is now clear that Eq. \eqref{eq_my_RossoJones3} in Sec. \ref{subsec6.2} above
is equivalent to
 \begin{equation} \label{eq_RossoJones0}
QI((\tilde{T}^0_{{\mathbf p}, {\mathbf q}},\rho_{\lambda})) = \frac{1}{d_{\lambda}}  \theta_{\lambda}^{-{\mathbf p}{\mathbf q}}  \sum_{\mu \in \Lambda_+}  c^{\mu}_{\lambda,{\mathbf p}} d_{\mu} \theta_{\mu}^{ \frac{{\mathbf q} }{{\mathbf p}}},
 \end{equation}
 which is  the  original Rosso-Jones formula, cf.  Eq. (10)
in \cite{GaMo}.

\end{document}